\documentclass[letterpaper, onecolumn, 10pt, accepted=2024-08-14]{quantumarticle}
\pdfoutput=1

\usepackage[backend=bibtex,sorting=none,firstinits=true,style=phys,maxbibnames=10,biblabel=brackets,chaptertitle=false,pageranges=false,doi=false]{biblatex}
\usepackage{dsfont}
\usepackage[utf8]{inputenc}
\usepackage[T1]{fontenc}
\usepackage{amsthm, amsmath, amssymb}

\usepackage{amsthm}
\usepackage{mathtools, yhmath}
\usepackage{amsfonts}
\usepackage{amssymb}
\usepackage{graphicx}
\usepackage{xcolor}
\usepackage{accents}
\usepackage{bm}
\usepackage{float}
\usepackage[
colorlinks=true,
citecolor=blue,
linktocpage=true
]{hyperref}
\usepackage[normalem]{ulem}
\usepackage[shortlabels]{enumitem}
\usepackage{cleveref}
\usepackage{algorithm2e}
\RestyleAlgo{ruled}

\makeatletter 
\newsavebox{\@brx}
\newcommand{\llangle}[1][]{\savebox{\@brx}{\(\m@th{#1\langle}\)}%
  \mathopen{\copy\@brx\kern-0.5\wd\@brx\usebox{\@brx}}}
\newcommand{\rrangle}[1][]{\savebox{\@brx}{\(\m@th{#1\rangle}\)}%
  \mathclose{\copy\@brx\kern-0.5\wd\@brx\usebox{\@brx}}}
\makeatother

\newlength{\dhatheight} 

\newtheorem{theorem}{Theorem}[section]

\newtheorem{lemma}[theorem]{Lemma}

\theoremstyle{definition}

\long\def\/*#1*/{}

\usepackage{suffix}
\newcommand{\az}[1]{{\textcolor{teal}{[#1]}}}
\WithSuffix\newcommand\az*[1]{{\textcolor{teal}{#1}}}

\newcommand{\Ord}{\mathcal{O}}
\newcommand{\Ot}{\widetilde{\mathcal{O}}}

\newcommand{\R}{\mathbb{R}}

\newcommand{\Z}{\mathbb{Z}}

\newcommand{\I}{\mathbb{I}}

\newcommand{\Spin}{\mathrm{Spin}}
\newcommand{\Pin}{\mathrm{Pin}}
\newcommand{\Orth}{\mathrm{O}}
\newcommand{\SO}{\mathrm{SO}}
\newcommand{\SU}{\mathrm{SU}}
\newcommand{\U}{\mathrm{U}}

\newcommand{\Cl}{\mathrm{Cl}}

\newcommand{\CNOT}{\mathrm{CNOT}}
\DeclareMathOperator{\sign}{sign}
\DeclareMathOperator*{\argmin}{arg\,min}
\DeclareMathOperator*{\argmax}{arg\,max}
\DeclareMathOperator{\spn}{span}
\DeclareMathOperator*{\E}{\mathbb{E}}
\DeclareMathOperator{\tr}{tr}
\newcommand{\op}[2]{\ket{#1}\!\bra{#2}}
\newcommand{\ip}[2]{\langle #1 | #2 \rangle}
\newcommand{\ev}[2]{\langle #2 | #1 | #2 \rangle}
\newcommand{\ket}[1]{| #1 \rangle}
\newcommand{\bra}[1]{\langle #1 |}

\renewcommand{\i}{\mathrm{i}}
\newcommand{\T}{\mathsf{T}}
\renewcommand{\l}[1]{\mathopen{}\left#1}
\renewcommand{\r}[1]{\right#1\mathclose{}}

\DeclareMathOperator{\conv}{conv}

\begin{document}


\title{Expanding the reach of quantum optimization with fermionic embeddings}

\author{Andrew Zhao}
\email{azhao@unm.edu}
\affiliation{Google Quantum AI, San Francisco, CA 94105, USA}
\affiliation{Center for Quantum Information and Control, Department of Physics and Astronomy, University of New Mexico, Albuquerque, NM 87106, USA}

\author{Nicholas C. Rubin}
\email{nickrubin@google.com}
\affiliation{Google Quantum AI, San Francisco, CA 94105, USA}


\begin{abstract}
    Quadratic programming over orthogonal matrices encompasses a broad class of hard optimization problems that do not have an efficient quantum representation. Such problems are instances of the little noncommutative Grothendieck problem (LNCG), a generalization of binary quadratic programs to continuous, noncommutative variables. In this work, we establish a natural embedding for this class of LNCG problems onto a fermionic Hamiltonian, thereby enabling the study of this classical problem with the tools of quantum information. This embedding is accomplished by a new representation of orthogonal matrices as fermionic quantum states, which we achieve through the well-known double covering of the orthogonal group. Correspondingly, the embedded LNCG Hamiltonian is a two-body fermion model. Determining extremal states of this Hamiltonian provides an outer approximation to the original problem, a quantum analogue to classical semidefinite relaxations. In particular, when optimizing over the \emph{special} orthogonal group our quantum relaxation obeys additional, powerful constraints based on the convex hull of rotation matrices. The classical size of this convex-hull representation is exponential in matrix dimension, whereas our quantum representation requires only a linear number of qubits. Finally, to project the relaxed solution back into the feasible space, we propose rounding procedures which return orthogonal matrices from appropriate measurements of the quantum state. Through numerical experiments we provide evidence that this rounded quantum relaxation can produce high-quality approximations.
\end{abstract}

\maketitle

\tableofcontents

\section{Introduction}
Finding computational tasks where a quantum computer could have a large speedup is a primary driver for the field of quantum algorithm development. While some examples of quantum advantage are known, such as quantum simulation~\cite{feynman1982,lloyd1996universal}, prime number factoring~\cite{shorfactoring}, and unstructured search~\cite{Grover1996}, generally speaking computational advantages for industrially relevant calculations are scarce. Specifically in the field of optimization, which has attracted a large amount of attention from quantum algorithms researchers due to the ubiquity and relevance of the computational problems, substantial quantum speedups, even on model problems, are difficult to identify.  This difficulty is in part because it is not obvious \emph{a priori} how the unique features of quantum mechanics---e.g., entanglement, unitarity, and interference---can be leveraged towards a computational advantage~\cite{Grover1996,PRXQuantum.2.010103,PRXQuantum.2.030312}.  

In this work we take steps toward understanding how to apply quantum computers to optimization problems by demonstrating that the class of optimization problems involving rotation matrices as decision variables has a natural quantum formulation and efficient embedding. Examples of such problems include the joint alignment of points in Euclidean space by isometries, which has applications within the contexts of structural biology via cryogenic electron microscopy (cryo-EM)~\cite{shkolnisky2012viewing,singer2018mathematics} and NMR spectroscopy~\cite{cucuringu2012eigenvector}, computer vision~\cite{arie2012global,ozyecsil2017survey}, robotics~\cite{rosen2019se,lajoie2019modeling}, and sensor network localization~\cite{cucuringu2012sensor}. The central difficulty in solving these problems is twofold:~first, the set of orthogonal transformations $\Orth(n)$ is nonconvex, making the optimization landscape challenging to navigate in general. Second, the objectives of these problems are quadratic in the decision variables, making them examples of quadratic programming under orthogonality constraints~\cite{nemirovski2007sums}. In this paper we specifically focus on the problem considered by Bandeira \emph{et al.}~\cite{bandeira2016approximating}, which is a special case of the real little noncommutative Grothendieck (LNCG) problem~\cite{briet2017tight}. While significant progress has been made in classical algorithms development for finding approximate solutions, for example by semidefinite relaxations~\cite{povh2010semidefinite,wang2013exact,naor2014efficient,saunderson2014semidefinite,bandeira2016approximating}, guaranteeing high-quality solutions remains difficult in general. This paper therefore provides a quantum formulation of the optimization problem, as a first step in exploring the potential use of a quantum computer to obtain more accurate solutions.

The difficulty of the LNCG problem becomes even more pronounced when restricting the decision variables to the group of rotation matrices $\SO(n)$~\cite{bandeira2017estimation,pumir2021generalized}. One promising approach to resolving this issue is through the convex relaxation of the problem, studied by Saunderson \emph{et al.}~\cite{saunderson2015semidefinite,saunderson2014semidefinite}. They identified that the convex hull of rotation matrices, $\conv\SO(n)$, is precisely the feasible region of a semidefinite program (SDP)~\cite{saunderson2015semidefinite}. Therefore, standard semidefinite relaxations of the quadratic optimization problem can be straightforwardly augmented with this convex-hull description as an additional constraint~\cite{saunderson2014semidefinite}. They prove that when the problem is defined over particular types of graphs, this enhanced SDP is exact, and for more general instances of the problem they numerically demonstrate that it yields significantly higher-quality approximations than the basic SDP. The use of this convex hull has since been explored in related optimization contexts~\cite{matni2014convex,rosen2015convex,saunderson2016convex}. Notably however, the semidefinite description of $\conv\SO(n)$ is exponentially large in $n$. Roughly speaking, this reflects the complexity of linearizing a nonlinear determinant constraint. One such representation is the so-called positive-semidefinite (PSD) lift of $\conv\SO(n)$, which is defined through linear functionals on the trace-1, PSD matrices of size $2^{n-1} \times 2^{n-1}$.

One may immediately recognize this description as the set of density operators on $n - 1$ qubits. In this paper we investigate this statement in detail and make a number of connections between the optimization of orthogonal/rotation matrices and the optimization of quantum states, namely fermionic states in second quantization. The upshot is that these connections provide us with a relaxation of the quadratic program into a quantum Hamiltonian problem. Although this relaxation admits solutions (quantum states) which lie outside the feasible space of the original problem, we show that it retains much of the important orthogonal-group structure due to this natural embedding. The notion of quantum relaxations have been previously considered in the context of combinatorial optimization (such as the Max-Cut problem), wherein quantum rounding protocols were proposed to return binary decision variables from the relaxed quantum state~\cite{fuller2021approximate}. In a similar spirit, in this paper we consider rounding protocols which return orthogonal/rotation matrices from our quantum relaxation.

Within the broader context of quantum information theory, our work here also provides an alternative perspective to relaxations of quantum Hamiltonian problems. There is a growing interest in classical methods for approximating quantum many-body problems based on SDP relaxations~\cite{brandao2013product, bravyi2019approximation, gharibian_et_al:LIPIcs:2019:11246,anshu_et_al:LIPIcs:2020:12066, parekh_et_al:LIPIcs.ESA.2021.74,parekh2021application,hastings2022optimizing,parekh2022optimal,hastings2022perturbation,king2022improved}. In that context, rounding procedures are more difficult to formulate because the space of quantum states is exponentially large. For instance, the algorithm may only round to a subset of quantum states with efficient classical descriptions such as product states~\cite{brandao2013product, bravyi2019approximation, gharibian_et_al:LIPIcs:2019:11246, parekh_et_al:LIPIcs.ESA.2021.74,parekh2022optimal} or low-entanglement states~\cite{anshu_et_al:LIPIcs:2020:12066,parekh2021application}, effectively restricting the approximation from representing the true ground state. Nonetheless, these algorithms can still obtain meaningful approximation ratios of the optimal energy, indicating that such states can at least capture some qualitative properties of the generically entangled ground state.

Our quantum relaxation can be viewed as working in the opposite direction:~we construct a many-body Hamiltonian where the optimal solution to the underlying classical quadratic program is essentially a product state. Therefore, we propose preparing an approximation to the ground state of the Hamiltonian,\footnote{While the physical problem typically considers the ground-state problem, this paper takes the convention of maximizing objectives.} which is then rounded to the nearest product state corresponding to the original classical solution space. This is not unlike quantum approaches to binary optimization such as quantum annealing or the quantum approximate optimization algorithm~\cite{PhysRevLett.101.130504,farhi2014quantum,PhysRevA.101.012320,hauke2020perspectives,PRXQuantum.2.030312}, which explore a state space outside the classical feasible region before projectively measuring, or rounding, the quantum state to binary decision variables. We furthermore provide numerical evidence that the physical qualitative similarity between optimal product and entangled states may translate into quantitative accuracy for the classical optimization problem, in a context beyond discrete combinatorial optimization.

Finally, we remark that Grothendieck-type problems and inequalities have a considerable historical connection to quantum theory. Tsirelson~\cite{tsirelson1987quantum} employed Clifford algebras to reformulate the commutative Grothendieck inequality into a statement about classical XOR games with entanglement. Regev and Vidick~\cite{regev2015quantum} later introduced the notion of quantum XOR games, which they studied through the generalization of such ideas to noncommutative Grothendieck inequalities. The mathematical work of Haagerup and Itoh~\cite{haagerup1995grothendieck} studied Grothendieck-type inequalities as the norms of operators on $C^*$-algebras;~their analysis makes prominent use of canonical anticommutation relation algebras over fermionic Fock spaces. Quadratic programming with orthogonality constraints has also been applied for classical approximation algorithms for quantum many-body problems, for instance by Bravyi \emph{et al.}~\cite{bravyi2019approximation}. Recasting noncommutative Grothendieck problems into a quantum Hamiltonian problem may therefore provide new insights into these connections.

The rest of this paper is organized as follows:~Section~\ref{sec:problem_statement} provides a formal description of the optimization problem that we study in this paper and reviews known complexity results of related problems. In Section~\ref{sec:lncg_applications} we describe two well-known applications of the problem:~the group synchronization problem and the generalized orthogonal Procrustes problem. Section~\ref{sec:summary_results} provides a summary of our quantum relaxation which embeds the optimization problem into a Hamiltonian, and two accompanying rounding protocols. In Section~\ref{sec:quantum_formalism} we derive an embedding of orthogonal matrices into quantum states via the Pin and Spin groups. We elaborate on the connection to fermionic theories and provide a quantum perspective on the convex hull of the orthogonal groups. From this embedding, Section~\ref{sec:quantum_relax_section} then establishes the quantum Hamiltonian relaxation of the quadratic optimization problem. Section~\ref{sec:rounding} describes both classical and quantum rounding protocols for relaxations of the problem. Notably, for the classical SDP we derive an approximation ratio for $\SO(n)$ rounding. Finally, in Section~\ref{sec:numerical_experiments} we demonstrate numerical experiments on random instances of the group synchronization problem for $\SO(3)$ on three-regular graphs and report the performance of various classical and quantum rounding protocols. For our simulations of the quantum relaxation, we consider two classes of quantum states:~maximal eigenstates of the Hamiltonian and quasi-adiabatically evolved states. We close in Section~\ref{sec:discussion} with a discussion on future lines of research.

\section{Problem statement}\label{sec:problem_statement}
In this paper we consider the class of little noncommutative Grothendieck (LNCG) problems over the orthogonal group, as studied previously by Bandeira \emph{et al.}~\cite{bandeira2016approximating}.\footnote{The authors also consider the complex-valued problem over the unitary group, which is outside the scope of this present paper.} Let $(V, E)$ be an undirected graph with $m = |V|$ vertices and edge set $E$. For integer $n \geq 1$, let $C \in \R^{mn \times mn}$ be a symmetric matrix, which for notation we partition into $n \times n$ blocks as
\begin{equation}
    C = \begin{bmatrix}
    C_{11} & \cdots & C_{1m}\\
    \vdots & \ddots & \vdots\\
    C_{m1} & \cdots & C_{mm}
    \end{bmatrix}.
\end{equation}
The quadratic program we wish to solve is of the form
\begin{equation}\label{eq:LNCG}
    \max_{R_1, \ldots, R_m \in G} \sum_{(u, v) \in E} \langle C_{uv}, R_u R_v^\T \rangle,
\end{equation}
where $G$ is either the orthogonal group
\begin{equation}
    \Orth(n) \coloneqq \{ R \in \R^{n \times n} \mid R^\T R = \I_n \},
\end{equation}
or the special orthogonal group
\begin{equation}
    \SO(n) \coloneqq \{ R \in \Orth(n) \mid \det R = 1 \}
\end{equation}
on $\R^n$. Here, $\langle A, B \rangle = \tr(A^\T B)$ denotes the Frobenius inner product on the space of real matrices and $\I_n$ is the $n \times n$ identity matrix. Note that when $G = \Orth(1) = \{\pm 1\}$, Problem~\eqref{eq:LNCG} reduces to combinatorial optimization of the form
\begin{equation}\label{eq:LCG}
    \max_{x_1, \ldots, x_m \in \{\pm 1\}} \sum_{(u, v) \in E} C_{uv} x_u x_v,
\end{equation}
where now $C \in \R^{m \times m}$. This is sometimes referred to as the commutative instance of the little Grothendieck problem. Problem~\eqref{eq:LNCG} can therefore be viewed as a natural generalization of quadratic binary optimization to 
the noncommutative matrix setting.

We now comment on the known hardness results of these optimization problems. The commutative problem~\eqref{eq:LCG} is already NP-hard in general, as can be seen by the fact that the Max-Cut problem can be expressed in this form. In particular, Khot \emph{et al.}~\cite{khot2007optimal} proved that, assuming the Unique Games conjecture, it is NP-hard to approximate the optimal Max-Cut solution to better than a fraction of $(2/\pi) \min_{\theta \in [0, \pi]} \frac{\theta}{1 - \cos\theta} \approx 0.878$. This value coincides with the approximation ratio achieved by the celebrated Goemans--Williamson (GW) algorithm for rounding the semidefinite relaxation of the problem~\cite{goemans1995improved}. More generally, let $(V, E)$ be fully connected and $C \succeq 0$ arbitrary. In this setting, Nesterov~\cite{nesterov1998semidefinite} showed that GW rounding guarantees an approximation ratio of $2/\pi \approx 0.636$, which Alon and Naor~\cite{alon2004approximating} showed matches the integrality gap of the semidefinite program. Khot and Naor~\cite{khot2009approximate} later demonstrated that this approximation ratio is also Unique-Games-hard to exceed, and finally Bri\"{e}t \emph{et al.}~\cite{briet2017tight} strengthened this result to be unconditionally NP-hard.

For the noncommutative problem~\eqref{eq:LNCG} that we are interested in, less is known about its hardness of approximability. However, it is a subclass of more general optimization problems for which some results are known. The most general instance is the ``big'' noncommutative Grothendieck problem, for which Naor \emph{et al.}~\cite{naor2014efficient} provided a rounding procedure of its semidefinite relaxation. Their algorithm achieves an approximation ratio of at least $1/2\sqrt{2} \approx 0.353$ in the real-valued setting, and $1/2$ in the complex-valued setting (wherein optimization is over the unitary group instead of the orthogonal group). This $1/2$ result was later shown to be tight by Bri\"{e}t \emph{et al.}~\cite{briet2017tight} for both the real- and complex-valued settings;~specifically, they showed that this value is the NP-hardness threshold of a special case of the problem, which is the \emph{little} noncommutative Grothendieck problem.\footnote{See Section 6 of Bri\"{e}t \emph{et al.}~\cite{briet2017tight} for the precise relation between the big and little NCG.} However, the threshold for Problem~\eqref{eq:LNCG}, which is an special case of the LNCG, is not known. Algorithmically, Bandeira \emph{et al.}~\cite{bandeira2016approximating} demonstrated constant approximation ratios for Problem~\eqref{eq:LNCG} when $C \succeq 0$ and $G = \Orth(n)$ or $\U(n)$ via an $(n \times n)$-dimensional generalization of GW rounding, along with matching integrality gaps. These approximation ratios exceed $1/2$, indicating that this subclass is less difficult than the more general formulation of LNCG. However, although the optimization of rotation matrices is of central importance to many applications, we are unaware of any approximation ratio guarantees for the $G = \SO(n)$ setting.

\section{Applications of the LNCG problem}\label{sec:lncg_applications}
Before describing our quantum relaxation, here we motivate the practical interest in Problem~\eqref{eq:LNCG} by briefly discussing some applications. Throughout, let $G = \Orth(n)$ or $\SO(n)$ and $(V, E)$ be a graph as before.

\subsection{Group synchronization}\label{sec:group_sync}

The group synchronization problem over orthogonal transformations has applications in a variety of disciplines, including structural biology, robotics, and wireless networking. For example, in structural biology the problem appears as part of the cryogenic electron microscopy (cryo-EM) technique. There, one uses electron microscopy on cryogenically frozen samples of a molecular structure to obtain a collection of noisy images of the structure. The images typically have low signal-to-noise ratio, and they feature the structure in different, unknown orientations. One approach to solving the group synchronization problem yields best-fit estimates for these orientations via least-squares minimization~\cite{boumal2013robust}, from which one can produce a model of the desired 3D structure.\footnote{Note that other loss functions are also considered in the literature, which may not necessarily have a reformulation as Problem~\eqref{eq:LNCG}.} See Ref.~\cite{singer2018mathematics} for a further overview, and Ref.~\cite{ozyecsil2017survey} for a survey of other applications of group synchronization.

The formal problem description is as follows. To each vertex $v \in [m] \coloneqq \{ 1, \ldots, m \}$ there is an unknown $g_v \in G$. An interaction between each pair of vertices connected by an edge $(u, v) \in E$ is modeled as $g_{uv} = g_u g_v^\T$. However, measurements of the interactions are typically corrupted by some form of noise. For instance, one may consider an additive noise model of the form $C_{uv} = g_{uv} + \sigma W_{uv}$, where $\sigma \geq 0$ characterizes the strength of the noise and each $W_{uv} \in \R^{n \times n}$ has independently, normally distributed entries. We would like to recover each $g_v$ given only access to the matrices $C_{uv}$. Therefore, as a proxy to the recovery problem one may cast the solution as the least-squares minimizer
\begin{equation}\label{eq:group_sync_argmin}
    \min_{\mathbf{R} \in G^m} \sum_{(u, v) \in E} \| C_{uv} - R_u R_v^\T \|_F^2,
\end{equation}
where $\| A \|_F = \sqrt{\langle A, A \rangle}$ is the Frobenius norm and we employ the notation $\mathbf{R} \equiv (R_1, \ldots, R_m)$. It is straightforward to see that the minimzer of this problem is equivalent to the maximizer of
\begin{equation}
    \max_{\mathbf{R} \in G^m} \sum_{(u, v) \in E} \langle C_{uv}, R_u R_v^\T \rangle,
\end{equation}
which is precisely in the form of Problem~\eqref{eq:LNCG}.

\subsection{Generalized orthogonal Procrustes problem}

Procrustes analysis has applications in fields such as shape and image recognition, as well as sensory analysis and market research on $n$-dimensional data. In this problem, one has a collection of point clouds, each representing for instance the important features of an image. One wishes to determine how similar these images are to each other collectively. This is achieved by simultaneously fitting each pair of point clouds to each other, allowing for arbitrary orthogonal transformations on each cloud to best align the individual points. We refer the reader to Ref.~\cite{gower2004procrustes} for a comprehensive review.

Consider $m$ sets of $K$ points in $\R^n$, $S_v = \{ x_{v,1}, \ldots, x_{v,K} \} \subset \R^n$ for each $v \in [m]$. We wish to find an orthogonal transformation $R_v \in G$ for each $S_v$ that best aligns all sets of points simultaneously. That is, for each $k \in [K]$ and $u, v \in [m]$ we wish to minimize the Euclidean distance $\| R_u^\T x_{u,k} - R_v^\T x_{v,k} \|_2$. Taking least-squares minimization as our objective, we seek to solve
\begin{equation}\label{eq:procrustes}
    \min_{\mathbf{R} \in G^m} \sum_{u, v \in [m]} \sum_{k \in [K]} \| R_u^\T x_{u,k} - R_v^\T x_{v,k} \|_2^2.
\end{equation}
From the relation between the vector 2-norm and matrix Frobenius norm, Eq.~\eqref{eq:procrustes} can be formulated as
\begin{equation}
    \max_{\mathbf{R} \in G^m} \sum_{u, v \in [m]} \langle C_{uv}, R_u R_v^\T \rangle,
\end{equation}
where each $C_{uv} \in \R^{n \times n}$ is defined as
\begin{equation}
    C_{uv} = \sum_{k \in [K]} x_{u,k} x_{v,k}^\T.
\end{equation}


\section{Summary of results}\label{sec:summary_results}
We now provide a high-level overview of the main contributions of this paper. We provide summary cartoon in Figure~\ref{fig:quantum_classical_round_cartoon}, depicting the quantum embedding of the problem and the quantum rounding protocols. Let $(V, E)$ be a graph where we label the vertices by $V = [m]$, and denote the objective function of Problem~\eqref{eq:LNCG} by
\begin{equation}\label{eq:summary_obj}
    f(\mathbf{R}) \coloneqq \sum_{(u, v) \in E} \langle C_{uv}, R_u R_v^\T \rangle.
\end{equation}

\begin{figure}
    \centering
    \includegraphics[width=0.75\textwidth]{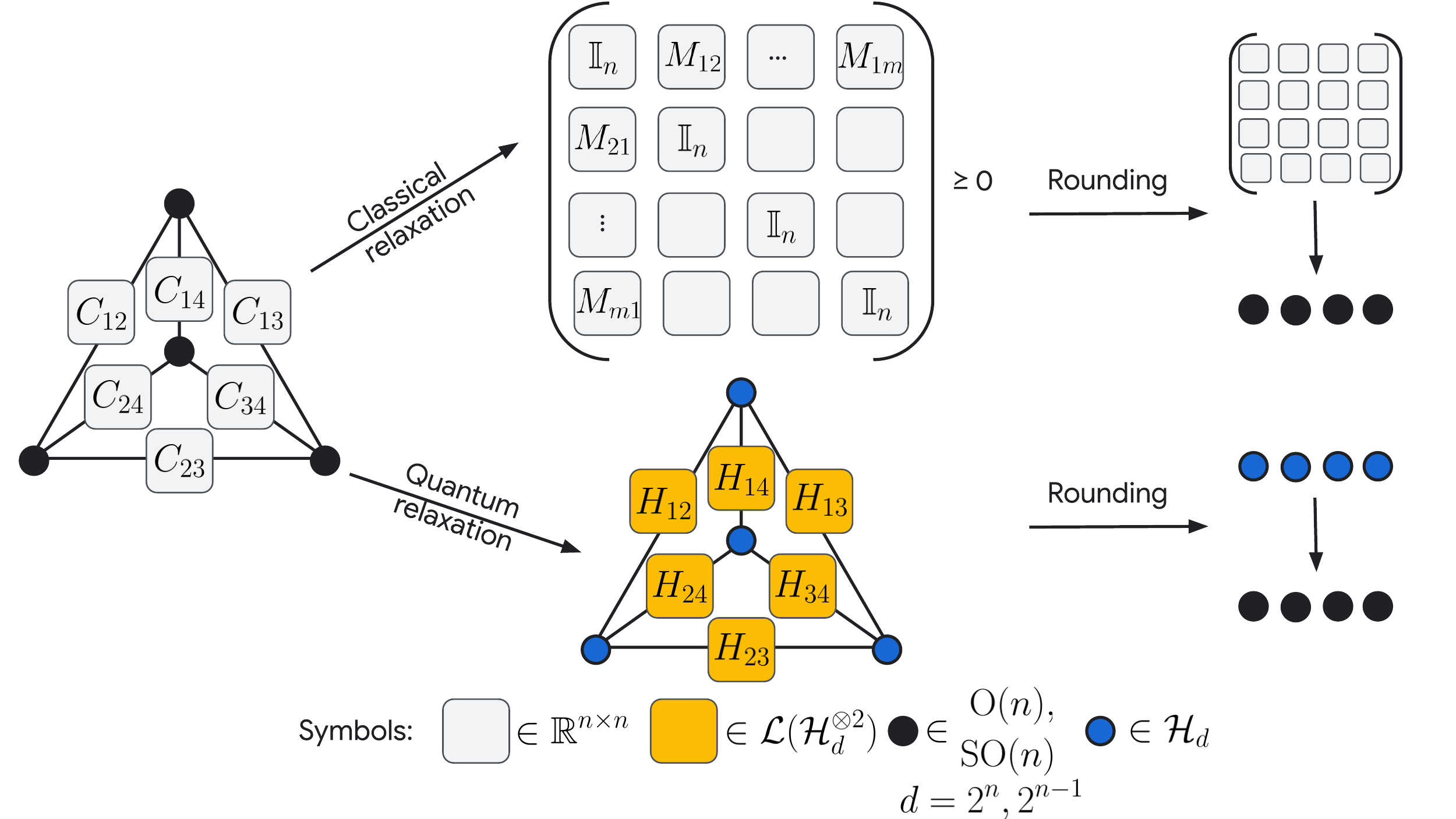}
    \caption{A cartoon description of our quantum encoding of the LNCG problem, compared to the classical formulation. (Left) The description of the problem that we consider, which is described by a graph $([m], E)$ and $n \times n$ matrices $C_{uv}$ for each edge $(u, v) \in E$. We wish to assign elements of $\Orth(n)$ or $\SO(n)$ to each vertex such that Eq.~\eqref{eq:summary_obj} is maximized. (Center top) The standard classical relaxation optimizes an $mn \times mn$ PSD matrix $M$ via semidefinite programming. (Right top) The classical rounding procedure, which returns a collection of orthogonal matrices from $M$. (Center bottom) Our quantum formulation of the problem as a two-body fermionic Hamiltonian. On each vertex we place a $d$-dimensional Hilbert space, with interaction terms $H_{uv}$ on the edges constructed from each $C_{uv}$. The classical solution to the LNCG problem lies as a subset of this full Hilbert space. (Right bottom) Our proposed quantum rounding protocols. One protocol requires knowledge of the two-body reduced density matrices across edges, while the other uses the one-body reduced density matrices on each vertex.}
    \label{fig:quantum_classical_round_cartoon}
\end{figure}

\subsection{Quantum Hamiltonian relaxation}

First, consider the setting in which $\mathbf{R} = (R_1, \ldots, R_m) \in \Orth(n)^m$. We embed this problem into a Hamiltonian by placing $n$ qubits on each vertex $v \in [m]$, resulting in a total Hilbert space $\mathcal{H}_{2^n}^{\otimes m}$ of $mn$ qubits. Define the $n$-qubit Pauli operators
\begin{equation}\label{eq:P_ij_summary}
    P_{ij} \coloneqq \begin{cases}
    -X_i Z_{i+1} \cdots Z_{j-1} X_j & i < j,\\
    Z_i & i = j,\\
    -Y_j Z_{j+1} \cdots Z_{i-1} Y_i & i > j,
    \end{cases}
\end{equation}
where $Z_i \coloneqq \I_2^{\otimes (i-1)} \otimes Z \otimes \I_2^{\otimes (n-i)}$ (similarly for $X_i$, $Y_i$). The Hamiltonian
\begin{equation}\label{eq:LNCG_H_summary}
    H \coloneqq \sum_{(u, v) \in E} \sum_{i, j \in [n]} [C_{uv}]_{ij} \sum_{k \in [n]} P_{ik}^{(u)} \otimes P_{jk}^{(v)}
\end{equation}
defines our quantum relaxation of the objective $f$ over $\Orth(n)^m$. The notation $A^{(v)}$ denotes the operator $A$ acting only on the Hilbert space of vertex $v$, and we overload this notation to indicate either the $n$-qubit operator or $mn$-qubit operator acting trivially on the remaining vertices. When the context is clear we typically omit writing the trivial support.

For optimization over $\SO(n)^m$, we consider instead the $(n-1)$-qubit Pauli operators
\begin{align}
    \widetilde{P}_{ij} &\coloneqq \Pi_0 P_{ij} \Pi_0^\T,\\
    \Pi_0 &= \frac{1}{\sqrt{2}} \l( \bra{+} \otimes \I_2^{\otimes (n-1)} + \bra{-} \otimes Z^{\otimes (n-1)} \r),
\end{align}
where $\Pi_0 : \mathcal{H}_{2^n} \to \mathcal{H}_{2^{n-1}}$ represents the projection onto the even-parity subspace of $\mathcal{H}_{2^n}$. The construction of the relaxed Hamiltonian for $\SO(n)$ is then analogous to Eq.~\eqref{eq:LNCG_H_summary}:
\begin{equation}
    \widetilde{H} \coloneqq \sum_{u, v \in E} \sum_{i, j \in [n]} [C_{uv}]_{ij} \sum_{k \in [n]} \widetilde{P}_{ik}^{(u)} \otimes \widetilde{P}_{jk}^{(v)},
\end{equation}
where now the relaxed quantum problem is defined over $m(n-1)$ qubits.

These Hamiltonians serve as relaxations to Problem~\eqref{eq:LNCG} as follows. It is known that the operators $P_{ij}$ are a representation of the double cover of the orthogonal group, in the sense that for each $R \in \Orth(n)$ there exists a unit vector $\ket{\phi(R)} \in \mathcal{H}_{2^n}$ such that all $R_{ij} = \ev{P_{ij}}{\phi(R)}$~\cite{atiyah1964clifford,saunderson2015semidefinite}. Additionally, $R \in \SO(n)$ if and only if such $\ket{\phi(R)}$ has even parity. We go beyond these classical results by identifying each $P_{ij}$ as a one-body fermion operator, which implies that each $\ket{\phi(R)}$ is a free-fermionic Gaussian state. Hence solving Problem~\eqref{eq:LNCG} is equivalent to optimizing the Hamiltonian over a subset of quantum states:
\begin{equation}\label{eq:quantum_reformulation}
    \max_{\mathbf{R} \in G^m} f(\mathbf{R}) = \max_{\substack{\ket{\psi} = \bigotimes_{v \in [m]} \ket{\phi(R_v)} \\ R_v \in G}} \ev{H}{\psi}.
\end{equation}
Dropping these product-state constraints on $\ket{\psi}$ implies the inequalities
\begin{align}
    \max_{\mathbf{R} \in \Orth(n)^m} f(\mathbf{R}) &\leq \max_{\rho \in \mathcal{D}(\mathcal{H}_{2^n}^{\otimes m})} \tr(H \rho),\\
    \max_{\mathbf{R} \in \SO(n)^m} f(\mathbf{R}) &\leq \max_{\rho \in \mathcal{D}(\mathcal{H}_{2^{n-1}}^{\otimes m})} \tr(\widetilde{H} \rho),
\end{align}
where $\mathcal{D}(\mathcal{H})$ denotes the set of density operators on a Hilbert space $\mathcal{H}$. This establishes our quantum Hamiltonian relaxation.


\subsection{Quantum rounding}

\begin{algorithm}[t]
\caption{$\conv G$-based rounding of edge marginals\label{alg:conv_rounding}}
\KwData{Quantum state $\rho \in \mathcal{D}(\mathcal{H}_d^{\otimes m})$ over a graph of $m$ vertices, each with local Hilbert space of dimension $d = 2^n$ if $G = \Orth(n)$, or $d = 2^{n-1}$ if $G = \SO(n)$}
\KwResult{Orthogonal matrices on each vertex, $R_1, \ldots, R_m \in G$}

$\mathcal{M} \gets \I_{mn}$ \;

\For{$u \neq v \in [m]$}{
    \For{$(i, j) \in [n]^2$}{
        \uIf{$G = \Orth(n)$}{
            $[\mathcal{M}_{uv}]_{ij} \gets \frac{1}{n} \tr(\Gamma_{ij}^{(u,v)} \rho)$ \;
        }
        \ElseIf{$G = \SO(n)$}{
            $[\mathcal{M}_{uv}]_{ij} \gets \frac{1}{n} \tr(\widetilde{\Gamma}_{ij}^{(u,v)} \rho)$ \;
        }
    }
}

\For{$v \in [m]$}{
    $R_v \gets \argmin_{Y \in G} \| Y - \mathcal{M}_{1v} \|_F$ \;
}
\end{algorithm}

\begin{algorithm}[t]

\caption{Rounding vertex marginals\label{alg:vertex_rounding}}
\KwData{Quantum state $\rho \in \mathcal{D}(\mathcal{H}_d^{\otimes m})$ over a graph of $m$ vertices, each with local Hilbert space of dimension $d = 2^n$ if $G = \Orth(n)$, or $d = 2^{n-1}$ if $G = \SO(n)$}
\KwResult{Orthogonal matrices on each vertex, $R_1, \ldots, R_m \in G$}

\For{$v \in [m]$}{
    ${Q}_v \gets 0 \in \R^{n \times n}$ \;
    \For{$(i, j) \in [n]^2$}{
    \uIf{$G = \Orth(n)$}{
        $[{Q}_v]_{ij} \gets \tr(P_{ij}^{(v)} \rho)$ \;
    }
    \ElseIf{$G = \SO(n)$}{
        $[{Q}_v]_{ij} \gets \tr(\widetilde{P}_{ij}^{(v)} \rho)$ \;
    }
    }
}

\For{$v \in [m]$}{
    $R_v \gets \argmin_{Y \in G} \| Y - {Q}_v \|_F$ \;
}
\end{algorithm}

In order to recover orthogonal matrices from a relaxed quantum solution $\rho$, we propose two rounding procedures, summarized in Algorithms~\ref{alg:conv_rounding} and \ref{alg:vertex_rounding}. These rounding procedures operate on local (i.e., single- or two-vertex observables) expectation values of $\rho$ stored in classical memory, which can be efficiently estimated, e.g., by partial state tomography.

Algorithm~\ref{alg:conv_rounding} is inspired by constructing a quantum analogue of the PSD variable appearing in semidefinite relaxations to Problem~\eqref{eq:LNCG}. Consider the $mn \times mn$ matrix of expectation values
\begin{equation}
    \mathcal{M} \coloneqq \begin{bmatrix}
    \I_n & T_{12} & \cdots & T_{1m}\\
    T_{21} & \I_n & \cdots & T_{2m}\\
    \vdots & \vdots & \ddots & \vdots\\
    T_{m1} & T_{m2} & \cdots & \I_n
    \end{bmatrix},
\end{equation}
where the off-diagonal blocks are defined as
\begin{align}
    T_{uv} &\coloneqq \frac{1}{n} \begin{bmatrix}
    \tr(\Gamma_{11}^{(u, v)} \rho) & \cdots & \tr(\Gamma_{1n}^{(u, v)} \rho)\\
    \vdots & \ddots & \vdots\\
    \tr(\Gamma_{n1}^{(u, v)} \rho) & \cdots & \tr(\Gamma_{nn}^{(u, v)} \rho)
    \end{bmatrix} = T_{vu}^\T,\\
    \Gamma_{ij}^{(u, v)} &\coloneqq \sum_{k \in [n]} P_{ik}^{(u)} \otimes P_{jk}^{(v)}
\end{align}
when $G = \Orth(n)$, and we replace the operators $P_{ij}$ with $\widetilde{P}_{ij}$ when $G = \SO(n)$. We show that $\mathcal{M}$ satisfies the following properties for all states $\rho$:
\begin{align}
    &\mathcal{M} \succeq 0,\\
    &\mathcal{M}_{uv} \in \conv G \quad \forall u, v \in [m],
\end{align}
where $\conv G$ is the convex hull of $G$. Notably, when $G = \SO(n)$, $\mathcal{M}$ obeys the same constraints as the $\conv\SO(n)$-based semidefinite relaxation proposed by Saunderson \emph{et al.}~\cite{saunderson2014semidefinite}. However, whereas the classical representation of the $\conv\SO(n)$ constraints involves matrices of size $2^{n-1} \times 2^{n-1}$ for each edge, our quantum state automatically satisfies these constraints (using only $n - 1$ qubits per vertex).

Algorithm~\ref{alg:vertex_rounding} uses the single-vertex information $\tr(P_{ij}^{(v)} \rho)$ of $\rho$, as opposed to the two-vertex information $\tr(\Gamma_{ij}^{(u, v)} \rho)$. We consider this rounding procedure due to the fact that, if $\rho$ is a pure Gaussian state satisfying the constraint of Eq.~\eqref{eq:quantum_reformulation}, then the matrix of expectation values
\begin{equation}
    Q_v \coloneqq \begin{bmatrix}
    \tr(P_{11}^{(v)} \rho) & \cdots & \tr(P_{1n}^{(v)} \rho)\\
    \vdots & \ddots & \vdots\\
    \tr(P_{n1}^{(v)} \rho) & \cdots & \tr(P_{nn}^{(v)} \rho)
    \end{bmatrix}
\end{equation}
lies in $\Orth(n)$. On the other hand, for arbitrary density matrices we have the relaxation $Q_v \in \conv\Orth(n)$, and again when we replace $P_{ij}^{(v)}$ with $\widetilde{P}_{ij}^{(v)}$ then $Q_v \in \conv\SO(n)$.

Both rounding procedures use the standard projection~\cite{schonemann1966generalized} of the matrices $X = T_{uv}$ or $Q_v$ to some $R \in G$ by finding the nearest orthogonal matrix according to Frobenius-norm distance:
\begin{equation}\label{eq:nearest_ortho_summary}
    R = \argmin_{Y \in G} \|X - Y\|_F.
\end{equation}
This can be solved efficiently as a classical postprocessing step, essentially by computing the singular value decomposition of $X = U \Sigma V^\T$. When $G = \Orth(n)$, the solution is $R = UV^\T$. When $G = \SO(n)$, we instead use the so-called special singular value decomposition~\cite{sanyal2011orbitopes} of $X = U \widetilde{\Sigma} \widetilde{V}^\T$, where $\widetilde{\Sigma} = \Sigma J$ and $\widetilde{V} = VJ$, with $J$ being the diagonal matrix
\begin{equation}\label{eq:summary_J}
    J = \begin{bmatrix}
    \I_{n-1} & 0\\
    0 & \det(UV^\T)
    \end{bmatrix},
\end{equation}
assuming that the singular values $\sigma_i(X)$ are in descending order, $\sigma_1(X) \geq \cdots \geq \sigma_n(X)$. Then the solution to Eq.~\eqref{eq:nearest_ortho_summary} is $R = U \widetilde{V}^\T \in \SO(n)$.

\section{\label{sec:quantum_formalism}Quantum formalism for optimization over orthogonal matrices}
Our key insight into encoding orthogonal matrices into quantum states comes from the construction of the orthogonal group from a Clifford algebra~\cite{atiyah1964clifford}. We review this mathematical construction in Appendix~\ref{sec:clifford} and only discuss the main aspects here. The Clifford algebra $\Cl(n)$ is a $2^n$-dimensional real vector space equipped with an inner product and multiplication operation satisfying the anticommutation relation
\begin{equation}
    e_i e_j + e_j e_i = -2 \delta_{ij} \openone,
\end{equation}
where $e_1, \ldots, e_n$ is an orthonormal basis for $\R^n$ and $\openone$ is the multiplicative identity of the algebra. The orthogonal group is then realized through a quadratic map $Q : \Cl(n) \to \R^{n \times n}$ and the identification of a subgroup $\Pin(n) \subset \Cl(n)$ such that $Q(\Pin(n)) = \Orth(n)$. Notably, the elements of $\Pin(n)$ have unit norm (with respect to the inner product on $\Cl(n)$). The special orthogonal group, meanwhile, is constructed by considering only the even-parity elements of $\Cl(n)$, denoted by $\Cl^0(n)$. The group $\Spin(n) = \Pin(n) \cap \Cl^0(n)$ then yields $Q(\Spin(n)) = \SO(n)$.

Because the Clifford algebra $\Cl(n)$ is a $2^n$-dimensional vector space, we observe that it can be identified with a Hilbert space of $n$ qubits.\footnote{In fact, $n$ rebits suffice since $\Cl(n)$ is a real vector space, but to keep the presentation straightforward we will not make such a distinction.} In this section we explore this connection in detail, showing how to represent orthogonal matrices as quantum states and how the mapping $Q$ acts as a linear functional on those states.

\subsection{Qubit representation of the Clifford algebra}

First we describe the canonical isomorphism between $\Cl(n)$ and $\mathcal{H}_{2^n} \coloneqq (\R^2)^{\otimes n}$ as Hilbert spaces. We denote the standard basis of $\Cl(n)$ by $\{ e_I \coloneqq e_{i_1} \cdots e_{i_k} \mid I = \{i_1, \ldots, i_k\} \subseteq [n] \}$. By convention we assume that the elements of $I$ are ordered as $i_1 < \cdots < i_k$. Each basis element $e_I$ maps onto to a computational basis state $\ket{b}$, where $b = b_1 \cdots b_n \in \{0, 1\}^n$, via the correspondence
\begin{equation}
    e_I \equiv \bigotimes_{i \in [n]} \ket{b_i}, \quad \text{where } b_i = \begin{cases}
1  & \text{if } i \in I,\\
0 & \text{otherwise}.
\end{cases}
\end{equation}
The inner products on both spaces coincide since this associates one orthonormal basis to another. This correspondence also naturally equates the grade $|I|$ of the Clifford algebra with the Hamming weight $|b|$ of the qubits. The notion of parity, $|I| \mod 2 = |b| \mod 2$, is therefore preserved, so $\Cl^0(n)$ corresponds to the subspace of $\mathcal{H}_{2^n}$ with even Hamming weight.

To represent the multiplication of algebra elements in this Hilbert space, we use the fact that left- and right-multiplication are linear automorphisms on $\Cl(n)$, which are denoted by
\begin{equation}
    \lambda_x(y) = xy, \quad \rho_x(y) = yx.
\end{equation}
The action of the algebra can therefore be represented on $\mathcal{H}_{2^n}$ as linear operators. We shall use the matrix representation provided in Ref.~\cite{saunderson2015semidefinite}, as it precisely coincides with the $n$-qubit computational basis described above. Because of linearity, it suffices to specify left- and right-multiplication by the generators $e_i$, which are the operators
\begin{align}
    \lambda_i &\equiv Z^{\otimes (i - 1)} \otimes (-\i Y) \otimes \I_2^{\otimes (n - i)}, \label{eq:lambda}\\
    \rho_i &\equiv \I_2^{\otimes (i - 1)} \otimes (-\i Y) \otimes Z^{\otimes (n - i)}. \label{eq:rho}
\end{align}
It will also be useful to write down the parity automorphism $\alpha(e_I) = (-1)^{|I|} e_I$ under this matrix representation. As the notion of parity is equivalent between $\Cl(n)$ and $\mathcal{H}_{2^n}$, $\alpha$ is simply the $n$-qubit parity operator,
\begin{equation}\label{eq:alpha}
    \alpha \equiv Z^{\otimes n}.
\end{equation}

It will also be useful to represent the subspace $\Cl^0(n)$ explicitly as an $(n-1)$-qubit Hilbert space. This is achieved by the projection from $\Cl(n)$ to $\Cl^0(n)$, expressed in Ref.~\cite{saunderson2015semidefinite} as the $2^{n-1} \times 2^n$ matrix
\begin{equation}\label{eq:projector}
    \Pi_{0} \coloneqq \frac{1}{\sqrt{2}} \l( \bra{+} \otimes \I_2^{\otimes (n - 1)} + \bra{-} \otimes Z^{\otimes (n - 1)} \r).
\end{equation}
It is straightforward to check that $\Pi_0 \ket{b} = 0$ if $|b| \mod 2 = 1$, and that its image is a $2^{n-1}$-dimensional Hilbert space.

\subsection{The quadratic mapping as quantum expectation values}

The quadratic map $Q : \Cl(n) \to \R^{n \times n}$ is defined as
\begin{equation}
    Q(x)(v) \coloneqq \pi_{\R^n}(\alpha(x) v \overline{x}) \quad \forall x \in \Cl(n), v \in \R^n,
\end{equation}
where $\pi_{\R^n}$ is the projector from $\Cl(n)$ to $\R^n$,
\begin{equation}
    \pi_{\R^n}(x) \coloneqq \sum_{i \in [n]} \langle e_i, x \rangle e_i \quad \forall x \in \Cl(n),
\end{equation}
and the conjugation operation $x \mapsto \overline{x}$ is defined as the linear extension of $\overline{e_I} = (-1)^{|I|} e_{i_k} \cdots e_{i_1}$. This map associates Clifford algebra elements with orthogonal matrices via the relations $Q(\Pin(n)) = \Orth(n)$ and $Q(\Spin(n)) = \SO(n)$ (see Appendix~\ref{sec:clifford} for a review of the construction). In the standard basis of $\R^n$, the linear map $Q(x) : \R^n \to \R^n$ has the matrix elements
\begin{equation}
\begin{split}
    [Q(x)]_{ij} &= \langle e_i, Q(x)(e_j) \rangle\\
    &= \langle e_i, \alpha(x) e_j \overline{x} \rangle.
\end{split}
\end{equation}
Using the linear maps $\lambda_i, \rho_j$ of left- and right-multiplication by $e_i, e_j$, as well as the conjugation identity $\langle x, y\overline{z} \rangle = \langle xz, y \rangle$ in the Clifford algebra, these matrix elements of $Q(x)$ can be rearranged as
\begin{equation}\label{eq:Q_matrix_elements}
\begin{split}
    [Q(x)]_{ij} &= \langle e_i, \alpha(x) e_j \overline{x} \rangle\\
    &= \langle e_i x, \alpha(x) e_j \rangle\\
    &= \langle \lambda_i(x), \rho_j(\alpha(x)) \rangle\\
    &= \langle x, \lambda_i^\dagger(\rho_j(\alpha(x))) \rangle.
\end{split}
\end{equation}

We now transfer this classical expression to the quantum representation developed above. First, define the following $n$-qubit Pauli operators as the composition of the linear maps appearing in Eq.~\eqref{eq:Q_matrix_elements}:
\begin{equation}
    P_{ij} \coloneqq \lambda_i^\dagger \rho_j \alpha = \begin{cases}
    -\I_2^{\otimes (i - 1)} \otimes X \otimes Z^{\otimes (j - i - 1)} \otimes X \otimes \I_2^{\otimes (n - j)} & i < j,\\
    \I_2^{\otimes (i - 1)} \otimes Z \otimes \I_2^{\otimes (n - i)} & i = j,\\
    -\I_2^{\otimes (j - 1)} \otimes Y \otimes Z^{\otimes (i - j - 1)} \otimes Y \otimes \I_2^{\otimes (n - i)} & i > j,
    \end{cases}
\end{equation}
where the expressions in terms of Pauli matrices follow from Eqs.~\eqref{eq:lambda} to \eqref{eq:alpha}. Then we may rewrite Eq.~\eqref{eq:Q_matrix_elements} as
\begin{equation}
    [Q(x)]_{ij} = \langle x | P_{ij} | x \rangle,
\end{equation}
where $\ket{x} \in \mathcal{H}_{2^n}$ is the quantum state identified with $x \in \Cl(n)$. Hence, the matrix elements of $Q(x) \in \R^{n \times n}$ possess the interpretation as expectation values of a collection of $n^2$ Pauli observables $\{ P_{ij} \}_{i, j \in [n]}$. Furthermore, recall that $Q(x) \in \Orth(n)$ if and only if $x \in \Pin(n)$, and $Q(x) \in \SO(n)$ if and only if $x \in \Spin(n)$. Because $\Spin(n) = \Pin(n) \cap \Cl^0(n)$, one can work in the even-parity sector directly by projecting the operators as
\begin{equation}\label{eq:P_ij_son}
    \widetilde{P}_{ij} \coloneqq \Pi_0 P_{ij} \Pi_0^\T.
\end{equation}
These are $(n-1)$-qubit Pauli operators, and we provide explicit expressions in Appendix~\ref{sec:even_subspace}. When necessary, we may specify another map $\widetilde{Q} : \Cl^0(n) \to \R^{n \times n}$,
\begin{equation}
    [\widetilde{Q}(x)]_{ij} \coloneqq \langle x | \widetilde{P}_{ij} | x \rangle,
\end{equation}
for which $\widetilde{Q}(\Spin(n)) = \SO(n)$.

In general, these double covers are only a subset of the unit sphere in $\mathcal{H}_{d}$ ($d = 2^n$ or $2^{n-1}$), so not all quantum states mapped by $Q$ yield orthogonal matrices. In Section~\ref{sec:free-fermion} we characterize the elements of $\Pin(n)$ and $\Spin(n)$ as a class of well-studied quantum states, namely, pure fermionic Gaussian states.

\subsection{\label{sec:free-fermion}Fermionic representation of the construction}

\subsubsection{Notation}

First we establish some notation. A system of $n$ fermionic modes, described by the creation operators $a_1^\dagger, \ldots, a_n^\dagger$, can be equivalently represented by the $2n$ Majorana operators
\begin{align}
    \gamma_i &= a_i + a_i^\dagger,\\
    \widetilde{\gamma}_i &= -\i(a_i - a_i^\dagger),
\end{align}
for all $i \in [n]$. These operators form a representation for the Clifford algebra $\Cl(2n)$, as they satisfy\footnote{Note that we adopt the physicist's convention here, which takes the generators to be Hermitian, as opposed to Eq.~\eqref{eq:anticommutator_cl_n} wherein they square to $-\openone$.}
\begin{align}
    \gamma_i \gamma_j + \gamma_j \gamma_i = \widetilde{\gamma}_i \widetilde{\gamma}_j + \widetilde{\gamma}_j \widetilde{\gamma}_i &= 2 \delta_{ij} \openone,\\
    \gamma_i \widetilde{\gamma}_j + \widetilde{\gamma}_j \gamma_i &= 0.
\end{align}
The Jordan--Wigner mapping allows us to identify this fermionic system with an $n$-qubit system via the relations
\begin{align}
    \gamma_{i} &= Z^{\otimes (i - 1)} \otimes X \otimes \I_2^{\otimes (n - i)}, \label{eq:gamma_X}\\
    \widetilde{\gamma}_i &= Z^{\otimes (i - 1)} \otimes Y \otimes \I_2^{\otimes (n - i)}.\label{eq:gamma_Y}
\end{align}
We will work with the two representations interchangeably.

A central tool for describing noninteracting fermions is the Bogoliubov transformation $\bm{\gamma} \mapsto O \bm{\gamma}$, where $O \in \Orth(2n)$ and
\begin{equation}
    \bm{\gamma} \coloneqq \begin{bmatrix}
    \widetilde{\gamma}_1 & \cdots & \widetilde{\gamma}_n & \gamma_1 & \cdots & \gamma_n
    \end{bmatrix}^\T.
\end{equation}
This transformation is achieved by fermionic Gaussian unitaries, which are equivalent to matchgate circuits on qubits under the Jordan--Wigner mapping~\cite{knill2001fermionic,terhal2002classical,jozsa2008matchgates}. In particular, we will make use of a subgroup of such unitaries corresponding to $\Orth(n) \times \Orth(n) \subset \Orth(2n)$. For any $U, V \in \Orth(n)$, let $\mathcal{U}_{(U, V)}$ be the fermionic Gaussian unitary with the adjoint action
\begin{align}
    \mathcal{U}_{(U, V)} \widetilde{\gamma}_i \mathcal{U}_{(U, V)}^\dagger &= \sum_{j \in [n]} U_{ij} \widetilde{\gamma}_j, \label{eq:gamma_tilde_gaussian}\\
    \mathcal{U}_{(U, V)} \gamma_i \mathcal{U}_{(U, V)}^\dagger &= \sum_{j \in [n]} V_{ij} \gamma_j. \label{eq:gamma_gaussian}
\end{align}
In contrast to arbitrary $\Orth(2n)$ transformations, these unitaries do not mix between the $\gamma$- and $\widetilde{\gamma}$-type Majorana operators.

\subsubsection{\label{sec:linear_free_fermion}Linear optimization as free-fermion models}

Applying the representation of Majorana operators under the Jordan--Wigner transformation, Eqs.~\eqref{eq:gamma_X} and \eqref{eq:gamma_Y},  to the Clifford algebra automorphisms, Eqs.~\eqref{eq:lambda} to \eqref{eq:alpha}, we see that $\lambda_i^\dagger = \i \widetilde{\gamma}_i$ and $\rho_j \alpha = \gamma_j$. Therefore the Pauli operators $P_{ij}$ defining the quadratic map $Q$ are equivalent to fermionic one-body operators,
\begin{equation}
    P_{ij} = \i \widetilde{\gamma}_i \gamma_j.
\end{equation}
Consider now a linear objective function $\ell(X) \coloneqq \langle C, X \rangle$ for some fixed $C \in \R^{n \times n}$, which we wish to optimize over $\Orth(n)$:
\begin{equation}\label{eq:linear_obj}
    \max_{X \in \Orth(n)} \ell(X) = \max_{X \in \Orth(n)} \langle C, X \rangle.
\end{equation}
Because we require $X \in \Orth(n)$, it is equivalent to search over all $x \in \Pin(n)$ through $Q$:
\begin{equation}
    \max_{X \in \Orth(n)} \langle C, X \rangle = \max_{x \in \Pin(n)} \langle C, Q(x) \rangle.
\end{equation}
Writing out the matrix elements explicitly, we see that the objective takes the form
\begin{equation}
\begin{split}
    \ell(X) &= \sum_{i,j \in [n]} C_{ij} [Q(x)]_{ij}\\
    &= \sum_{i,j \in [n]} C_{ij} \langle x | P_{ij} | x \rangle\\
    &= \langle x | F(C) | x \rangle,
\end{split}
\end{equation}
where we have defined the noninteracting fermionic Hamiltonian
\begin{equation}
    F(C) \coloneqq \sum_{i,j \in [n]} C_{ij} P_{ij} = \i \sum_{i,j \in [n]} C_{ij} \widetilde{\gamma}_i \gamma_j.
\end{equation}
The linear optimization problem is therefore equivalent to solving a free-fermion model,
\begin{equation}\label{eq:linear_obj_Pin}
    \max_{X \in \Orth(n)} \langle C, X \rangle = \max_{x \in \Pin(n)} \langle x | F(C) | x \rangle,
\end{equation}
the eigenvectors of which are fermionic Gaussian states. As such, this problem can be solved efficiently by a classical algorithm. In fact, the known classical algorithm for solving the optimization problem is exactly the same as that used for diagonalizing $F(C)$.

We now review the standard method to diagonalize $F(C)$. Consider the singular value decomposition of $C = U \Sigma V^\T$, which is computable in time $\Ord(n^3)$. This decomposition immediately reveals the diagonal form of the Hamiltonian:
\begin{equation}
\begin{split}
    F(C) &= \i \sum_{i, j \in [n]} [U \Sigma V^\T]_{ij} \widetilde{\gamma}_i \gamma_j\\
    &= \i \sum_{k \in [n]} \sigma_k(C) \l( \sum_{i \in [n]} [U^\T]_{ki} \widetilde{\gamma}_i \r) \l( \sum_{j \in [n]} [V^\T]_{kj} \gamma_j \r)\\
    &= {\mathcal{U}}_{(U, V)}^\dagger \l( \sum_{k \in [n]} \sigma_k(C) \i \widetilde{\gamma}_k \gamma_k \r) {\mathcal{U}}_{(U, V)}.
\end{split}
\end{equation}
Because $\i \widetilde{\gamma}_k \gamma_k = Z_k$, it follows that the eigenvectors of $F(C)$ are the fermionic Gaussian states
\begin{equation}
    \ket{\phi_b} = {\mathcal{U}}_{(U, V)}^\dagger \ket{b}, \quad b \in \{0, 1\}^n,
\end{equation}
with eigenvalues
\begin{equation}
    E_b = \sum_{k \in [n]} (-1)^{b_k} \sigma_k(C).
\end{equation}
The maximum energy is $E_{0^n} = \tr \Sigma$ since all singular values are nonnegative. The corresponding eigenstate $\ket{\phi_{0^n}}$ is the maximizer of Eq.~\eqref{eq:linear_obj_Pin}, so it corresponds to an element $\phi_{0^n} \in \Pin(n)$. It is straightforward to see this by recognizing that $[Q(\phi_{0^n})]_{ij} = \ev{\i \widetilde{\gamma}_i \gamma_j}{\phi_{0^n}} = [UV^\T]_{ij}$. The fact that $Q(\phi_{0^n}) \in \Orth(n)$ if and only if $\phi_{0^n} \in \Pin(n)$ concludes the argument. 


Indeed, the standard classical algorithm~\cite{schonemann1966generalized} for solving Eq.~\eqref{eq:linear_obj} uses precisely the same decomposition. From the cyclic property of the trace and the fact that $\Orth(n)$ is a group, we have
\begin{equation}
    \max_{X \in \Orth(n)} \langle U \Sigma V^\T, X \rangle = \max_{X' \in \Orth(n)} \langle \Sigma, X' \rangle,
\end{equation}
where we have employed the change of variables $X' \coloneqq U^\T X V$. Again, because $\Sigma$ has only nonnegative entries, $\langle \Sigma, X' \rangle$ achieves its maximum, $\tr \Sigma$, when $X' = \I_n$. This implies that the optimal solution is $X = UV^\T$. Note that this problem is equivalent to minimizing the Frobenius-norm distance, since
\begin{equation}
\begin{split}
    \argmin_{X \in \Orth(n)} \| C - X \|_F^2 &= \argmin_{X \in \Orth(n)} \l( \|C\|_F^2 + \|X\|_F^2 - 2 \langle C, X \rangle \r)\\
    &= \argmax_{X \in \Orth(n)} \, \langle C, X \rangle.
\end{split}
\end{equation}

Now suppose we wish to optimize $\ell$ over $\SO(n)$. In this setting, one instead computes $X = U\widetilde{V}^\T$ from the special singular value decomposition of $C = U \widetilde{\Sigma} \widetilde{V}^\T$. This ensures that $\det(X) = 1$ while maximizing $\ell(X)$, as only the smallest singular value $\sigma_n(C)$ has its sign potentially flipped to guarantee the positive determinant constraint. This sign flip also has a direct analogue within the free-fermion perspective. Recall that the determinant of $Q(x) \in \Orth(n)$ is given by the parity of $x \in \Pin(n)$, or equivalently the parity of the state $\ket{x}$ in the computational basis. Note also that all fermionic states are eigenstates of the parity operator. To optimize over $\SO(n)$, we therefore seek the maximal eigenstate $\ket{\phi_b}$ of $F(C)$ which has even parity. If $\ev{Z^{\otimes n}}{\phi_{0^n}} = 1$ then we are done. On the other hand, if $\ev{Z^{\otimes n}}{\phi_{0^n}} = -1$ then we need to flip only a single bit in $0^n$ to reach an even-parity state. The smallest change in energy by such a flip is achieved from changing the occupation of the mode corresponding to the smallest singular value of $C$. The resulting eigenstate $\ket{\phi_{0^{n-1}1}}$ is then the even-parity state with the largest energy, $E_{0^{n-1}1} = \tr\Sigma - 2\sigma_n(C)$.

Finally, we point out that all elements of $\Pin(n)$ are free-fermion states. To see this, observe that $C$ is arbitrary. We can therefore construct the family of Hamiltonians $\{F(C) \mid C \in \Orth(n)\}$. Clearly, the maximum $\langle C, X \rangle = n$ within this family is achieved when $X = C$, each of which corresponds to a fermionic Gaussian state $\ket{\phi}$ satisfying $F(C)\ket{\phi} = n\ket{\phi}$ and $Q(\phi) = C$. We note that this argument generalizes the mathematical one presented in Ref.~\cite{saunderson2015semidefinite}, which only considered the eigenvectors lying in $\Spin(n)$.

\subsubsection{\label{sec:mixed_states}Mixed states and the convex hull}

First we review descriptions of the convex hull of orthogonal and rotation matrices, the latter of which was characterized by Saunderson \emph{et al.}~\cite{saunderson2015semidefinite}. The convex hull of $\Orth(n)$ is the set of all matrices with operator norm bounded by 1,
\begin{equation}\label{eq:convOn_singval}
    \conv\Orth(n) = \l\{ X \in \R^{n \times n} \mid \sigma_1(X) \leq 1 \r\}.
\end{equation}
On the other hand, the convex hull of $\SO(n)$ has a more complicated description in terms of special singular values:
\begin{equation}\label{eq:convSOn_singval}
    \conv\SO(n) = \l\{ X \in \R^{n \times n} \mathrel{\Bigg|} \sum_{i \in [n] \setminus I} \widetilde{\sigma}_i(X) - \sum_{i \in I} \widetilde{\sigma}_i(X) \leq n - 2 \quad \forall I \subseteq [n], |I| \text{ odd} \r\}.
\end{equation}
Saunderson \emph{et al.}~\cite{saunderson2015semidefinite} establish that this convex body is a spectrahedron, the feasible region of a semidefinite program. The representation that we will be interested in is called a PSD lift:
\begin{equation}\label{eq:conv-son_psd-lift}
    \conv\SO(n) = \l\{
    \begin{bmatrix}
    \langle \widetilde{P}_{11}, \rho \rangle & \cdots & \langle \widetilde{P}_{1n}, \rho \rangle\\
    \vdots & \ddots & \vdots\\
    \langle \widetilde{P}_{n1}, \rho \rangle & \cdots & \langle \widetilde{P}_{nn}, \rho \rangle
    \end{bmatrix}
    \mathrel{\Bigg|} \rho \succeq 0, \tr \rho = 1 \r\},
\end{equation}
where the $2^{n-1} \times 2^{n-1}$ matrices $\widetilde{P}_{ij}$ are defined in Eq.~\eqref{eq:P_ij_son}.\footnote{Technically, Saunderson \emph{et al.}~\cite{saunderson2015semidefinite} use the definition $\widetilde{P}_{ij} = -\Pi_0 \lambda_i \rho_j \Pi_0^\T$ because they employ the standard adjoint representation, which differs from our use of the twisted adjoint representation which includes the parity automorphism $\alpha$. However since $\alpha(x) = x$ for all $x \in \Cl^0(n)$, both definitions of $\widetilde{P}_{ij}$ coincide.} 

Recall that the density operators on a Hilbert space $\mathcal{H}$ form the convex hull of its pure states:
\begin{equation}
    \mathcal{D}(\mathcal{H}) \coloneqq \conv\{ \op{\psi}{\psi} \mid \ket{\psi} \in \mathcal{H}, \ip{\psi}{\psi} = 1 \} = \{ \rho \in \mathcal{L}(\mathcal{H}) \mid \rho \succeq 0, \tr\rho = 1 \}.
\end{equation}
From Eq.~\eqref{eq:conv-son_psd-lift} one immediately recognizes that the PSD lift of $\conv\SO(n)$ corresponds to $\mathcal{D}(\mathcal{H}_{2^{n-1}})$, where we recognize that $\mathcal{H}_{2^{n-1}} \cong \Cl^0(n)$. Furthermore, the projection of the lift is achieved through the convexification of the map $Q : \Cl(n) \to \R^{n \times n}$, where the fact that $Q$ is quadratic in $\Cl(n)$ translates to being linear in $\mathcal{D}(\Cl(n))$. Specifically, by a slight abuse of notation we shall extend the definition of $Q$ to act on density operators as
\begin{equation}
    Q(\rho) = \sum_{\mu} p_\mu Q(x_\mu), \quad \text{where } \rho = \sum_\mu p_\mu \op{x_\mu}{x_\mu}.
\end{equation}
Then Eq.~\eqref{eq:conv-son_psd-lift} is the statement that $Q(\mathcal{D}(\Cl^0(n))) = \conv\SO(n)$.

In Appendix~\ref{sec:psd_lift_proof} we show that this statement straightforwardly generalizes for $Q(\mathcal{D}(\Cl(n))) = \conv\Orth(n)$. We prove this using the fermionic representation developed in Section~\ref{sec:linear_free_fermion}, and furthermore use these techniques to provide an alternative derivation for the PSD lift of $\conv\SO(n)$. The core of our argument is showing that the singular-value conditions of Eqs.~\eqref{eq:convOn_singval} and \eqref{eq:convSOn_singval} translate into bounds on the largest eigenvalue of corresponding $n$-qubit observables:
\begin{align}
    \sigma_i(X) &= \tr(\i \widetilde{\gamma}_i \gamma_i \rho) \leq 1,\\
    \sum_{i \in [n] \setminus I} \widetilde{\sigma}_i(X) - \sum_{i \in I} \widetilde{\sigma}_i(X) &= \tr\l[ \rho_0 \l( \sum_{i \in [n] \setminus I} \i \widetilde{\gamma}_i \gamma_i - \sum_{i \in I} \i \widetilde{\gamma}_i \gamma_i \r) \r] \leq n - 2,
\end{align}
where $\rho \in \mathcal{D}(\Cl(n))$ and $\rho_0 \in \mathcal{D}(\Cl^0(n))$. The physical interpretation here is that not all pure quantum states map onto to orthogonal or rotation matrices (which is clear from the fact that fermionic Gaussian states are only a subset of quantum states). However, all density operators \emph{do} map onto to their convex hulls, and the distinction between $\conv\Orth(n)$ and $\conv\SO(n)$ can be automatically specified by restricting the support of $\rho$ to the even-parity subspace.


\section{Quantum relaxation for the quadratic problem}\label{sec:quantum_relax_section}
We now arrive at the primary problem of interest in this work, the little noncommutative Grothendieck problem over the (special) orthogonal group. While the linear problem of Eq.~\eqref{eq:linear_obj} can be solved classically in polynomial time, quadratic programs are considerably more difficult. Here, we use the quantum formalism of the Pin and Spin groups developed above to construct a quantum relaxation of this problem. Then in Section~\ref{sec:rounding} we describe rounding procedures to recover a collection of orthogonal matrices from the quantum solution to this relaxation.

Recall the description of the input to Problem~\eqref{eq:LNCG}. Let $(V, E)$ be a graph, and associate to each edge $(u, v) \in E$ a matrix $C_{uv} \in \R^{n \times n}$. We label the vertices as $V = [m]$. We wish to maximize the objective
\begin{equation}
    f(R_1, \ldots, R_m) = \sum_{(u, v) \in E} \langle C_{uv}, R_u R_v^\T \rangle.
\end{equation}
over $(R_1, \ldots, R_m) \in \Orth(n)^m$. First, expand this expression in terms of matrix elements:
\begin{equation}
    \sum_{(u, v) \in E} \langle C_{uv}, R_u R_v^\T \rangle = \sum_{(u, v) \in E} \sum_{i,j \in [n]} [C_{uv}]_{ij} \sum_{k \in [n]} [R_u]_{ik} [R_v^\T]_{kj}.
\end{equation}
From the quadratic mapping $Q : \Cl(n) \to \R^{n \times n}$, we know that for each $R \in G$ there exists some $\phi \in \Pin(n)$ such that $R_{ij} = \langle \phi | P_{ij} | \phi \rangle$. Hence we can express the matrix product as
\begin{equation}
    \begin{split}
    [R_u]_{ik} [R_v^\T]_{kj} &= \langle \phi_u | P_{ik} | \phi_u \rangle \langle \phi_v | P_{jk} | \phi_v \rangle\\
    &= \langle \phi_u \otimes \phi_v | P_{ik} \otimes P_{jk} | \phi_u \otimes \phi_v \rangle,
    \end{split}
\end{equation}
which is now the expectation value of a $2n$-qubit Pauli operator with respect to a product state of two Gaussian states $\ket{\phi_u}$, $\ket{\phi_v}$. To extend this over the entire graph, we define a Hilbert space of $m$ registers of $n$ qubits each. For each edge $(u, v) \in E$ we introduce the Hamiltonian terms
\begin{equation}
    H_{uv} \coloneqq \sum_{i,j \in [n]} [C_{uv}]_{ij} \Gamma_{ij}^{(u, v)},
\end{equation}
where
\begin{equation}
    \Gamma_{ij}^{(u, v)} \coloneqq \l(\sum_{k \in [n]} P_{ik}^{(u)} \otimes P_{jk}^{(v)}\r) \bigotimes_{w \in V \setminus \{u, v\}} \I_{2^n}^{(w)}.
\end{equation}
To simplify notation, we shall omit the trivial support $\bigotimes_{w \in V \setminus \{u, v\}} \I_{2^n}^{(w)}$ when the context is clear.

The problem is now reformulated as optimizing the $mn$-qubit Hamiltonian
\begin{equation}\label{eq:H_LNCG}
    H \coloneqq \sum_{(u, v) \in E} H_{uv} = \sum_{(u, v) \in E} \sum_{i,j \in [n]} [C_{uv}]_{ij} \sum_{k \in [n]} P_{ik}^{(u)} \otimes P_{jk}^{(v)}.
\end{equation}
The exact LNCG problem over $\Orth(n)$ then corresponds to
\begin{equation}\label{eq:quantum_exact}
    \max_{\mathbf{R} \in \Orth(n)^m} f(\mathbf{R}) = \max_{\ket{\psi} \in \mathcal{H}_{2^n}^{\otimes m}} \ev{H}{\psi} \quad \text{subject to } \begin{cases}
    \ip{\psi}{\psi} = 1,\\
    \ket{\psi} = \bigotimes_{v \in [m]} \ket{\phi_v},\\
    \ket{\phi_v} = {\mathcal{U}}_{(R_v, \I_n)} \ket{0^n}, & R_v \in \Orth(n) \ \forall v \in [m].
    \end{cases}
\end{equation}
The hardness of this problem is therefore related to finding the optimal separable state for local Hamiltonians, which is NP-hard in general~\cite{gurvits2004classical,ioannou2006computational,gharibian2008strong}. Dropping these constraints on the state provides a relaxation of the problem, since
\begin{equation}
    \max_{\substack{\ket{\psi} \in \mathcal{H}_{2^n}^{\otimes m}, \\ \ip{\psi}{\psi} = 1}} \ev{H}{\psi} \geq \max_{\mathbf{R} \in G^m} f(\mathbf{R}).
\end{equation}



We point out here that the Hamiltonian terms $H_{uv}$ can be interpreted as two-body fermionic interactions. Note that there is an important distinction between two-body fermionic operators (Clifford-algebra products of four Majorana operators) and two-body qudit operators (tensor products of two qudit Pauli operators). Recall that $P_{ij} = \i \widetilde{\gamma}_i \gamma_j$ is one-body in the fermionic sense. While the operators $P_{ik}^{(u)} \otimes P_{jk}^{(v)}$ appear to mix both notions, here they in fact coincide. To see this, we consider a global algebra of Majorana operators $\{\gamma_{i + (v-1)n}, \widetilde{\gamma}_{i + (v-1)n} \mid i \in [n], v \in [m] \}$ acting on a Hilbert space of $mn$ fermionic modes. While it is not true that the local single-mode Majorana operators map onto the global single-mode operators, i.e.,
\begin{align}
    \gamma_i^{(v)} \bigotimes_{w \in V \setminus \{v\}} \I_{2^n}^{(w)} \neq \gamma_{i + (v-1)n},
\end{align}
the local two-mode Majorana operators in fact do correspond to global two-mode operators:
\begin{equation}
    \widetilde{\gamma}_i^{(v)} \gamma_j^{(v)} \bigotimes_{w \in V \setminus \{v\}} \I_{2^n}^{(w)} = \widetilde{\gamma}_{i + (v-1)n} \gamma_{j + (v-1)n}.
\end{equation}
Thus, taking the tensor product of two local two-mode Majorana operators on different vertices is equivalent to taking the product of two global two-mode Majorana operators:
\begin{equation}
    \widetilde{\gamma}_i^{(u)} \gamma_j^{(u)} \otimes \widetilde{\gamma}_k^{(v)} \gamma_l^{(v)} \bigotimes_{w \in V \setminus \{u, v\}} \I_{2^n}^{(w)} = \widetilde{\gamma}_{i + (u-1)n} \gamma_{j + (u-1)n} \widetilde{\gamma}_{k + (v-1)n} \gamma_{l + (v-1)n}.
\end{equation}
Therefore Eq.~\eqref{eq:H_LNCG} can be equivalently expressed as a Hamiltonian with two-body fermionic interactions.

Finally, when we wish to optimize over $(R_1, \ldots, R_m) \in \SO(n)^m$, it is straightforward to see that we can simply replace the terms $P_{ij}$ with $\widetilde{P}_{ij}$. Defining
\begin{align}
    \widetilde{\Gamma}_{ij}^{(u, v)} &\coloneqq \l(\sum_{k \in [n]} \widetilde{P}_{ik}^{(u)} \otimes \widetilde{P}_{jk}^{(v)}\r) \bigotimes_{w \in V \setminus \{u, v\}} \I_{2^{n-1}}^{(w)},\\
    \widetilde{H}_{uv} &\coloneqq \sum_{i,j \in [n]} [C_{uv}]_{ij} \widetilde{\Gamma}_{ij}^{(u, v)},
\end{align}
the quantum relaxation for the $\SO(n)$ problem is given by the $m(n-1)$-qubit Hamiltonian
\begin{equation}
    \widetilde{H} \coloneqq \sum_{(u, v) \in E} \widetilde{H}_{uv}.
\end{equation}

Clearly, as local (fermionic) Hamiltonians, both $H$ and $\widetilde{H}$ lie in QMA~\cite{kempe2006complexity,liu2007quantum}. In fact, they are also QMA-hard, which can be shown by a reduction to an arbitrary instance of the XY model~\cite{cubitt2016complexity}. We prove this reduction in Appendix~\ref{sec:qma-hardness_proof}. Therefore rather than trying to solve the QMA-hard problem, our goal will be to produce an approximate ground state in polynomial time. This approximation is then classically rounded to the desired collection of orthogonal matrices, which we subsequently describe in Section~\ref{sec:rounding}. Afterwards, in Section~\ref{sec:numerical_experiments} we will explore the effectiveness of this approximation with numerical experiments.

\section{\label{sec:rounding}Rounding algorithms}
Optimizing the energy of a local Hamiltonian is a well-studied problem, both from the perspective of quantum and classical algorithms. In this section we will assume that such an algorithm has been used to produce the state $\rho \in \mathcal{D}(\mathcal{H}_d^{\otimes m})$ which (approximately) maximizes the energy $\tr(H \rho)$. We wish to round this state into the feasible space, namely the set of product states of Gaussian states. We do so by rounding the expectation values of $\rho$ appropriately, such that we return some valid approximation $R_1, \ldots, R_m \in G$. In this section we propose two approaches to perform this quantum rounding.

The first uses insight from the fact that our quantum relaxation is equivalent to a classical semidefinite relaxation with additional constraints based on the convex hull of the orthogonal group. This is approach is particularly advantageous when optimizing over $\SO(n)$, as $\conv\SO(n)$ has a matrix representation exponential in $n$ (its PSD lift). To build the semidefinite variable from the quantum state, we require measurements of the expectation values of the two-vertex operators $\Gamma_{ij}^{(u, v)} = \sum_{k \in [n]} P_{ik}^{(u)} \otimes P_{jk}^{(v)}$ for each pair of vertices $(u, v)$. We refer this procedure as \emph{$\conv\SO(n)$-based rounding}.\footnote{This rounding can also be applied to the optimization problem over $\Orth(n)$ as well, but we are particularly interested in the $\conv\SO(n)$ constraints due to their exponentially large classical representation.} Our second rounding protocol uses the expectation values of $P_{ij}^{(v)}$ of each vertex $v$ directly. In this case, rather than expectation values of two-vertex operators as before, we only require the information of single-vertex marginals $\rho_v \coloneqq \tr_{\neg v}(\rho)$. Therefore we call this approach \emph{vertex-marginal rounding}.


If $\rho$ is produced by a deterministic classical algorithm, then the relevant expectation values can be exactly computed (to machine precision). However if the state is produced by a randomized algorithm, or is otherwise prepared by a quantum computer, then we can only estimate the expectation values to within statistical error by some form of sampling. In the quantum setting, this can be achieved either by partial state tomography~\cite{bonet2020nearly,zhao2021fermionic} or a more sophisticated measurement protocol~\cite{huggins2021nearly}.\footnote{For the present discussion we do not consider the effects of finite sampling, although we expect that rounding is fairly robust to such errors since it will always return a solution in the feasible space.} See Appendix~\ref{sec:tomography} for further comments on this quantum measurement aspect. The rounding algorithms then operate entirely as classical postprocessing after estimating the necessary expectation values.

\subsection{Approximation ratios for rounding the classical SDP}

Before describing our quantum rounding protocols, we first review classical relaxations and rounding procedures for Problem~\eqref{eq:LNCG} in order to put our results in context. The standard semidefinite relaxation can be expressed as the SDP
\begin{equation}\label{eq:LNCG_relaxed}
    \max_{M \in \R^{mn \times mn}} \langle C, M \rangle \quad \text{subject to } \begin{cases}
    M \succeq 0,\\
    M_{vv} = \I_n & \forall v \in [m],
    \end{cases}
\end{equation}
where $C \in \R^{mn \times mn}$ is the matrix with $n \times n$ blocks $C_{uv}$. If an additional nonconvex constraint $\mathrm{rank}(M) = n$ is imposed, then the solution would be exact:
\begin{equation}
    M = \mathbf{R} \mathbf{R}^\T = \begin{bmatrix}
    \I_n & R_1 R_2^\T & \cdots & R_1 R_m^\T\\
    R_2 R_1^\T & \I_n & \cdots & R_2 R_m^\T\\
    \vdots & \vdots & \ddots & \vdots\\
    R_m R_1^\T & R_m R_2^\T & \cdots & \I_n
    \end{bmatrix}.
\end{equation}
Problem~\eqref{eq:LNCG_relaxed}, without the rank constraint, is therefore a relaxation of the original problem. However, the solution $M \in \R^{mn \times mn}$ is still PSD, so it can be decomposed as $M = \mathbf{X} \mathbf{X}^\T$, where
\begin{equation}
    \mathbf{X} = \begin{bmatrix}
    X_1\\
    \vdots\\
    X_m
    \end{bmatrix}, \quad X_v \in \R^{n \times mn}.
\end{equation}
The rounding algorithm of Bandeira \emph{et al.}~\cite{bandeira2016approximating} then computes, for each $v \in [m]$,
\begin{equation}
    O_v = \mathcal{P}(X_v Z) \coloneqq \argmin_{Y \in \Orth(n)} \| Y - X_v Z \|_F,
\end{equation}
where $Z$ is an $mn \times n$ Gaussian random matrix whose entries are drawn i.i.d.~from $\mathcal{N}(0, 1/n)$. When optimizing over $G = \Orth(n)$, this rounded solution guarantees (in expectation) an approximation ratio of
\begin{equation}
    \alpha_{\Orth(n)}^2 = \E\l[ \frac{1}{n} \sum_{i \in [n]} \sigma_i(Z_1) \r{]^2},
\end{equation}
where $Z_1$ is a random $n \times n$ matrix with i.i.d.~entries from $\mathcal{N}(0, 1 / n)$. This expression can be evaluated for any given $n$, and some numerical values are provided in Ref.~\cite{bandeira2016approximating}.

\subsubsection{The $\SO(n)$ setting}

The approximation ratios when one demands rounding to $\SO(n)$ elements were not originally considered by Bandeira \emph{et al.}~\cite{bandeira2016approximating}. However, as the special orthogonal setting is one of the main draws of our quantum formulation, we first establish analogous classical approximation results here. In Appendix~\ref{sec:son_approx_ratio}, we show how the argument of Bandeira \emph{et al.}~\cite{bandeira2016approximating} can be extended to the $\SO(n)$ setting, yielding an approximation ratio of
\begin{equation}
    \alpha_{\SO(n)}^2 = \E\l[ \frac{1}{n} \sum_{i \in [n-1]} \sigma_i(Z_1) \r{]^2}.
\end{equation}
The only algorithmic change is that the rounding procedure is modified to project to the nearest $\SO(n)$ element via $\widetilde{\mathcal{P}}$, defined as
\begin{equation}
    \widetilde{\mathcal{P}}(X) \coloneqq \argmin_{Y \in \SO(n)} \|Y - X \|_F.
\end{equation}
We also show that $\E[\sigma_n(Z_1)] > 0$ for all finite $n$, hence it follows that $\alpha_{\SO(n)}^2 < \alpha_{\Orth(n)}^2$. This provides evidence that solving for rotations is generally a more difficult problem than optimizing over all orthogonal matrices (see Ref.~\cite[Section~4.3]{pumir2021generalized} for a brief discussion). For small values of $n$, the numerical values of both the $\Orth(n)$ and $\SO(n)$ approximation ratios are (computed using Mathematica):
\begin{align}
    \alpha_{\Orth(2)}^2 \approx 0.6564, &\qquad \alpha_{\SO(2)}^2 \approx 0.3927,\\
    \alpha_{\Orth(3)}^2 \approx 0.6704, &\qquad \alpha_{\SO(3)}^2 \approx 0.5476,\\
    \alpha_{\Orth(4)}^2 \approx 0.6795, &\qquad \alpha_{\SO(4)}^2 \approx 0.6096.
\end{align}
In Appendix~\ref{sec:son_approx_ratio} we provide an integral expression for $\alpha_{\SO(n)}$ which can be evaluated for arbitrary $n$.

The problem over $\SO(n)$ can be augmented further by introducing constraints to the SDP based on the convex hull of $\SO(n)$. Saunderson \emph{et al.}~\cite{saunderson2014semidefinite} apply such constraints on each block of $M$:
\begin{equation}\label{eq:LNCG_SD_conv}
    \max_{M \in \R^{mn \times mn}} \langle C, M \rangle \quad \text{subject to } \begin{cases}
    M \succeq 0,\\
    M_{vv} = \I_n & \forall v \in [m],\\
    M_{uv} \in \conv\SO(n) & \forall u, v \in [m].
    \end{cases}
\end{equation}
Although they do not prove approximation guarantees for this enhanced SDP, they first show that, if one reintroduces the rank constraint on $M$, then the convex constraint $M_{uv} \in \conv\SO(n)$ is sufficient to guarantee the much stronger condition $M_{uv} \in \SO(n)$. Then, when dropping the rank constraint (but leaving the $\conv\SO(n)$ constraint) they show that the relaxed problem is still exact over certain types of graphs, such as tree graphs. Finally, they provide numerical evidence that even when the relaxation is not exact, it returns substantially more accurate approximations than the standard SDP.

However, the cost of this augmented SDP is the exponential size of representing $\conv\SO(n)$ elements, as seen by the PSD lift in Eq.~\eqref{eq:conv-son_psd-lift}. In a different work by the same authors, Saunderson \emph{et al.}~\cite{saunderson2015semidefinite} proved that, roughly speaking, semidefinite representations of $\conv\SO(n)$ necessarily have size at least $\frac{1}{n} 2^{\Omega(n)}$.

\subsection{\label{sec:quantum_gram}Quantum Gram matrix}

We now introduce the rounding procedures for our quantum relaxed solution. Analogous to the classical SDP solution $M$, we can form a matrix $\mathcal{M} \in \R^{mn \times mn}$ from the expectation values of $\rho$ as
\begin{equation}
    \mathcal{M} \coloneqq \begin{bmatrix}
    \I_n & T_{12} & \cdots & T_{1m}\\
    T_{21} & \I_n & \cdots & T_{2m}\\
    \vdots & \vdots & \ddots & \vdots\\
    T_{m1} & T_{m2} & \cdots & \I_n
    \end{bmatrix},
\end{equation}
where
\begin{equation}\label{eq:T_uv}
    T_{uv} \coloneqq \frac{1}{n} \begin{bmatrix}
    \tr(\Gamma_{11}^{(u, v)} \rho) & \cdots & \tr(\Gamma_{1n}^{(u, v)} \rho)\\
    \vdots & \ddots & \vdots\\
    \tr(\Gamma_{n1}^{(u, v)} \rho) & \cdots & \tr(\Gamma_{nn}^{(u, v)} \rho)
    \end{bmatrix}
\end{equation}
and $T_{vu} = T_{uv}^\T$ for $u < v$. Just as $\langle C, M \rangle$ gives the relaxed objective value (up to rescaling and constant shifts), here we have that $\langle C, \mathcal{M} \rangle = \frac{2}{n} \tr(H \rho) + \tr(C)$. In Appendix~\ref{sec:quantum_conv_son} we show that for any quantum state, $\mathcal{M}$ satisfies the following properties:
\begin{equation}
\begin{cases}\label{eq:quantum_M_conditions}
    \mathcal{M} \succeq 0,\\
    \mathcal{M}_{vv} = \I_n & \forall v \in [m],\\
    \mathcal{M}_{uv} \in \conv\Orth(n) & \forall u, v \in [m].
\end{cases}
\end{equation}
Furthermore, we show that when $\rho$ is supported only on the even subspace of each single-vertex Hilbert space (or equivalently, if we replace $\Gamma_{ij}^{(u, v)}$ with $\widetilde{\Gamma}_{ij}^{(u, v)}$ in Eq.~\eqref{eq:T_uv}), then
\begin{equation}\label{eq:conv_son_T}
    \frac{1}{n} \begin{bmatrix}
    \tr(\widetilde{\Gamma}_{11}^{(u, v)} \rho) & \cdots & \tr(\widetilde{\Gamma}_{1n}^{(u, v)} \rho)\\
    \vdots & \ddots & \vdots\\
    \tr(\widetilde{\Gamma}_{n1}^{(u, v)} \rho) & \cdots & \tr(\widetilde{\Gamma}_{nn}^{(u, v)} \rho)
    \end{bmatrix} \in \conv\SO(n) \quad \forall \rho \in \mathcal{D}(\mathcal{H}_{2^{n-1}}^{\otimes m}).
\end{equation}
Therefore when optimizing the relaxed Hamiltonian $\widetilde{H}$ for the $\SO(n)$ setting, we are guaranteed to automatically satisfy the $\conv\SO(n)$ constraints.

\subsection{\label{sec:conv_rounding}$\conv\SO(n)$-based rounding}

Given the construction of the $\mathcal{M}$ from quantum expectation values, we proceed to round the Gram matrix as in the classical SDP with $\conv\SO(n)$ constraints~\cite{saunderson2014semidefinite}. This consists of computing the matrices
\begin{equation}
    R_v = \widetilde{\mathcal{P}}(\mathcal{M}_{1v}),
\end{equation}
where the projection to $\SO(n)$ can be efficiently computed from the special singular value decomposition, i.e.,
\begin{equation}
    \widetilde{\mathcal{P}}(X) = U \widetilde{V}^\T
\end{equation}
(recall Eq.~\eqref{eq:summary_J}). Our choice of rounding using the first $n \times mn$ ``row'' of $\mathcal{M}$ amounts to fixing $R_1 = \I_n$. We note that the same rounding procedure can naturally be applied to the $\Orth(n)$ setting as well, replacing $\widetilde{\mathcal{P}}$ with $\mathcal{P}$.

\subsection{\label{sec:vertex_rounding}Vertex-marginal rounding}

The single-vertex marginals are obtained by tracing out the qudits associated to all but one vertex $v \in [m]$,
\begin{equation}
    \rho_v = \tr_{\neg v}(\rho).
\end{equation}
As $\rho_v \in \mathcal{D}(\mathcal{H}_d^{\otimes m})$, from Section~\ref{sec:mixed_states} we have that $Q(\rho_v) \in \conv G$, where we linearly extend the definition of $Q$ to
\begin{equation}
    Q(\rho_v) \coloneqq \begin{bmatrix}
    \tr(P_{11}^{(v)} \rho) & \cdots & \tr(P_{1n}^{(v)} \rho)\\
    \vdots & \ddots & \vdots\\
    \tr(P_{n1}^{(v)} \rho) & \cdots & \tr(P_{nn}^{(v)} \rho)
    \end{bmatrix}.
\end{equation}
The rounding scheme we propose here then projects $Q(\rho_v)$ to $G$ using either $\mathcal{P}$ or $\widetilde{\mathcal{P}}$:
\begin{equation}
    R_v = \argmin_{Y \in G} \|Y - Q(\rho_v)\|_F.
\end{equation}

We point out that the relaxed Hamiltonian only has two-vertex terms which we seek to maximize. In Appendix~\ref{sec:symmetry_vertex} we show that $H$ commutes with $\mathcal{U}_{(\I_n, V)}^{\otimes m}$ for all $V \in \Orth(n)$, which we further show implies that $H$ may possess eigenstates whose single-vertex marginals obey $Q(\sigma_v) = 0$. This indicates that there may exist eigenstates of $H$ whose single-vertex marginals yield no information, despite the fact that their two-vertex marginals are nontrivial. In our numerical studies, we observe that breaking this symmetry resolves this issue. We accomplish this by including small perturbative one-body terms which correspond to the trace of $Q(\sigma_v)$:
\begin{equation}
    H_1 = \sum_{v \in [m]} \sum_{i \in [n]} P_{ii}^{(v)},
\end{equation}
Note that this trace quantity is importantly invariant with respect to the choice of basis for $\R^n$. We then augment the objective Hamiltonian with $H_1$, defining
\begin{equation}
    H'(\zeta) \coloneqq H + \zeta H_1
\end{equation}
where $\zeta > 0$ is a small regularizing parameter. While this one-body perturbation does not correspond to any terms in the original quadratic objective function, any arbitrarily small $\zeta > 0$ suffices to break the $\Orth(n)$ symmetry. Furthermore, the rounding procedure always guarantees that the solution is projected back into the feasible space $G^m$. When $G = \SO(n)$ we define $\widetilde{H}'(\zeta)$ analogously.

\section{Numerical experiments}\label{sec:numerical_experiments}
To explore the potential of our quantum relaxation and rounding procedures, we performed numerical experiments on randomly generated instances of the group synchronization problem. Because the Hilbert-space dimension grows exponentially in both $m$ and $n$, our classical simulations here are limited to small problem sizes. However, optimizing over rotations in $\R^3$ (requiring only two qubits per vertex) is highly relevant to many practical applications, so here we focus on the problem of $\SO(3)$ group synchronization. For example, this problem appears in the context of cryo-EM as described in Section~\ref{sec:group_sync}.

To model the problem, we generated random instances by selecting random three-regular graphs $([m], E)$, uniformly sampling $m$ rotations $g_1, \ldots, g_m \in \SO(n)$ according to the Haar measure, and then constructing $C_{uv} = g_u g_v^\T + \sigma W_{uv}$ for each $(u, v) \in E$, where the Gaussian noise matrix $W_{uv} \in \R^{n \times n}$ has i.i.d.~elements drawn from $\mathcal{N}(0, 1)$ and $\sigma \geq 0$ represents the strength (standard deviation) of this noise. This model describes a typical instance of group synchronization wherein the data is inherently corrupted by a low signal-to-noise ratio~\cite{ling2023solving}. Therefore our numerics serve to demonstrate the average-case performance of our quantum relaxation and rounding, as opposed to a worst-case theoretical analysis.

While the classical $\conv\SO(n)$-based SDP is not guaranteed to find the optimal solution, the problems studied here were selected for such that this enhanced SDP in fact does solve the exact problem. We verify this property by confirming that $\mathrm{rank}(M) = n$ before rounding on each problem instance. In this way we are able to calculate an approximation ratio for the other methods (as it is not clear how to solve for the globally optimal solution in general, even with an exponential-time classical algorithm). The methods compared here include our quantum relaxation with $\conv\SO(n)$-based rounding (denoted CR), vertex-marginal rounding (VR), and the classical SDP (without $\conv\SO(n)$ constraints but using the $\widetilde{\mathcal{P}}$ projection to guarantee that the rounded solutions are elements of $\SO(n)$). When using the vertex-rounding method, we employ $\widetilde{H}'(\zeta)$ as the objective Hamiltonian with $\zeta = 10^{-6}$.

We explore two classes of quantum states to probe the effectiveness of our relaxation. First, in Section~\ref{sec:eigenvector_numerics} we study the maximal eigenstate of the Hamiltonian. This represents the optimal solution to the relaxation, which can be seen as the quantum analogue to solving the classical SDP before rounding. However, computing extremal eigenstates of local Hamiltonians is QMA-hard in general (which is harder than the underlying classical NP-hard problem). Therefore, we next consider in Section~\ref{sec:asp_numerics} quantum states which can be prepared in at most polynomial time on a quantum computer. This allows us to study the potential for developing an efficient quantum algorithm derived from the quantum relaxation.

\subsection{\label{sec:eigenvector_numerics}Exact eigenvectors}

First, we consider the solution obtained by rounding the maximum eigenvector of $\widetilde{H}$. Although the hardness of preparing such a state is equivalent that of the ground-state problem, this nonetheless provides us with a benchmark for the ultimate approximation quality of our quantum relaxation. In Figure~\ref{fig:eigen_group-sync} we plot the approximation ratio of the rounded quantum states and compare to that of the classical SDP on the same problem instances. Each violin plot was constructed from the results of 50 random instances.

The results here demonstrate that, while the approximation quality of the classical SDP quickly falls off with larger graph sizes, our rounded quantum solutions maintain high approximation ratios, at least for the problem sizes probed here. Notably, the $\conv\SO(n)$-based rounding on the quantum state is significantly more powerful and consistent than the vertex-marginal rounding. This feature is not unexpected since, as discussed in Section~\ref{sec:vertex_rounding}, we are maximizing an objective Hamiltonian with only two-body terms, whereas the single-vertex rounding uses strictly one-body expectation values. Furthermore, as demonstrated in previous works~\cite{saunderson2014semidefinite, matni2014convex} the $\conv\SO(n)$ constraints are powerful in practice, and so we expect that the quantum rounding protocol which makes use of this structure enjoys the same advantages.

Meanwhile, when varying the noise parameter $\sigma$, we observe that all methods are fairly consistent. In particular, the $\conv\SO(n)$-based rounding only shows an appreciable decrease in approximation quality when the noise is considerable (note that $\sigma = 0.5 \approx 10^{-0.3}$ is a relatively large amount of noise, since $g_u g_v^\T$ is an orthogonal matrix and therefore has matrix elements bounded in magnitude by 1).
\begin{figure}
    \includegraphics[width=0.495\textwidth]{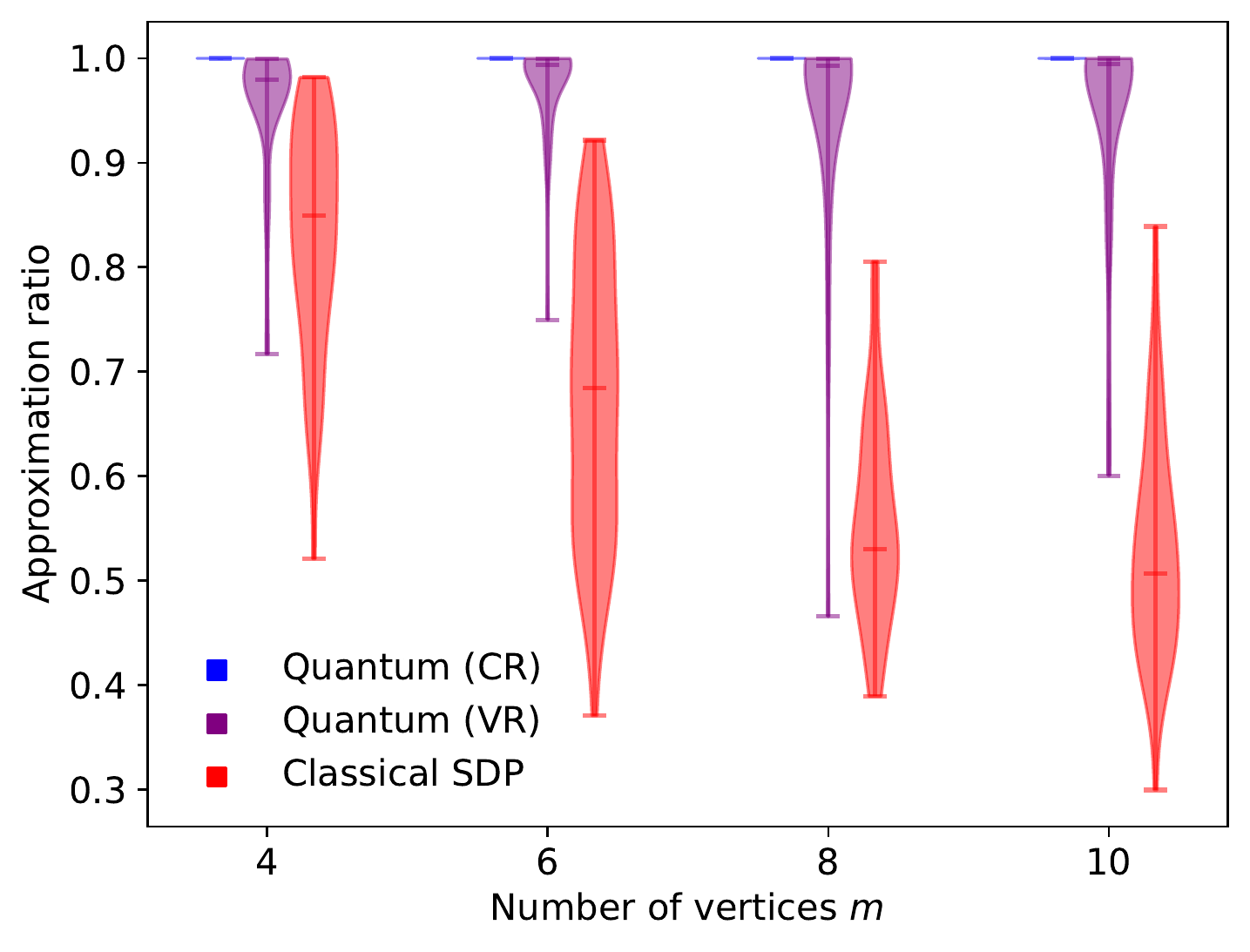}
    \includegraphics[width=0.495\textwidth]{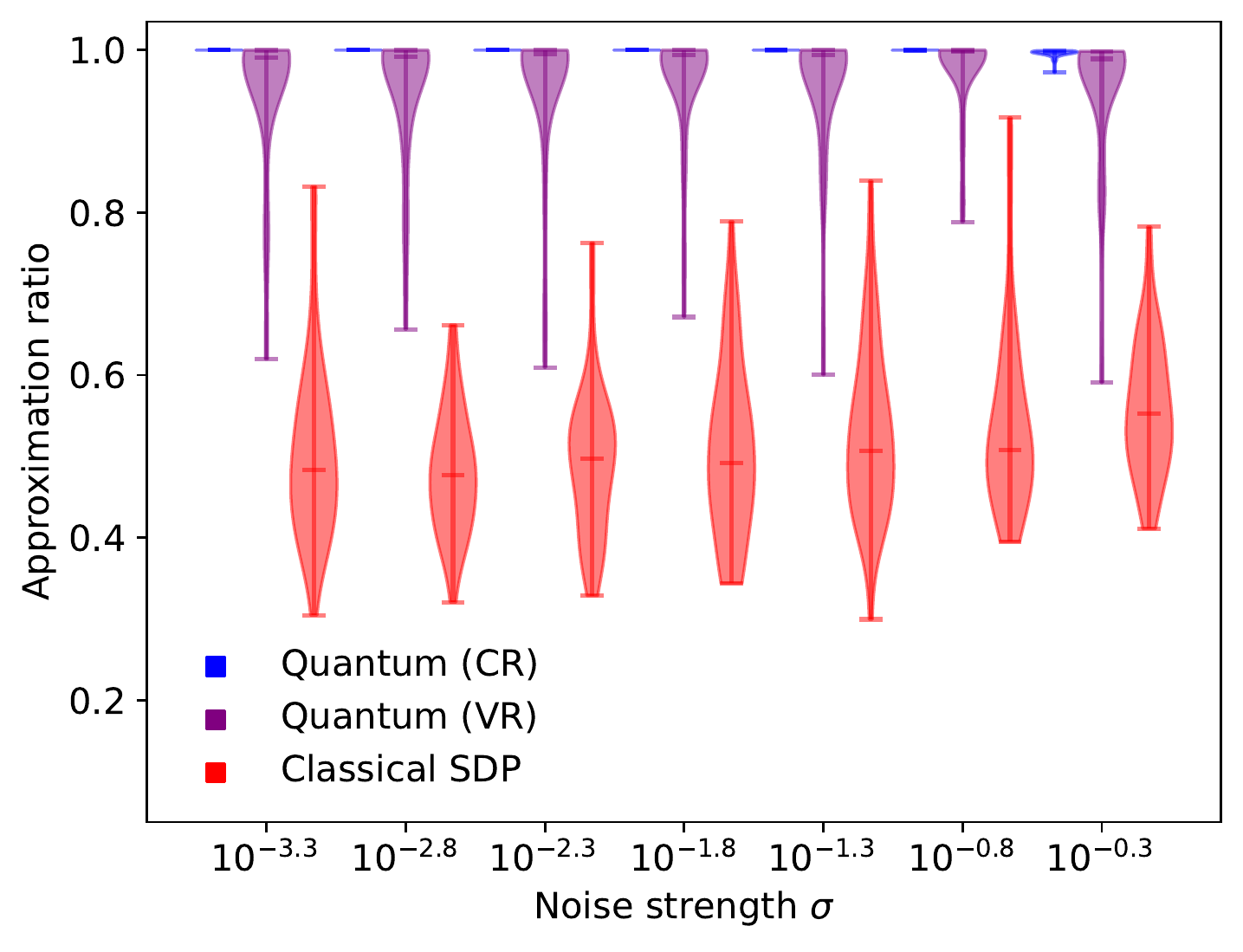}
    \caption{Approximation ratios for solutions obtained from rounding the maximum eigenvector of the relaxed Hamiltonian $\widetilde{H}$. Violin plots show the distribution of approximation ratios over 50 randomly generated instance, and with the median being indicated by the center marker. CR refers to rounding according to the $\conv\SO(n)$-based scheme (Section~\ref{sec:conv_rounding}), while VR denotes the vertex-marginal rounding scheme (Section~\ref{sec:vertex_rounding}). The classical SDP solution was rounded by the standard randomized algorithm~\cite{bandeira2016approximating}, and we report the best solution over 1000 rounding trials. (Left) Varying the number of vertices $m$ in the graph (random 3-regular graphs). Note that the number of qubits required here is $2m$. (Right) Varying the noise strength parameter $\sigma$ which defines the problem via $C_{uv} = g_u g_v^\T + \sigma W_{uv}$.}
    \label{fig:eigen_group-sync}
\end{figure}





\subsection{\label{sec:asp_numerics}Quasi-adiabatic state preparation}

\begin{figure}
    \centering
    \includegraphics[width=0.5\textwidth]{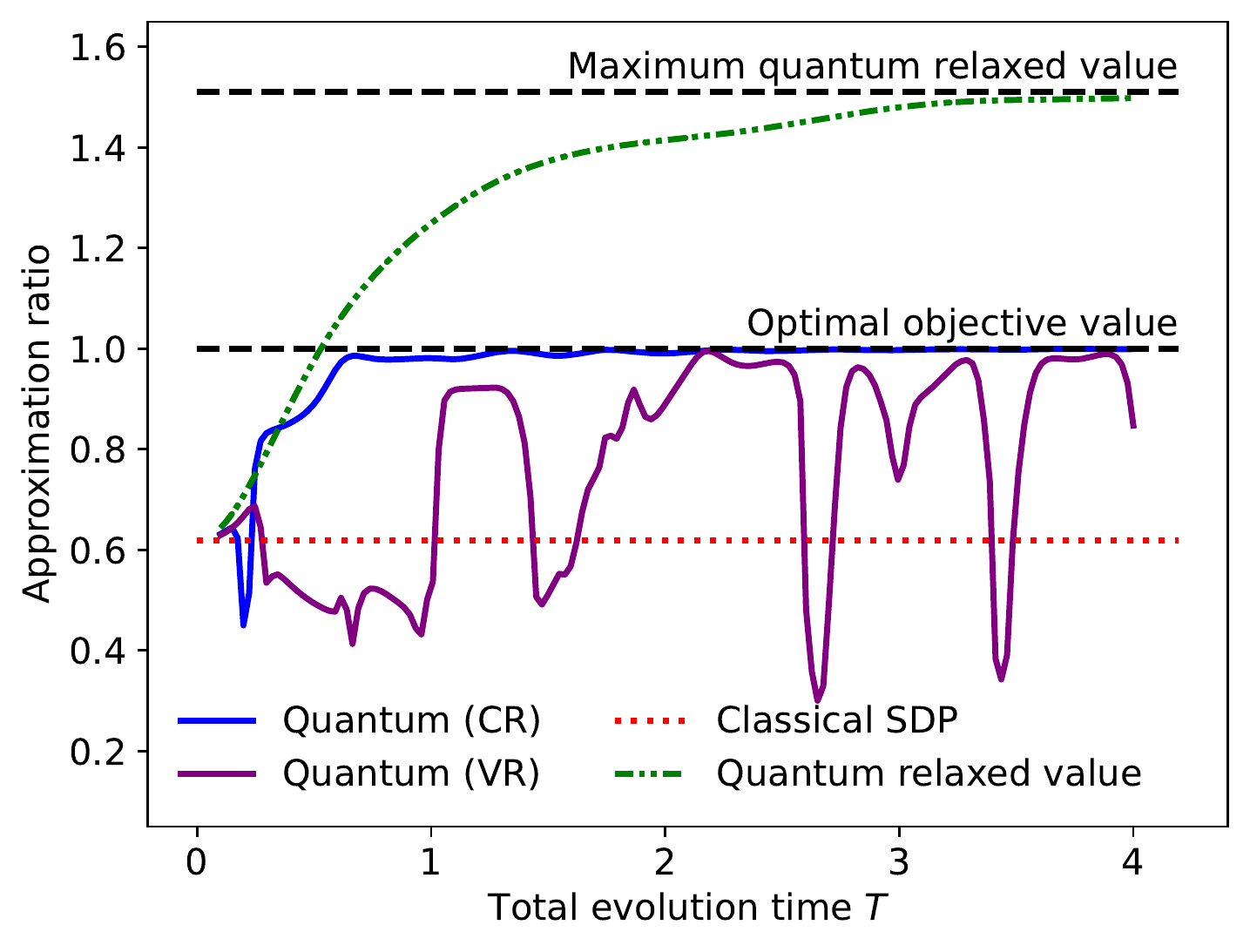}
    \caption{Demonstration of a typical instance of adiabatic state preparation for preparing relaxed quantum solutions. The initial state is the product of Gaussian states corresponding to the rounded solution of the classical SDP. As the total evolution time $T$ increases, the evolution becomes more adiabatic, indicated by the convergence of the relaxed value to the maximum eigenvalue (in units of the original problem's optimal value). The rounded solutions of course can never exceed the original problem's optimal value.}
    \label{fig:adiabatic_group-sync}
\end{figure}

Because it may be unrealistic to prepare the maximum eigenvector of $\widetilde{H}$, here we consider preparing states using ideas from adiabatic quantum computation~\cite{RevModPhys.90.015002}. Specifically, we wish to demonstrate that states whose relaxed energy may be far from the maximum eigenvalue can still provide high-quality approximations after rounding. If this is the case then we do not need to prepare very close approximations to the maximum eigenstate of $\widetilde{H}$, so the rigorous conditions of adiabatic state preparation may not be required in this context. Hence we consider ``quasi-adiabatic'' state preparation, wherein we explore how time-evolution speeds far from the adiabatic limit may still return high-quality approximations. Our numerical experiments here provide a preliminary investigation into this conjecture.

For simplicity of the demonstration, we consider a linear annealing schedule according to the time-dependent Hamiltonian
\begin{equation}
    H(t) = \l( 1 - \frac{t}{T} \r) H_i + \frac{t}{T} H_f,
\end{equation}
which prepares the state
\begin{equation}
    \ket{\psi(T)} = \mathcal{T}\exp\l( -\i \int_0^T \mathrm{d}t \, H(t) \r) \ket{\psi(0)}
\end{equation}
for some $T > 0$, where $\mathcal{T}$ is the time-ordering operator. The final Hamiltonian $H_f$ is the desired objective LNCG Hamiltonian,
\begin{equation}
    H_f = \widetilde{H}.
\end{equation}
The initial Hamiltonian $H_i$ is the parent Hamiltonian of the initial state, which we choose to be the approximation obtained from the classical SDP, as it can be obtained classically in polynomial time. Let $R_1, \ldots, R_m \in \SO(n)$ be the SDP solution. Our initial state is then the product of Gaussian states
\begin{equation}
    \ket{\psi(0)} = \bigotimes_{v \in [m]} \ket{\phi(R_v)},
\end{equation}
where each $\ket{\phi(R_v)}$ is the maximum eigenvector of the free-fermion Hamiltonian
\begin{equation}
    F(R_v) = \i \sum_{i, j \in [n]} [R_v]_{ij} \widetilde{\gamma}_i \gamma_j.
\end{equation}
Therefore the initial Hamiltonian $H_i$ is a sum of such free-fermion Hamiltonians (here we include the even-subspace projection since we are working with $\SO(n)$):
\begin{equation}
    H_i = \sum_{v \in [m]} \Pi_0 F(R_v) \Pi_0^\T.
\end{equation}
As a Gaussian state, $\ket{\phi(R_v)}$ can be prepared exactly from a quantum circuit of $\Ord(n^2)$ gates~\cite{PhysRevApplied.9.044036,zhao2024group}. Note that since we are working directly in the even subspace of $n - 1$ qubits here, this $n$-qubit circuit must be projected appropriately using $\Pi_0$. We discuss how to perform this circuit recompilation in Appendix~\ref{sec:even_subspace}. We comment that this choice of initial state is that of a mean-field state for non-number-preserving fermionic systems, for instance as obtained from Hartree--Fock--Bogoliubov theory. Suitably, the final Hamiltonian we evolve into is non-number-preserving two-body fermionic Hamiltonian.

In adiabatic state preparation, the total evolution time $T$ controls how close the final state $\ket{\psi(T)}$ is to the maximum eigenstate\footnote{We remind the reader that we are starting in the maximum eigenstate of the initial Hamiltonian, whereas in the physics literature, adiabatic theorems are typically stated in terms of ground states. Of course, the two perspectives are equivalent by simply an overall sign change (note that all Hamiltonians here are traceless).} of the final Hamiltonian $H_f$. One metric of closeness is how the energy of the prepared state, $\ev{H_f}{\psi(T)}$, compares to the maximum eigenvalue of $H_f$. On the other hand, as a relaxation, this maximum energy is already larger than the optimal objective value of the original problem. We showcase this in Figure~\ref{fig:adiabatic_group-sync}, using one random problem instance as a demonstrative (typical) example on a graph of $m = 6$ vertices (12 qubits). For each total evolution time point $T$, we computed $\ket{\psi(T)}$ by numerically integrating the time-dependent Schr\"{o}dinger equation, and we plot its relaxed energy as well as its rounded objective values. For large $T$ we approach the maximum eigenstate of $\widetilde{H}$ as expected (thereby also demonstrating that the initial ``mean-field'' state $\ket{\psi(0)}$ has appreciable overlap). Particularly interesting is the behavior for relatively small total evolution times $T$, wherein the energy of $\ket{\psi(T)}$ is far from the maximum eigenenergy. Despite this, the approximation quality after rounding the state using $\mathcal{M}$ is nearly exact around $T \approx 1$. On the other hand, the approximation quality of vertex-marginal rounding is highly inconsistent, which again we attribute to the fact that the single-vertex information is not directly seen by the final Hamiltonian $H_f$.

\begin{figure}
    \centering
    \includegraphics[width=0.495\textwidth]{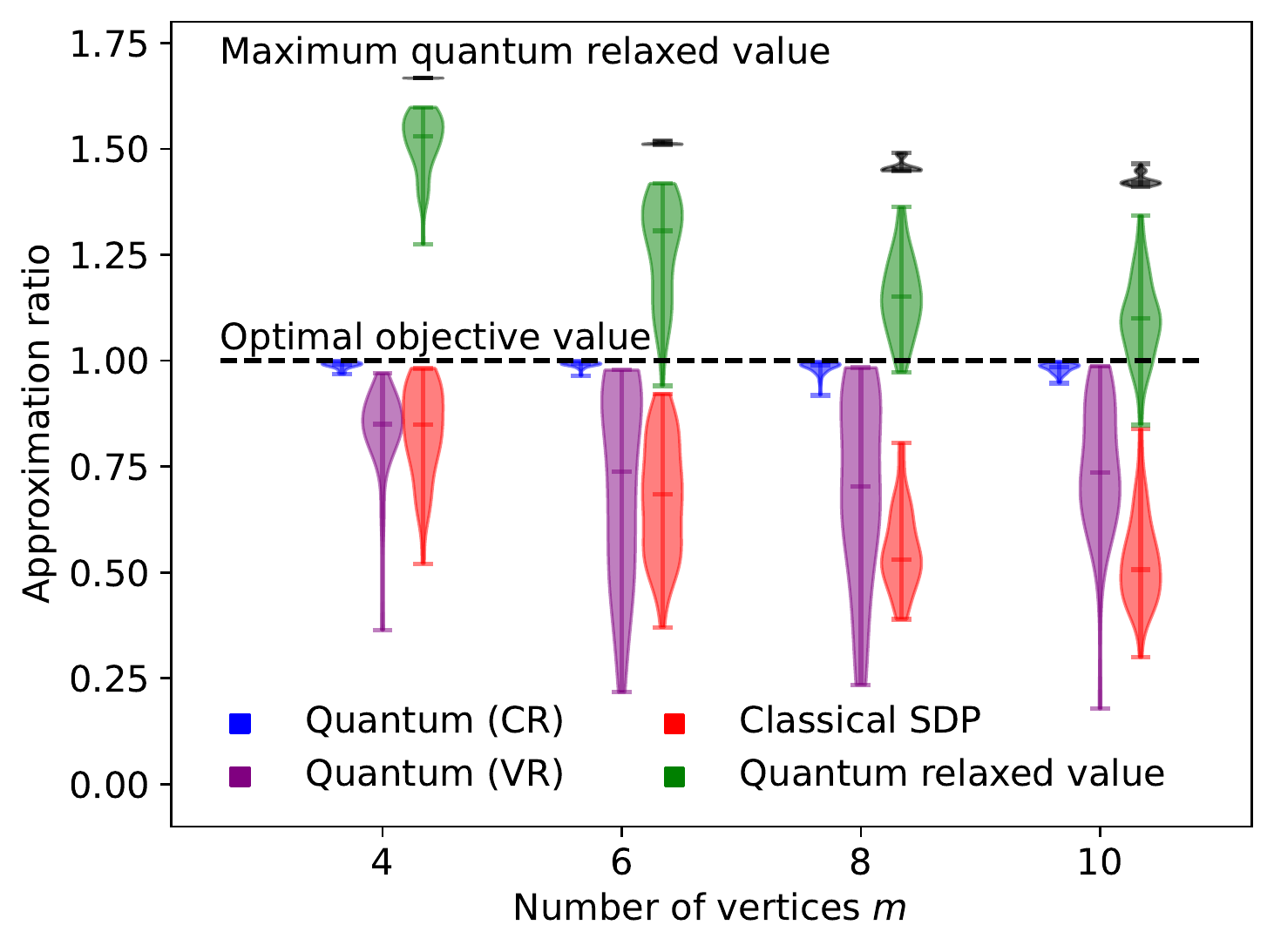}
    \includegraphics[width=0.495\textwidth]{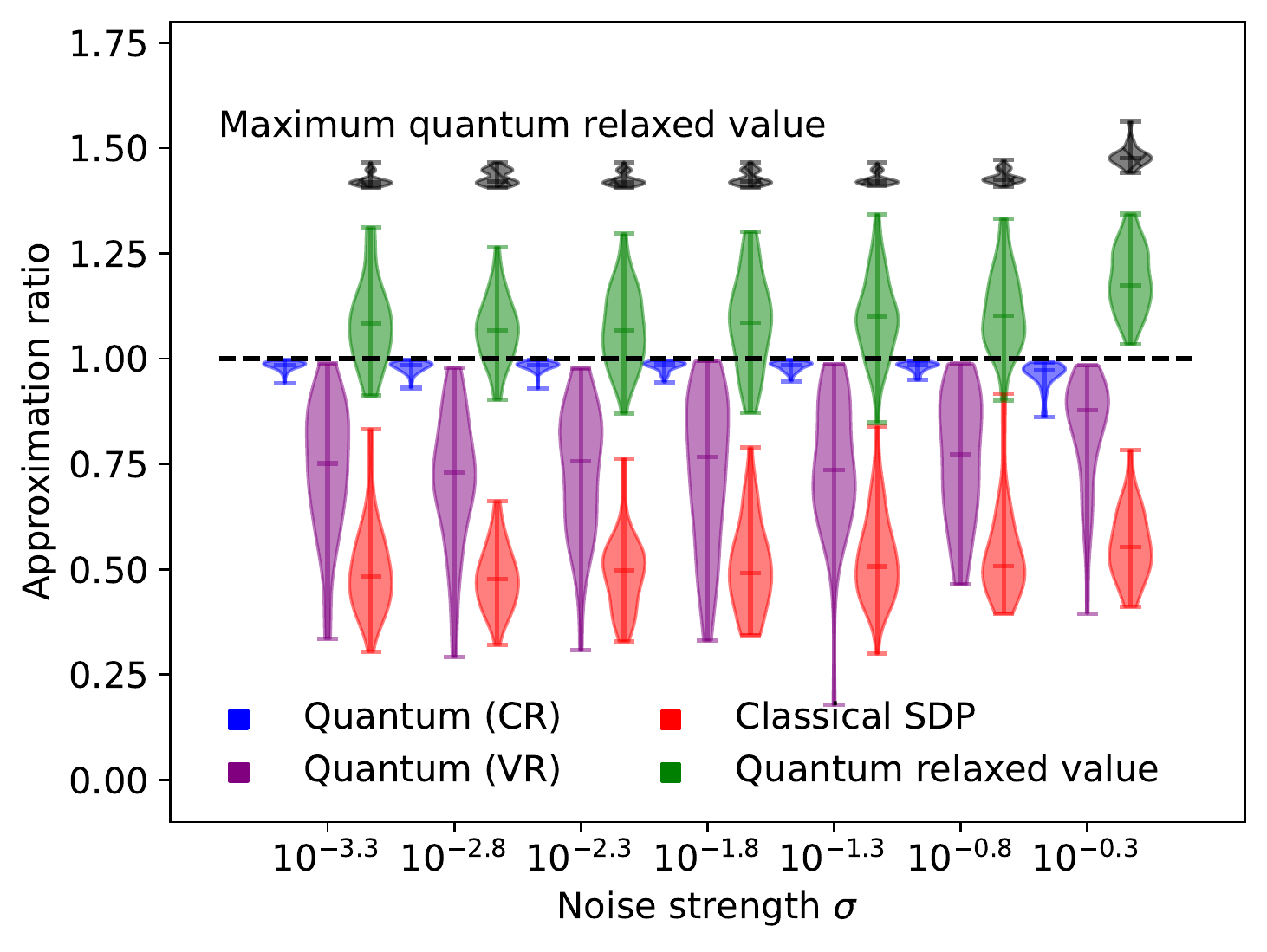}
    \caption{Approximation ratios for solutions obtained from rounding the ``adiabatically'' evolved state $\ket{\psi(T)}$ with fixed $T = 1$ for all $m, \sigma$. Problem instances and visualization is the same as in Figure~\ref{fig:eigen_group-sync}. We also include the maximum eigenvalue of the relaxed Hamiltonian and the energy of the prepared unrounded state, to demonstrate how far $\ket{\psi(T)}$ is from the exact maximum eigenvector. (Left) Varying the number of vertices $m$ in the graph. Note that the number of qubits required here is $2m$. (Right) Varying the noise strength parameter $\sigma$.}
    \label{fig:adiabatic_group-sync_violin}
\end{figure}

Then in Figure~\ref{fig:adiabatic_group-sync_violin} we plot the same 50 problem instances (per graph size/noise level) as in Figure~\ref{fig:eigen_group-sync}, but using the quasi-adibatically prepared state $\ket{\psi(T)}$ where we have fixed $T = 1$ for all graph sizes. The classical SDP results are the same as in Figure~\ref{fig:adiabatic_group-sync_violin}, and for reference we include the energy of the unrounded quantum state and the maximum eigenvalue of the relaxed Hamiltonian (normalized with respect to the optimal objective value). Qualitatively, we observe features similar to those seen in Figure~\ref{fig:adiabatic_group-sync}. Namely, although the annealing schedule is too fast to prepare a close approximation to the maximum eigenstate, the rounded solutions (using the $\conv\SO(n)$-based protocol) consistently have high approximation ratios. Meanwhile, the vertex-rounded solutions are highly inconsistent, which reflects the highly fluctuating behavior seen in Figure~\ref{fig:adiabatic_group-sync}.

\section{Discussion and future work}\label{sec:discussion}
In this paper we have developed a quantum relaxation for a quadratic program over orthogonal and rotation matrices, known as an instance of the little noncommutative Grothendieck problem. The embedding of the classical objective is achieved by recognizing an intimate connection between the geometric-algebra construction of the orthogonal group and the structure of quantum mechanics, in particular the formalism of fermions in second quantization. From this perspective, the determinant condition of $\SO(n)$ is succinctly captured by a simple linear property of the state---its parity---and the convex bodies $\conv\Orth(n)$ and $\conv\SO(n)$ (relevant to convex relaxations of optimization over orthogonal matrices) are completely characterized by density operators on $n$ and $n - 1$ qubits, respectively. Recognizing that the reduced state on each vertex therefore corresponds to an element of this convex hull, we proposed vertex-marginal rounding which classically rounds the measured one-body reduced density matrix of each vertex.

We additionally showed that these convex hulls are characterized by density operators on $2n$ and $2(n-1)$ qubits as well, where the linear functionals defining this PSD lift are the Hamiltonian terms appearing in our quantum relaxation. This insight enables our second proposed rounding scheme, $\conv G$-based edge rounding, which is inspired by the fact that the a quantum Gram matrix $\mathcal{M}$ can be constructed from the expectation values of the quantum state which obeys the same properties as the classical SDP of Saunderson~\emph{et al.}~\cite{saunderson2014semidefinite}. Numerically we observe that this approach to quantum rounding is significantly more accurate and consistent than vertex rounding, and it consistently achieves larger approximation ratios than the basic SDP relaxation. However, we are severely limited by the exponential scaling of classically simulating quantum states;~further investigations would be valuable to ascertain the empirical performance of these ideas at larger scales.

The primary goal of this work was to formulate the problem of orthogonal-matrix optimization into a familiar quantum Hamiltonian problem, and to establish the notion of a quantum relaxation for such optimization problems over continuous-valued decision variables. A clear next step is to prove nontrivial approximation ratios from our quantum relaxation. If such approximation ratios exceed known guarantees by classical algorithms, for example on certain types of graphs, then this would potentially provide a quantum advantage for a class of applications not previously considered in the quantum literature. We have proposed one standard, realistically preparable class of states---quasi-adiabatic time evolution---but a variety of energy-optimizing ansatze exist in the literature, especially considering that the constructed Hamiltonian is an interacting-fermion model. From this perspective, it would also be interesting to see if a classical many-body method can produce states which round down to high-quality approximations, even heuristically. Such an approach would constitute a potential example of a quantum-inspired classical algorithm.

From a broader perspective, the quantum formalism described here may also provide new insights into the computational hardness of the classical problem. First, the NP-hard thresholds for Problem~\eqref{eq:LNCG} are not currently known. However, by establishing the classical problem as an instance of Gaussian product state optimization on the many-body Hamiltonian, it may be possible to import tools from quantum computational complexity to study the classical problem. This idea also applies to the more general instances of noncommutative Grothendieck problems,
\begin{equation}\label{eq:NCG}
    \max_{U, V \in \Orth(N)} \sum_{i,j,k,l \in [N]} T_{ijkl} U_{ij} V_{kl},
\end{equation}
where the $N \times N \times N \times N$ tensor $T$ specifies the problem input. It is straightforward to apply our quantum relaxation construction to this problem, yielding a $2N$-qubit Hamiltonian whose terms are of the form $P_{ij} \otimes P_{kl}$. While Bri\"{e}t \emph{et al.}~\cite{briet2017tight} showed that the NP-hardness threshold of approximating this problem is $1/2$, it remains an open problem to construct an algorithm which is guaranteed to achieve this approximation ratio.

Although we have provided new approximation ratios for the instance of Problem~\eqref{eq:LNCG} over $\SO(n)$, it is unclear precisely how much harder the $\SO(n)$ problem is compared to the $\Orth(n)$ problem. The work by Saunderson \emph{et al.}~\cite{saunderson2015semidefinite} establishes a clear distinction between the representation sizes required for $\conv\Orth(n)$ and $\conv\SO(n)$, and this paper has connected this structure to properties of quantum states on $n$ qubits. However this does not yet establish a difference of hardness for the corresponding quadratic programs. Again it would be interesting to see if the tools of quantum information theory can be used to further understand this classical problem. For example, one might study the NP-hardness threshold of Problem~\eqref{eq:NCG} where instead $U, V \in \SO(N)$ and leverage the quantum (or equivalently, Clifford-algebraic) representation of $\SO(N)$. In such a setting, the size of the problem is given by a single parameter $N$ and so the exponentially large parametrization of $\conv\SO(N)$ appears to signify a central difficulty of this problem.

We note that it is straightforward to extend our quantum relaxation to the unitary groups $\U(n)$ and $\SU(n)$, essentially by doubling the number of qubits per vertex via the inclusions $\U(n) \subset \Orth(2n)$ and $\SU(n) \subset \SO(2n)$. However this is likely an inefficient embedding, since the $n$-qubit Majorana operators already form a representation of $\Cl(2n)$. It may therefore be possible to encode complex-valued matrices via a complexification of $Q$, using the same amount of quantum space. It is interesting to note that Bri\"{e}t \emph{et al.}~\cite{briet2017tight} in fact utilize  a ``complex extension'' of Clifford algebras when considering Problem~\eqref{eq:NCG} over the unitary group, although the usage is different from ours.

\begin{acknowledgments}
    The authors thank Ryan Babbush, Bill Huggins, Robin Kothari, Jarrod McClean, Chaithanya Rayudu, and Jun Takahashi for helpful discussions and feedback on the manuscript. AZ thanks Akimasa Miyake for support.
\end{acknowledgments}

\printbibliography

\appendix

\section{\label{sec:clifford}Clifford algebras and the orthogonal group}
In this appendix we review the key components for constructing the orthogonal and special orthogonal groups from a Clifford algebra. Our presentation of this material broadly follows Refs.~\cite{atiyah1964clifford,saunderson2015semidefinite}.

The Clifford algebra $\Cl(n)$ of $\R^n$ is a $2^n$-dimensional real vector space, equipped with an inner product $\langle \cdot, \cdot \rangle : \Cl(n) \times \Cl(n) \to \R$ and a multiplication operation satisfying the anticommutation relation
\begin{equation}\label{eq:anticommutator_cl_n}
    e_i e_j + e_j e_i = -2 \delta_{ij} \openone,
\end{equation}
where $\{e_1, \ldots, e_n\}$ is an orthonormal basis of $\R^n$ and $\openone$ is the multiplicative identity of the algebra. The basis elements $e_i$ are called the generators of the Clifford algebra, in the sense that they generate all other basis vectors of $\Cl(n)$ as
\begin{equation}
    e_I \coloneqq e_{i_1} \cdots e_{i_k}, \quad I = \{i_1, \ldots, i_k\} \subseteq [n].
\end{equation}
By convention we order the indices $i_1 < \cdots < i_k$, and the empty set corresponds to the identity, $e_\varnothing = \openone$. Taking all subsets $I \subseteq [n]$ and extending the inner product definition from $\R^n$ to $\Cl(n)$, it follows that $\{e_I \mid I \subseteq [n]\}$ is an orthonormal basis with $2^n$ elements. Specifically, we can write any element $x \in \Cl(n)$ as
\begin{equation}
    x = \sum_{I \subseteq [n]} x_I e_I
\end{equation}
with each $x_I \in \R$, and the inner product on $\Cl(n)$ is\footnote{Equipping an inner product to the vector representation of $\Cl(n)$ elements is achieved using the fact that algebra elements square to a multiple of the identity.}
\begin{equation}
    \langle x, y \rangle = \sum_{I \subseteq [n]} x_I y_I,
\end{equation}
where $y = \sum_{I \subseteq [n]} y_I e_I$. Hence $\Cl(n)$ is isomorphic as a Hilbert space to $\R^{2^n}$.

Now we show how to realize the orthogonal group $\Orth(n)$ from this algebra. First observe the inclusion $\R^n = \spn\{e_i \mid i \in [n]\} \subset \Cl(n)$. We shall identify the sphere $S^{n-1} \subset \R^n$ as all $u \in \R^n$ satisfying $\langle u, u \rangle = 1$. We then define the Pin group as all possible products of $S^{n-1}$ elements:
\begin{equation}\label{eq:pin}
    \Pin(n) \coloneqq \{u_1 \cdots u_{k} \mid u_1, \ldots, u_{k} \in S^{n-1}, 0 \leq k \leq n\}.
\end{equation}
It is straightforward to check that this is indeed a group. Each $x \in \Pin(n)$ is also normalized, $\langle x, x \rangle = 1$. In fact, an equivalent definition of this group is all elements $x \in \Cl(n)$ satisfying $x \overline{x} = \openone$, where conjugation $\overline{x}$ is defined from the linear extension of
\begin{equation}
    \overline{e_I} \coloneqq (-1)^{|I|} e_{i_k} \cdots e_{i_1}.
\end{equation}

The Pin group is a double cover of $\Orth(n)$, which can be seen from defining a quadratic map $Q : \Cl(n) \to \R^{n \times n}$. This map arises from the so-called twisted adjoint action, introduced by Atiyah \emph{et al.}~\cite{atiyah1964clifford}:\footnote{Saunderson \emph{et al.}~\cite{saunderson2015semidefinite} consider the standard adjoint action, which is sufficient for describing rotations. However, the ``twist'' due to $\alpha$ is necessary to construct arbitrary orthogonal transformations.}
\begin{equation}
    v \mapsto \alpha(x) v \overline{x}, \quad x, v \in \Cl(n),
\end{equation}
where the linear map $\alpha : \Cl(n) \to \Cl(n)$ is the parity automorphism, defined by linearly extending
\begin{equation}
    \alpha(e_I) \coloneqq (-1)^{|I|} e_I.
\end{equation}
Then for any $x \in \Cl(n)$, the linear map $Q(x) : \R^n \to \R^n$ is defined as
\begin{equation}\label{eq:quad_map_clifford}
    Q(x)(v) \coloneqq \pi_{\R^n}(\alpha(x) v \overline{x}) \quad \forall v \in \R^n,
\end{equation}
where $\pi_{\R^n}$ is the projection from $\Cl(n)$ onto $\R^n$. To show that $Q(\Pin(n)) = \Orth(n)$, it suffices to recognize that, for any $u \in S^{n-1}$, $\alpha(u) v \overline{u} \in \R^n$ is the reflection of the vector $v \in \R^n$ across the hyperplane normal to $u$. To see this, first observe that $uv + vu = -2 \langle u, v \rangle \openone$, which follows from Eq.~\eqref{eq:anticommutator_cl_n} by linearity. Then
\begin{equation}
\begin{split}
    \alpha(u) v \overline{u} &= uvu\\
    &= (-vu - 2 \langle u, v \rangle \openone) u\\
    &= v - 2 \langle u, v \rangle u,
\end{split}
\end{equation}
which is precisely the elementary reflection as claimed. By the Cartan--Dieudonn{\'e} theorem, one can implement any orthogonal transformation on $\R^n$ by composing $k \leq n$ such reflections about arbitrary hyperplanes $u_1, \ldots, u_k$~\cite{gallier2011geometric}. This characterization coincides precisely with the definition of the Pin group provided in Eq.~\eqref{eq:pin}, through the composition of the linear maps $Q(u_1), \ldots, Q(u_k)$ on $\R^n$. Hence for all $x \in \Pin(n)$, $Q(x)$ is an orthogonal transformation on $\R^n$. The double cover property follows from the fact that $Q$ is quadratic in $x$, so $Q(x) = Q(-x)$. 

The special orthogonal group arises from the subgroup $\Spin(n) \subset \Pin(n)$ containing only even-parity Clifford elements. First observe that $\Cl(n)$ is a $\Z_2$-graded algebra:
\begin{equation}
    \Cl(n) = \Cl^{0}(n) \oplus \Cl^{1}(n),
\end{equation}
where
\begin{align}
    \Cl^0(n) &\coloneqq \spn\{e_I \mid |I| \text{ even}\},\\
    \Cl^1(n) &\coloneqq \spn\{e_I \mid |I| \text{ odd}\}.
\end{align}
By a $\Z_2$ grading we mean that for each $x \in \Cl^a(n)$ and $y \in \Cl^b(n)$, their product $xy$ lies in $\Cl^{a + b \mod 2}(n)$. We say that elements in $\Cl^0(n)$ (resp., $\Cl^1(n)$) have even (resp., odd) parity. In particular, this grading implies that $\Cl^{0}(n)$ is a subalgebra, hence its intersection with the Pin group is also a group, which defines
\begin{equation}
    \Spin(n) \coloneqq \Pin(n) \cap \Cl^{0}(n) = \{u_1 \cdots u_{2k} \mid u_1, \ldots, u_{2k} \in S^{n-1}, 0 \leq k \leq \lfloor n/2 \rfloor\}.
\end{equation}
Just as the Pin group double covers $\Orth(n)$, so does the Spin group double cover $\SO(n)$. This is again a consequence of the Cartan--Dieudonn{\'e} theorem, wherein all rotations on $\R^n$ can be decomposed into an even number of (at most $n$) arbitrary reflections.

\section{\label{sec:quantum_conv_son}Convex hull of orthogonal matrices and quantum states}
Recall the following characterizations of the convex hulls:
\begin{align}
\conv\Orth(n) &= \l\{ X \in \R^{n \times n} \mid \sigma_1(X) \leq 1 \r\}\\
\conv\SO(n) &= \l\{ X \in \R^{n \times n} \mathrel{\Bigg|} \sum_{i \in [n] \setminus I} \widetilde{\sigma}_i(X) - \sum_{i \in I} \widetilde{\sigma}_i(X) \leq n - 2 \quad \forall I \subseteq [n], |I| \text{ odd} \r\}, \label{eq:convSOn_app}
\end{align}
where $\{\sigma_i(X)\}_{i \in [n]}$ and $\{\widetilde{\sigma}_i(X\}_{i \in [n]}$ are the singular values and special singular values of $X$ in descending order, respectively. Note that $\sigma_i(X) = \widetilde{\sigma}_i(X)$ for all $i \leq n-1$ and $\sigma_n(X) = \sign(\det(X)) \sigma_n(X)$.

\subsection{\label{sec:psd_lift_proof}PSD lift of $\conv\Orth(n)$ and $\conv\SO(n)$}

In this section we show that $Q(\mathcal{D}(\Cl(n))) = \conv\Orth(n)$ and $Q(\mathcal{D}(\Cl^0(n))) = \conv\SO(n)$, using the quantum formalism described in Section~\ref{sec:quantum_formalism}.

First we show that for all $X \in \conv\Orth(n)$, there exists some $\rho \in \mathcal{D}(\mathcal{H}_{2^n})$ which generates $X$, essentially by the convex extension of $Q$. Every $X \in \conv\Orth(n)$ can be expressed as a convex combination ($\sum_\mu p_\mu = 1$, $p_\mu \geq 0$) of orthogonal matrices $R_\mu \in \Orth(n)$:
\begin{equation}
    X = \sum_{\mu} p_\mu R_\mu.
\end{equation}
For each $R_\mu$ there exists some $x_\mu \in \Pin(n)$ such that $[R_\mu]_{ij} = \ev{P_{ij}}{x_\mu}$. Therefore the matrix elements of $X$ can be expressed as
\begin{equation}
    X_{ij} = \tr\l( P_{ij} \sum_{\mu} p_\mu \op{x_\mu}{x_\mu} \r) = \tr(P_{ij} \rho),
\end{equation}
where $\rho \coloneqq \sum_{\mu} p_\mu \op{x_\mu}{x_\mu} \in \mathcal{D}(\mathcal{H}_{2^n})$.

Next we show the reverse direction, that for all $\rho \in \mathcal{D}(\mathcal{H}_{2^n})$, the matrix $X \coloneqq [\tr(P_{ij} \rho)]_{i,j \in [n]}$ is an element of $\conv\Orth(n)$. Recall that $X \in \conv\Orth(n)$ if and only if $\sigma_1(X) \leq 1$. Therefore we take the singular value decomposition of $X = U \Sigma V^\T$ and, using $P_{ij} = \i \widetilde{\gamma}_i \gamma_j$, each singular value is equal to
\begin{equation}
\begin{split}
    \sigma_k(X) &= [U^\T X V]_{kk}\\
    &= \sum_{i,j \in [n]} U_{ik} \tr(\i \widetilde{\gamma}_i \gamma_j \rho) V_{jk}\\
    &= \tr(\i \mathcal{U}_{(U, \I_n)}^\dagger \widetilde{\gamma}_k \mathcal{U}_{(U, \I_n)} \mathcal{U}_{(\I_n, V)}^\dagger \gamma_k \mathcal{U}_{(\I_n, V)} \rho)\\
    &= \tr(\i \widetilde{\gamma}_k \gamma_k \rho'),
\end{split}
\end{equation}
where $\rho' \coloneqq \mathcal{U}_{(U, V)} \rho \mathcal{U}_{(U, V)}^\dagger$. Because $\i \widetilde{\gamma}_k \gamma_k$ has eigenvalues $\pm 1$, we see that $\sigma_k(X) \leq 1$ for all $k \in [n]$.

For the restriction to $\conv\SO(n)$, the first argument is essentially the same. One merely replaces $P_{ij}$ with $\widetilde{P}_{ij}$, hence $\sum_\mu p_\mu \op{x_\mu}{x_\mu} \in \mathcal{D}(\mathcal{H}_{2^{n-1}})$. For the reverse direction, we instead employ the special singular value decomposition which yields
\begin{equation}
    \widetilde{\sigma}_k(X) = \tr(\i \widetilde{\gamma}_k \gamma_k \rho'),
\end{equation}
where now $\rho' \coloneqq \mathcal{U}_{(U, \widetilde{V})} \rho \mathcal{U}_{(U, \widetilde{V})}^\dagger \in \mathcal{D}(\mathcal{H}_{2^{n-1}})$. Note that we have not projected to the even subspace this time, as it is more convenient to work in the full $n$-qubit space when handling the Gaussian unitaries. Instead, we will impose the constraint that $\rho$ only has support on the even-parity subspace, so $\tr(Z^{\otimes n} \rho) = 1$. Furthermore, because the special singular value decomposition guarantees that $\det(U)\det(\widetilde{V}) = 1$, $\mathcal{U}_{(U, \widetilde{V})}$ is parity preserving so that $\tr(Z^{\otimes n} \rho') = 1$ as well. Now recall that $X \in \conv\SO(n)$ if and only if
\begin{equation}\label{eq:conv-son_ineq_X}
    \sum_{k \in [n] \setminus I} \widetilde{\sigma}_k(X) - \sum_{k \in I} \widetilde{\sigma}_k(X) \leq n - 2
\end{equation}
for all subsets $I \subseteq [n]$ of odd size. By linearity,
\begin{equation}
    \sum_{k \in [n] \setminus I} \widetilde{\sigma}_k(X) - \sum_{k \in I} \widetilde{\sigma}_k(X) = \tr\l( \rho' \sum_{k \in [n]} (-1)^{z_k} Z_k \r),
\end{equation}
where $z = z_1 \cdots z_n \in \{0, 1\}^n$ is defined as $z_k = 1$ if $k \in I$ and $z_k = 0$ otherwise, and we have used the fact that $\i \widetilde{\gamma}_k \gamma_k = Z_k$. It therefore suffices to examine the spectrum of $A_z \coloneqq \sum_{k \in [n]} (-1)^{z_k} Z_k$:
\begin{equation}\label{eq:conv-son-eigval}
    A_z \ket{b} = \l( \sum_{k \in [n]} (-1)^{[z \oplus b]_k} \r) \ket{b}, \quad b \in \{0, 1\}^n,
\end{equation}
where $\oplus$ denotes addition modulo 2. As we are only interested in the subspace spanned by even-parity states, we restrict attention to the eigenvalues for which $|b| \mod 2 = 0$. Because $|I|$ is odd, so too is $|z|$, hence $|z \oplus b| \mod 2 = 1$. This implies that there must be at least one term in the sum of Eq.~\eqref{eq:conv-son-eigval} which is negative, so it can only take integer values at most $n - 2$. This establishes Eq.~\eqref{eq:conv-son_ineq_X}, hence $X \in \conv\SO(n)$.
\subsection{Relation to $\conv\SO(n)$-based semidefinite relaxation}

Here we provide details for our claim that the relaxed quantum solution obeys the same constraints as the classical SDP which uses the exponentially large representation of $\conv\SO(n)$. Recall that this relaxation can be formulated as
\begin{equation}\label{eq:conv_sdp_app}
    \max_{M \in \R^{mn \times mn}} \sum_{(u, v) \in E} \langle C_{uv}, M_{uv} \rangle \quad \text{subject to } \begin{cases}
    M \succeq 0,\\
    M_{vv} = \I_n & \forall v \in [m],\\
    M_{uv} \in \conv\SO(n) & \forall u, v \in [m].
    \end{cases}
\end{equation}
We will show that the Gram matrix $\mathcal{M}$ constructed from the measurements of a quantum state $\rho$, defined in Section~\ref{sec:quantum_gram} as
\begin{equation}\label{eq:quantum_gram_elements}
    [\mathcal{M}_{uv}]_{ij} = \begin{cases}
    \delta_{ij} & u = v,\\
    \frac{1}{n} \tr(\Gamma_{ij}^{(u, v)} \rho) & u < v,\\
    [M_{vu}]_{ji} & u > v,
    \end{cases}
\end{equation}
obeys the constraints of Eq.~\eqref{eq:conv_sdp_app}. Specifically, when the marginals of $\rho$ on each vertex are even-parity states (recall this is equivalent to replacing $\Gamma_{ij}$ with $\widetilde{\Gamma}_{ij}$), we obtain the $\conv\SO(n)$ condition, whereas when the parity of $\rho$ is not fixed then $M_{uv} \in \conv\Orth(n)$.

First, we show that $\mathcal{M}$ is positive semidefinite for all quantum states.

\begin{lemma}\label{lem:quantum_gram_matrix}
    Let $\mathcal{M} \in \R^{mn \times mn}$ be defined as in Eq.~\eqref{eq:quantum_gram_elements}. For all $\rho \in \mathcal{D}(\mathcal{H}_{2^n}^{\otimes m})$, $\mathcal{M} \succeq 0$.
\end{lemma}

\begin{proof}
We prove the statement by a sum-of-squares argument. To see where the fact of $1/n$ appears in the quantum definition of $\mathcal{M}$ above, we first construct a matrix $\mathcal{M}' \succeq 0$ which turns out to simply be $\mathcal{M}' = n\mathcal{M}$.

For each $k \in [n]$ define the Hermitian operator
\begin{equation}
    A_k = \sum_{v \in V} \sum_{i \in [n]} c_i^{(v)} P_{ik}^{(v)},
\end{equation}
where $c_i^{(v)} \in \R$ are arbitrary coefficients. Consider its square,
\begin{equation}\label{eq:A_k}
\begin{split}
    A_k^2 &= \l( \sum_{v \in V} \sum_{i \in [n]} c_i^{(v)} P_{ik}^{(v)} \r{)^2}\\
    &= \sum_{v \in V} \sum_{i, j \in [n]} c_i^{(v)} c_j^{(v)} P_{ik}^{(v)} P_{jk}^{(v)} + \sum_{\substack{u, v \in V \\ u \neq v}} \sum_{i, j \in [n]} c_i^{(u)} c_j^{(v)} P_{ik}^{(u)} \otimes P_{jk}^{(v)}.
\end{split}
\end{equation}
Note that the terms with $u \neq v$ feature the two-vertex operators as desired, while the diagonal terms of the sum contain products of the Pauli operators acting on the same vertex. Because $P_{ik} = \i \widetilde{\gamma}_i \gamma_k$, the diagonal terms reduce to (suppressing superscripts here)
\begin{equation}
\begin{split}
    \sum_{i, j \in [n]} c_i c_j P_{ik} P_{jk} &= - \sum_{i, j \in [n]} c_i c_j \widetilde{\gamma}_i \gamma_k \widetilde{\gamma}_j \gamma_k\\
    &= \sum_{i, j \in [n]} c_i c_j \widetilde{\gamma}_i \widetilde{\gamma}_j\\
    &= \sum_{i \in [n]} c_i^2 \I_{2^n} + \sum_{1 \leq i < j \leq n} c_i c_j (\widetilde{\gamma}_i \widetilde{\gamma}_j + \widetilde{\gamma}_j \widetilde{\gamma}_i)\\
    &= \l( \sum_{i \in [n]} c_i^2 \r) \I_{2^n}.
\end{split}
\end{equation}
Plugging this result into Eq.~\eqref{eq:A_k} and summing over all $k \in [n]$, we obtain
\begin{equation}
\begin{split}
    \sum_{k \in [n]} A_k^2 &= \l( \sum_{v \in V} \sum_{i, k \in [n]} |c_i^{(v)}|^2 \r) \I_{2^n} + \sum_{\substack{u, v \in V \\ u \neq v}} \sum_{i, j, k \in [n]} c_i^{(u)} c_j^{(v)} P_{ik}^{(u)} \otimes P_{jk}^{(v)}\\
    &= n \langle c, c \rangle \I_{2^n} + \sum_{\substack{u, v \in V \\ u \neq v}} \sum_{i, j \in [n]} c_i^{(u)} c_j^{(v)} \Gamma_{ij}^{(u, v)},
\end{split}
\end{equation}
where we have collected the coefficients $c_i^{(v)}$ into a vector $c \in \R^{mn}$. Similarly, if we arrange the expectation values $\tr(\Gamma_{ij}^{(u, v)} \rho)$ into a matrix $T \in \R^{mn \times mn}$ (where the $n \times n$ blocks on the diagonal are $0$), then the expectation value of the sum-of-squares operator is
\begin{equation}
\begin{split}
    \tr\l( \sum_{k \in [n]} A_k^2 \rho \r) &= n \langle c, c \rangle + \langle c, T c \rangle\\
    &= \langle c, \mathcal{M}' c \rangle
\end{split}
\end{equation}
where we have defined $\mathcal{M}' \coloneqq n\I_{mn} + T$. Because $\sum_{k \in [n]} A_k^2$ is a sum of PSD operators, its expectation value is always nonnegative, hence $\langle c, \mathcal{M}' c \rangle \geq 0$. This inequality holds for all vectors $c \in \R^{mn}$, so $\mathcal{M}' \succeq 0$ and hence $\mathcal{M} = \mathcal{M}'/n \succeq 0$ as claimed.
\end{proof}

The fact that the diagonal blocks $\mathcal{M}_{vv} = \I_n$ holds by definition. Finally, we need to show that each block of $\mathcal{M}$ lies in $\conv(\SO(n))$ when $\rho$ has even parity. A straightforward corollary of this result is that the blocks lie in $\conv\Orth(n)$ when $\rho$ does not have fixed parity. Note that in Section~\ref{sec:mixed_states} we showed that the matrix of expectation values $\tr(P_{ij} \rho_1)$ lies in $\conv\Orth(n)$ for any $n$-qubit density operator $\rho_1$, a straightforward extension of the PSD-lift representation of $\conv\SO(n)$ presented in Ref.~\cite{saunderson2015semidefinite}. Here we instead show that the matrix of expectation values $\frac{1}{n} \tr(\Gamma_{ij} \rho_2)$ for any $2n$-qubit density operator $\rho_2$ also lies in $\conv\Orth(n)$, and is an element of $\conv\SO(n)$ when $\rho_2$ has support only in the even-parity sector.

Because $\I_n \in \conv\SO(n) \subset \conv\Orth(n)$, and because these convex hulls are closed under transposition, all that remains is to prove the statement for the blocks $\mathcal{M}_{uv}$ when $u < v$. We show this by considering all density matrices on the reduced two-vertex Hilbert space.

\begin{lemma}
    Let $\rho_2 \in \mathcal{D}(\mathcal{H}_{2^n}^{\otimes 2})$. Define the $n \times n$ matrix $T$ by
    \begin{equation}
    T_{ij} \coloneqq \tr(\Gamma_{ij} \rho_2) = \tr\l( \sum_{k \in [n]} P_{ik} \otimes P_{jk} \rho_2 \r), \quad i, j \in [n].
    \end{equation}
    Then $T/n \in \conv\Orth(n)$. Furthermore, if $\tr(Z^{\otimes 2n} \rho_2) = 1$, then $M_{uv} \in \conv\SO(n)$.
\end{lemma}

\begin{proof}
    Consider the special singular value decomposition of $T = U \widetilde{\Sigma} \widetilde{V}^\T$. We can express the special singular values as
    \begin{equation}\label{eq:ssv_T}
    \begin{split}
    \widetilde{\sigma}_i(T) &= [U^\T T \widetilde{V}]_{ii}\\
    &= \sum_{j,\ell \in [n]} U_{ji} T_{j\ell} \widetilde{V}_{\ell i}\\
    &= \sum_{j,\ell \in [n]} U_{ji} \widetilde{V}_{\ell i} \sum_{k \in [n]} \tr\l( \i \widetilde{\gamma}_j \gamma_k \otimes \i \widetilde{\gamma}_\ell \gamma_k \rho_2 \r)\\
    &= \sum_{k \in [n]} \tr\l( {\mathcal{U}}_{(U, \I_n)}^\dagger \i \widetilde{\gamma}_i \gamma_k {\mathcal{U}}_{(U, \I_n)} \otimes {\mathcal{U}}_{(\widetilde{V}, \I_n)}^\dagger \i \widetilde{\gamma}_i \gamma_k {\mathcal{U}}_{(\widetilde{V}, \I_n)} \rho_2 \r)\\
    &= \sum_{k \in [n]} \tr\l( \i \widetilde{\gamma}_i \gamma_k \otimes \i \widetilde{\gamma}_i \gamma_k \rho'_2 \r),
    \end{split}
    \end{equation}
    where $\rho'_2 \coloneqq \l({\mathcal{U}}_{(U, \I_n)} \otimes {\mathcal{U}}_{(\widetilde{V}, \I_n)}\r) \rho_2 \l({\mathcal{U}}_{(U, \I_n)} \otimes {\mathcal{U}}_{(\widetilde{V}, \I_n)}\r{)^\dagger}$. The fact that $T/n \in \conv\Orth(n)$ follows immediately from the fact that the spectrum of $\i \widetilde{\gamma}_i \gamma_k \otimes \i \widetilde{\gamma}_i \gamma_k$ is $\{\pm 1\}$:
    \begin{equation}
    \begin{split}
    \sigma_1(T/n) &= \frac{1}{n} \l| \widetilde{\sigma}_1(T) \r|\\
    &= \frac{1}{n} \l| \sum_{k \in [n]} \tr\l( \i \widetilde{\gamma}_1 \gamma_k \otimes \i \widetilde{\gamma}_1 \gamma_k \rho'_2 \r) \r|\\
    &\leq \frac{1}{n} \sum_{k \in [n]} \l| \tr\l( \i \widetilde{\gamma}_1 \gamma_k \otimes \i \widetilde{\gamma}_1 \gamma_k \rho'_2 \r) \r|\\
    &\leq 1.
    \end{split}
    \end{equation}
    
    Now we examine the inclusion in $\conv\SO(n)$. Suppose that $\rho_2$ is an even-parity state. By virtue of the special singular value decomposition, we have that $\det(U\widetilde{V}^\T) = 1$ which implies that the Gaussian unitary ${\mathcal{U}}_{(U, \I_n)} \otimes {\mathcal{U}}_{(\widetilde{V}, \I_n)}$ preserves the parity of $\rho_2$:~$\tr(Z^{\otimes 2n} \rho_2) = \tr(Z^{\otimes 2n} \rho_2') = 1$. For $T/n$ to lie in $\conv\SO(n)$, the following inequality from Eq.~\eqref{eq:convSOn_app} must hold:
    \begin{equation}\label{eq:conv_ineq}
    \sum_{i \in [n] \setminus I} \widetilde{\sigma}_i(T) - \sum_{i \in I} \widetilde{\sigma}_i(T) \leq n(n - 2)
    \end{equation}
    for all subsets $I \subseteq [n]$ of odd size. To show this, first we write the left-hand side in terms of the result derived from Eq.~\eqref{eq:ssv_T}:
    \begin{equation}\label{eq:conv_T}
    \begin{split}
    \sum_{i \in [n] \setminus I} \widetilde{\sigma}_i(T) - \sum_{i \in I} \widetilde{\sigma}_i(T) &= \sum_{i \in [n] \setminus I} \sum_{k \in [n]} \tr\l( \i \widetilde{\gamma}_i \gamma_k \otimes \i \widetilde{\gamma}_i \gamma_k \rho'_2 \r) - \sum_{i \in I} \sum_{k \in [n]} \tr\l( \i \widetilde{\gamma}_i \gamma_k \otimes \i \widetilde{\gamma}_i \gamma_k \rho'_2 \r)\\
    &= \tr\l( \rho_2' \sum_{i,k \in [n]} (-1)^{z_i} \i \widetilde{\gamma}_i \gamma_k \otimes \i \widetilde{\gamma}_i \gamma_k \r),
    \end{split}
    \end{equation}
    where the string $z = z_1 \cdots z_n \in \{0, 1\}^n$ is defined as
    \begin{equation}
    z_i = \begin{cases}
    1 & i \in I,\\
    0 & i \notin I.
    \end{cases}
    \end{equation}
    Note that $|I|$ being odd implies that the Hamming weight of $z$ is also odd. To bound Eq.~\eqref{eq:conv_T} we shall seek a bound on the largest eigenvalue of the Hermitian operator
    \begin{equation}
    \sum_{i,k \in [n]} (-1)^{z_i} \i \widetilde{\gamma}_i \gamma_k \otimes \i \widetilde{\gamma}_i \gamma_k = \l( \sum_{i \in [n]} (-1)^{z_i \oplus 1} \widetilde{\gamma}_i \otimes \widetilde{\gamma}_i \r) \l( \sum_{k \in [n]} \gamma_k \otimes \gamma_k \r) \eqqcolon A_z B
    \end{equation}
    over the space of even-parity states. Here we use $\oplus$ to denote addition modulo 2.
    
    The operator $A_z B$ can in fact be exactly diagonalized in the basis of Bell states. First, observe that $[A_z, B] = 0$, which follows from the fact that $\widetilde{\gamma}_i \gamma_k = -\gamma_k \widetilde{\gamma}_i$, hence $[\widetilde{\gamma}_i \otimes \widetilde{\gamma}_i, \gamma_k \otimes \gamma_k] = 0$ for all $i, k \in [n]$. We can therefore seek their simultaneous eigenvectors, which can be determined from looking at the Jordan--Wigner representation of the Majorana operators:
    \begin{align}
    \widetilde{\gamma}_i \otimes \widetilde{\gamma}_i &= (Z_1 \cdots Z_{i-1} Y_i) \otimes (Z_1 \cdots Z_{i-1} Y_i), \label{eq:tensor_tilde_gamma}\\
    \gamma_k \otimes \gamma_k &= (Z_1 \cdots Z_{k-1} X_k) \otimes (Z_1 \cdots Z_{k-1} X_k). \label{eq:tensor_gamma}
    \end{align}
    These operators are diagonalized by the $2n$-qubit state
    \begin{equation}
    \ket{\beta(x, y)} = \bigotimes_{j=1}^n \ket{\beta(x_j, y_j)},
    \end{equation}
    where $x, y \in \{0, 1\}^n$, and the Bell state $\ket{\beta(x_j, y_j)}$ between the $j$th qubits across the two subsystems is defined as
    \begin{equation}
    \ket{\beta(x_j, y_j)} \coloneqq \frac{\ket{0} \otimes \ket{y_j} + (-1)^{x_j} \ket{1} \otimes \ket{y_j \oplus 1}}{\sqrt{2}}.
    \end{equation}
    Indeed there are $2^{2n}$ such states $\ket{\beta(x,y)}$, so they form an orthonormal basis for the $2n$ qubits. The eigenvalues of Eqs.~\eqref{eq:tensor_tilde_gamma} and \eqref{eq:tensor_gamma} can be determined by a standard computation,
    \begin{align}
    (X_j \otimes X_j) \ket{\beta(x_j, y_j)} &= (-1)^{x_j} \ket{\beta(x_j, y_j)},\\
    (Y_j \otimes Y_j) \ket{\beta(x_j, y_j)} &= (-1)^{x_j \oplus y_j \oplus 1} \ket{\beta(x_j, y_j)},\\
    (Z_j \otimes Z_j) \ket{\beta(x_j, y_j)} &= (-1)^{y_j} \ket{\beta(x_j, y_j)}.\label{eq:ZZ_beta}
    \end{align}
    Taking the appropriate products furnishes the eigenvalues of the $B$ and $A_z$ as
    \begin{align}
    \begin{split}
    B \ket{\beta(x, y)} &= \sum_{k \in [n]} (Z_1 \otimes Z_1) \cdots (Z_{k-1} \otimes Z_{k-1}) (X_k \otimes X_k) \ket{\beta(x,y)}\\
    &= \sum_{k \in [n]} (-1)^{y_1 \oplus \cdots \oplus y_{k-1}} (-1)^{x_k} \ket{\beta(x, y)},
    \end{split}\\
    \begin{split}
    A_z \ket{\beta(x, y)} &= \sum_{i \in [n]} (-1)^{z_i \oplus 1} (Z_1 \otimes Z_1) \cdots (Z_{k-1} \otimes Z_{k-1}) (Y_k \otimes Y_k) \ket{\beta(x,y)} \\
    &= \sum_{i \in [n]} (-1)^{z_i \oplus 1} (-1)^{y_1 \oplus \cdots \oplus y_{i-1}} (-1)^{x_i \oplus y_i \oplus 1} \ket{\beta(x, y)}\\
    &= \sum_{i \in [n]} (-1)^{z_i} (-1)^{y_1 \oplus \cdots \oplus y_{i}} (-1)^{x_i} \ket{\beta(x, y)}
    \end{split}
    \end{align}
    Altogether we arrive at the expression for the eigenvalues of $A_z B$,
    \begin{equation}\label{eq:AzB_eigval}
    \ev{A_z B}{\beta(x, y)} = \l( \sum_{i \in [n]} (-1)^{z_i} (-1)^{y_1 \oplus \cdots \oplus y_i} (-1)^{x_i} \r) \l( \sum_{k \in [n]} (-1)^{y_1 \oplus \cdots \oplus y_{k-1}} (-1)^{x_k} \r).
    \end{equation}
    We wish to find the largest value this can take over even-parity states. First, observe that the eigenstates $\ket{\beta(x, y)}$ have fixed parity according to
    \begin{equation}
    \ev{Z^{\otimes 2n}}{\beta(x, y)} = (-1)^{|y|},
    \end{equation}
    which follows from Eq.~\eqref{eq:ZZ_beta}. Hence we shall only consider $y$ to have even Hamming weight. Additionally recall that $z$ has odd Hamming weight, while the Hamming weight of $x$ is unrestricted.
    
    Let us denote the sums in Eq.~\eqref{eq:AzB_eigval} by
    \begin{align}
    a_z(x, y) &\coloneqq \sum_{i \in [n]} (-1)^{z_i} (-1)^{y_1 \oplus \cdots \oplus y_{i-1} \oplus y_i} (-1)^{x_i},\\
    b(x, y) &\coloneqq \sum_{k \in [n]} (-1)^{y_1 \oplus \cdots \oplus y_{k-1}} (-1)^{x_k}.
    \end{align}
    Clearly, they can be at most $n$, and this occurs whenever all the terms in their sum are positive. For $b(x, y)$, this is possible if and only if $x = p(y)$, where $p : \{0, 1\}^n \to \{0, 1\}^n$ stores the parity information of the $(k-1)$-length substring of its input into the $k$th bit of its output:
    \begin{equation}
    [p(y)]_k \coloneqq y_1 \oplus \cdots \oplus y_{k-1}.
    \end{equation}
    For notation clarity we point out that $[p(y)]_1 = 0$ and $[p(y)]_2 = y_1$. The forward direction, $b(p(y), y) = n$, is clear by construction. The reverse direction, that $b(x, y) = n$ implies $x = p(y)$, follows from the bijectivity of modular addition.
    
    Plugging this value of $x = p(y)$ into $a_z(x, y)$ yields
    \begin{equation}\label{eq:sum_i}
    a_z(p(y), y) = \sum_{i \in [n]} (-1)^{z_i} (-1)^{y_i} = \sum_{i \in [n]} (-1)^{[z \oplus y]_i}.
    \end{equation}
    Because $y$ has even and $z$ has odd Hamming weight, their sum $y \oplus z$ must have odd Hamming weight. Therefore at least one term in Eq.~\eqref{eq:sum_i} must be negative, implying that $a_z(p(y), y) \leq n - 2$. It follows that
    \begin{equation}
    \ev{A_z B}{\beta(p(y), y)} = a_z(p(y), y) b(p(y), y) \leq (n-2) n.
    \end{equation}
    We now show that no other assignment of $(x, y)$ can exceed this bound. Recall that $b(x, y) = n$ if and only if $x = p(y)$. Thus any other choice of $x$ necessarily returns a smaller value of $b(x, y)$. Because sums of $\pm 1$ cannot yield $n - 1$, the next largest value would be $b(x,y) = n - 2$. However we can always trivially bound $a_z(x, y) \leq n$ for all $x, y$. This implies that such a choice of $x$ for which $b(x, y) = n - 2$ (whatever it is) also cannot provide a value of $a_z(x, y) b(x, y)$ exceeding $n(n-2)$.
    
    Note that there is another assignment that saturates the upper bound, which is simply considering a global negative sign in front of both products. Specifically, let $x = p(y) \oplus 1^n$. In this case $b(x,y) = -n$ and $a_z(x, y) \geq -(n-2)$, yielding the same bound $a_z(x,y) b(x, y) \leq n(n-2)$.

    To conclude, we use the fact that the maximum eigenvalue of $A_z B$ in the even-parity subspace is $n(n-2)$ for any odd-weight $z$ (equiv., any odd-size $I \subseteq [n]$). This bounds the value of Eq.~\eqref{eq:conv_T} by $n(n-2)$, hence validating the inequality of Eq.~\eqref{eq:conv_ineq}. Thus $T/n \in \conv\SO(n)$ whenever $\rho_2$ has even parity.
\end{proof}

As usual, replacing the operators $\Gamma_{ij}$ with $\widetilde{\Gamma}_{ij}$ is equivalent to enforcing the even-parity constraint. In fact, since $\widetilde{\Gamma}_{ij} = \sum_{k \in [n]} \widetilde{P}_{ik} \otimes \widetilde{P}_{jk}$, the equivalent constraint involves the reduced single-vertex marginals, $\tr(Z^{\otimes n} \otimes \I_{2^n} \rho_2) = \tr(\I_{2^n} \otimes Z^{\otimes n} \rho_2) = 1$, rather than the entire $2n$-qubit Hilbert space. Of course, if both single-vertex parity constraints are satisfied, then the two-vertex constraint automatically follows.

\section{\label{sec:even_subspace}Details for working in the even-parity subspace}
First we provide an expression for $n$-qubit Pauli operators projected to $\Cl^0(n)$. Let
\begin{equation}
    A \coloneqq W_1 \otimes \cdots \otimes W_n, \quad W_i \in \{\I, X, Y, Z\}.
\end{equation}
A straightforward calculation yields the conditional expression:
\begin{equation}\label{eq:projected_paulis}
    \widetilde{A} = \Pi_0 A \Pi_0^\T = \begin{cases}
    0 & \text{if } [A, Z^{\otimes n}] \neq 0\\
    \begin{cases}
    W_2 \otimes \cdots \otimes W_n & \text{if } W_1 = \I, X\\
    \i (W_2 Z) \otimes \cdots \otimes (W_n Z) & \text{if } W_1 = Y\\
    (W_2 Z) \otimes \cdots \otimes (W_n Z) & \text{if } W_1 = Z.
    \end{cases} & \text{if } [A, Z^{\otimes n}] = 0.
    \end{cases}
\end{equation}
Notably, if $A$ does not commute with the parity operator then $\widetilde{A} = 0$.

Now we generalize from the main text, defining the operator
\begin{equation}
    \Pi_k \coloneqq \frac{1}{\sqrt{2}} \l( \bra{+} \otimes \I_2^{\otimes (n-1)} + (-1)^k \bra{-} \otimes Z^{\otimes (n-1)} \r), \quad k \in \{0, 1\},
\end{equation}
which is the projector $\Cl(n) \to \Cl^k(n)$. These operators obey
\begin{align}
    \Pi_k \Pi_k^\T &= \I_{2^{n-1}}, \label{eq:PPT}\\
    \Pi_k^\T \Pi_k &= \frac{\I_{2^n} + (-1)^k Z^{\otimes n}}{2}.\label{eq:PTP}
\end{align}

Given a state $\widetilde{\rho} \in \mathcal{D}(\mathcal{H}_{2^{n-1}})$, the expectation values of $\widetilde{P}_{ij}$ satisfy
\begin{equation}
\begin{split}
    \tr(\widetilde{P}_{ij} \widetilde{\rho}) &= \tr(P_{ij} \Pi_0^\T \widetilde{\rho} \Pi_0)\\
    &= \tr(P_{ij} \rho_0),
\end{split}
\end{equation}
where $\rho_0 \coloneqq \Pi_0^\T \widetilde{\rho} \Pi_0 \in \mathcal{D}(\mathcal{H}_{2^n})$ has only support on even-parity computational basis states. By Eq.~\eqref{eq:PPT} we can ``invert'' this relation in the following sense:~given a state $\rho \in \mathcal{D}(\mathcal{H}_{2^n})$ which only has support on the even subspace, its $(n-1)$-qubit representation is $\Pi_0 \rho \Pi_0^\T \in \mathcal{D}(\mathcal{H}_{2^{n-1}})$.

This translation is useful when using the fermionic interpretation of the $P_{ij}$ but we wish to work directly in the $(n-1)$-qubit subspace. For example, suppose we wish to prepare a product of Gaussian states in a quantum computer (as is done in Section~\ref{sec:asp_numerics} to prepare the initial state of the quasi-adiabatic evolution). It is well known how to compile linear-depth circuits for this task~\cite{PhysRevApplied.9.044036}, however this is within the standard $n$-qubit representation. Here we show how to translate those circuits into the $(n-1)$-qubit representation under $\Pi_0$. In fact, these techniques apply to any sequence of gates which commute with the parity operator $Z^{\otimes n}$.

Let $\ket{\psi} \coloneqq V \ket{0^n} \in \mathcal{H}_{2^n}$, where the circuit $V$ is constructed from $L$ gates, $V = V_L \cdots V_1$. We can assume without loss of generality that the initial state is the vacuum $\ket{0^n}$ and that all gates $V_\ell$ commute with the parity operator, as otherwise $\ket{\psi}$ would not lie in the even subspace of $\mathcal{H}_{2^n}$.\footnote{While it is possible to have an even number of gates which anticommute with the parity operator, for simplicity we assume that the circuit has been compiled such that each $V_\ell$ preserves parity.} For example, parity-preserving Gaussian unitaries can be decomposed into single- and two-qubit gates of the form $e^{-\i \theta Z_i}$, $e^{-\i \theta X_i X_{i+1}}$. We wish to obtain a circuit description for preparing the state $\Pi_0 \ket{\psi} \in \mathcal{H}_{2^{n-1}}$. By Eq.~\eqref{eq:PTP}, $\Pi_0^\T \Pi_0 \ket{0^n} = \ket{0^n}$, and furthermore $\Pi_0 \ket{0^n} = \ket{0^{n-1}}$ is the $(n-1)$-qubit representation of the vacuum. Therefore
\begin{equation}
\begin{split}
    \Pi_0 \ket{\psi} &= \Pi_0 V \ket{0^n}\\
    &= \Pi_0 V_L \cdots V_1 \Pi_0^\T \Pi_0 \ket{0^n}\\
    &= \Pi_0 V_L \cdots V_1 \Pi_0^\T \ket{0^{n-1}}.
\end{split}
\end{equation}
Because $\Pi_0$ is not unitary ($\Pi_0^\T$ is merely an isometry), we cannot simply insert terms like $\Pi_0^\T \Pi_0$ in between each gate. However, observe that if each $V_\ell$ preserves parity, then they can be block diagonalized into the even and odd subspaces of $\mathcal{H}_{2^n}$,
\begin{equation}\label{eq:V_l_block}
    V_\ell = \begin{bmatrix}
    V_{\ell,0} & 0\\
    0 & V_{\ell,1}
    \end{bmatrix},
\end{equation}
where $V_{\ell,k} \coloneqq \Pi_k V_\ell \Pi_k^\T$ are $2^{n-1}$-dimensional unitary matrices. Thus
\begin{equation}
    V = \begin{bmatrix}
    V_{L,0} \cdots V_{1,0} & 0\\
    0 & V_{L,1} \cdots V_{1,1}
    \end{bmatrix},
\end{equation}
and conjugation by the projector $\Pi_0$ precisely extracts the first block of this matrix:
\begin{equation}
    \Pi_0 V \Pi_0^\T = V_{L,0} \cdots V_{1,0}.
\end{equation}
This sequence of gates is what we wish to implement on the physical $(n-1)$-qubit register. When the gates $V_\ell$ take the form
\begin{equation}
    V_\ell = e^{-\i \theta A_\ell}
\end{equation}
for some $n$-qubit Pauli operator $A_\ell$, then
\begin{equation}\label{eq:V_l0}
\begin{split}
    V_{\ell,0} &= \Pi_0 V_\ell \Pi_0^\T\\
    &= \Pi_0 \l( \I_{2^n} \cos\theta - \i A_\ell \sin\theta \r) \Pi_0^\T\\
    &= \I_{2^{n-1}} \cos\theta - \i \widetilde{A}_\ell \sin\theta\\
    &= e^{-\i \theta \widetilde{A}_\ell},
\end{split}
\end{equation}
where $\widetilde{A}_\ell \coloneqq \Pi_0 A_\ell \Pi_0^\T$. Note that this calculation assumes that $V_\ell$ commutes with parity, hence $[A_\ell, Z^{\otimes n}] = 0$, which guarantees that $\widetilde{A}_\ell$ is unitary and Hermitian and hences furnishes the final line of Eq.~\eqref{eq:V_l0}. (If they did not commute then the decomposition of Eq.~\eqref{eq:V_l_block} would not be valid to begin with.)

\section{Small $n$ examples}
In this appendix we write down the Pauli operators $P_{ij}$ and $\widetilde{P}_{ij}$ for small but relevant values of $n$, to provide the reader with some concerte examples of how to construct the LNCG Hamiltonian.

\subsection{The Ising model from the $\Orth(1)$ setting}

As a warm-up we first demonstrate that the LNCG Hamiltonian reduces to the Ising formulation of the commutative combinatorial optimization problem when considering $G = \Orth(1)$. Each local Hilbert space $\mathcal{H}_d$ with $d = 2$ is simply a qubit, and we only have to consider a single Pauli operator on each local qubit,
\begin{equation}
    P_{11} = \i \widetilde{\gamma}_1 \gamma_1 = Z.
\end{equation}
It then follows that the LNCG interaction terms are
\begin{equation}
    \Gamma_{11}^{(u, v)} = Z^{(u)} \otimes Z^{(v)},
\end{equation}
and so the full Hamiltonian acting on $\mathcal{H}_2^{\otimes |V|}$ is indeed the classical Ising Hamiltonian,
\begin{equation}
    H = \sum_{(u, v) \in E} C_{uv} Z^{(u)} \otimes Z^{(v)},
\end{equation}
with weights $C_{uv} \in \R$. It is also instructive to write down the elements of $\Pin(1)$ as quantum states. The Clifford algebra $\Cl(1)$ is spanned by $e_1$ and $e_\varnothing = \openone$, with the only elements of $S^0$ being $\pm e_1$. Therefore, taking all possible products of elements in $S^0$ (including the empty product), we arrive at
\begin{equation}
    \Pin(1) = \{\pm e_\varnothing, \pm e_1\},
\end{equation}
which corresponds to the qubit computational basis states $\{\ket{0}, \ket{1}\}$ (up to global phases), as expected. Indeed, one sees that the mappings $Q(e_\varnothing) = \ev{Z}{0} = 1$ and $Q(e_1) = \ev{Z}{1} = -1$ fully cover $\Orth(1)$.

\subsection{The projected operators of the $\SO(3)$ setting}

As $\SO(3)$ is arguably the most ubiquitous group for physical applications, for reference we explicitly write down its Pauli operators under the projection to $\Cl^0(3) \cong \mathcal{H}_4$. As seen by the dimension of this Hilbert space, only two qubits per variable (for a total of $2m$ qubits) are required to represent this problem. Using Eq.~\eqref{eq:projected_paulis} we have
\begin{equation}
    \begin{bmatrix}
    \widetilde{P}_{11} & \widetilde{P}_{12} & \widetilde{P}_{13}\\
    \widetilde{P}_{21} & \widetilde{P}_{22} & \widetilde{P}_{23}\\
    \widetilde{P}_{31} & \widetilde{P}_{32} & \widetilde{P}_{33}
    \end{bmatrix} =
    \begin{bmatrix}
    Z_1 Z_2 & -X_1 & -Z_1 X_2\\
    X_1 Z_2 & Z_1 & -X_1 X_2\\
    X_2 & - Y_1 Y_2 & Z_2
    \end{bmatrix}.
\end{equation}

\section{\label{sec:qma-hardness_proof}QMA-hardness of the relaxed Hamiltonian}
Here we show that the relaxed LNCG Hamiltonian is QMA-hard for both $\Orth(n)$ and $\SO(n)$, for all $n \geq 2$. Note that this hardness result does not hold for $n = 1$, since $\Orth(1)$ corresponds precisely to the classic NP-complete problem of combinatorial optimization, while $\SO(1)$ is the trivial group.

We prove this by reducing to an instance of the XY model. Placing $m$ qubits on a graph $G = (V, E)$, this model is defined as
\begin{equation}
    H_{XY} \coloneqq \sum_{(u,v) \in E} \alpha_{uv} \l( X^{(u)} X^{(v)} + Y^{(u)} Y^{(v)} \r),
\end{equation}
where $\alpha_{uv} \in \R$ are arbitrary. Deciding the ground-state energy of this Hamiltonian to within inverse-polynomial precision is QMA-complete~\cite{cubitt2016complexity}. We will show that the LNCG Hamiltonian for both $\Orth(2)$ and $\SO(2)$ can be reduced to $H_{XY}$. Because $\Orth(2)$ and $\SO(2)$ can be trivially embedded into $\Orth(n)$ and $\SO(n)$ respectively, for any $n \geq 2$, this proves the claim.

To begin, recall the definitions of the relaxed Hamiltonians.
\begin{align}
    H_{\Orth(2)} &= \sum_{(u, v) \in E} \sum_{i, j \in [2]} [C_{uv}]_{ij} \sum_{k \in [2]} P_{ik}^{(u)} \otimes P_{jk}^{(v)},\\
    H_{\SO(2)} &= \sum_{(u, v) \in E} \sum_{i, j \in [2]} [C_{uv}]_{ij} \sum_{k \in [2]} \widetilde{P}_{ik}^{(u)} \otimes \widetilde{P}_{jk}^{(v)}.
\end{align}

\begin{theorem}
    Let $G = (V, E)$ be a graph with $m$ vertices. For any collection of $\alpha_{uv} \in \R$, where $(u, v) \in E$, there exists an instance of $H_{\Orth(2)}$ and $H_{\SO(2)}$ which are equivalent to $H_{XY}$.
\end{theorem}

\begin{proof}
    First, consider the $\Orth(2)$ case. From Eq.~\eqref{eq:P_ij_summary},
    \begin{equation}
        \begin{bmatrix}
            P_{11} & P_{12}\\
            P_{21} & P_{22}
        \end{bmatrix} = \begin{bmatrix}
            Z_1 & -X_1 X_2\\
            -Y_1 Y_2 & Z_2
        \end{bmatrix}
    \end{equation}
    are the two-qubit operators defined on each vertex $v \in V$. Defining an instance of $C_{uv}$ as
    \begin{equation}\label{eq:qma-hard_encoding_Cuv}
        C_{uv} = \begin{bmatrix}
            \alpha_{uv} & 0\\
            0 & 0
        \end{bmatrix},
    \end{equation}
    we get
    \begin{equation}
        H_{\Orth(2)} = \sum_{(u, v) \in E} \alpha_{uv} \l( Z_1^{(u)} Z_1^{(v)} + X_1^{(u)} X_2^{(u)} X_1^{(v)} X_2^{(v)} \r).
    \end{equation}
    To see that this Hamiltonian is equivalent to $H_{XY}$ up to a degeneracy of the spectrum (because $H_{XY}$ is over $m$ qubits, while $H_{\Orth(2)}$ is over $2m$ qubits), observe that on each vertex we can apply the two-qubit Clifford unitary $\CNOT \cdot (U \otimes \I_2)$, where $U \coloneqq \frac{1}{\sqrt{2}}(Y + Z)$. This transformation on each vertex achieves
    \begin{equation}
    \begin{split}
        Z_1 &\mapsto Y_1,\\
        X_1 X_2 &\mapsto -X_1,
    \end{split}
    \end{equation}
    hence demonstrating the reduction.
    
    Next we show the $\SO(2)$ reduction. The projected single-qubit operators per vertex can be determined via Eq.~\eqref{eq:projected_paulis}:
    \begin{equation}
        \begin{bmatrix}
            \widetilde{P}_{11} & \widetilde{P}_{12}\\
            \widetilde{P}_{21} & \widetilde{P}_{22}
        \end{bmatrix} = \begin{bmatrix}
            Z & -X\\
            X & Z
        \end{bmatrix}.
    \end{equation}
    Setting $C_{uv}$ as before, Eq.~\eqref{eq:qma-hard_encoding_Cuv}, we get\footnote{One can also make the choice $C_{uv} = \frac{1}{2} \alpha_{uv} \I_2$.}
    \begin{equation}
        H_{\SO(2)} = \sum_{(u, v) \in E} \alpha_{uv} \l( Z^{(u)} Z^{(v)} + X^{(u)} X^{(v)} \r).
    \end{equation}
    Relabeling $Z \mapsto Y$ via the local Clifford $U$ furnishes the claim.
\end{proof}

\section{\label{sec:son_approx_ratio}Classical approximation ratio for $\SO(n)$}
Approximation ratios for the rounded solution of the classical semidefinite relaxation of Problem~\eqref{eq:LNCG} were obtained in Ref.~\cite{bandeira2016approximating} for the cases $G = \Orth(n)$ and $\U(n)$. However, no such approximation ratios were derived for the case of $\SO(n)$. Here we adapt the argument of Ref.~\cite{bandeira2016approximating} to this setting, wherein the rounding algorithm performs the special singular value decomposition to guarantee that the rounded solutions have unit determinant. As we will see, this feature results in a approximation ratio for the classical semidefinite program over $\SO(n)$ that is strictly worse than the previously studied $\Orth(n)$ case.

Recall that the semidefinite relaxation of Problem~\eqref{eq:LNCG} can be formulated as
\begin{equation}\label{eq:LNCG_SDP_relax}
    \max_{X_1, \ldots, X_m \in \R^{n \times mn}} \sum_{(u, v) \in E} \langle C_{uv}, X_u X_v^\T \rangle \quad \text{subject to} \quad X_v X_v^\T = \I_n.
\end{equation}
To round the relaxed solution back into the feasible space of orthogonal matrices, Ref.~\cite{bandeira2016approximating} proposes the following randomized algorithm with a guarantee on the approximation ratio.

\begin{theorem}[{\cite[Theorem 4]{bandeira2016approximating}}]\label{thm:O(n)_approx_ratio}
    Let $X_1, \ldots, X_n \in \R^{n \times mn}$ be a solution to Problem~\eqref{eq:LNCG_SDP_relax}. Let $Z$ be an $mn \times n$ Gaussian random matrix whose entries are drawn i.i.d.~from $\mathcal{N}(0, 1/n)$. Compute the orthogonal matrices
    \begin{equation}
        Q_v = \mathcal{P}(X_v Z),
    \end{equation}
    where $\mathcal{P}(X) = \argmin_{Y \in \Orth(n)} \| Y - X \|_F$. The expected value of this approximate solution (averaged over $Z$) obeys
    \begin{equation}
        \E\l[ f(Q_1, \ldots, Q_m) \r] \geq \alpha_{\Orth(n)}^2 \max_{R_1, \ldots, R_m \in \Orth(n)} f(R_1, \ldots, R_m).
    \end{equation}
    The approximation ratio $\alpha_{\Orth(n)}^2$ is defined by the average singular value of random Gaussian $n \times n$ matrices $Z_1 \sim \mathcal{N}(0, \I_n / n)$,
    \begin{equation}
        \alpha_{\Orth(n)} \coloneqq \E\l[ \frac{1}{n} \sum_{i \in [n]} \sigma_i(Z_1) \r],
    \end{equation}
    where $\sigma_i(Z_1)$ is the $i$th singular value of $Z_1$.
\end{theorem}

Our adaptation to the $\SO(n)$ setting simply replaces the rounding operator $\mathcal{P}$ with $\widetilde{\mathcal{P}}(X) \coloneqq \argmin_{Y \in \SO(n)} \| Y - X \|_F$. With this change we obtain an analogous result for optimizing over $\SO(n)$ elements with the same classical semidefinite program:
\begin{theorem}\label{thm:SO(n)_approx_ratio}
    Let $X_1, \ldots, X_n$ and $Z$ be as in Theorem~\ref{thm:O(n)_approx_ratio}. Compute the rotation matrices
    \begin{equation}
        Q_v = \widetilde{\mathcal{P}}(X_v Z),
    \end{equation}
    where $\widetilde{\mathcal{P}}(X) = \argmin_{Y \in \SO(n)} \| Y - X \|_F$. The expected value of this approximate solution (averaged over $Z$) obeys
    \begin{equation}
        \E\l[ f(Q_1, \ldots, Q_m) \r] \geq \alpha_{\SO(n)}^2 \max_{R_1, \ldots, R_m \in \SO(n)} f(R_1, \ldots, R_m).
    \end{equation}
    The approximation ratio $\alpha_{\SO(n)}^2$ is defined by the average $n - 1$ largest singular values of random Gaussian $n \times n$ matrices $Z_1 \sim \mathcal{N}(0, \I_n / n)$,
    \begin{equation}
        \alpha_{\SO(n)} \coloneqq \E\l[ \frac{1}{n} \sum_{i \in [n-1]} \sigma_i(Z_1) \r],
    \end{equation}
    where $\sigma_i(Z_1)$ is the $i$th singular value of $Z_1$, in descending order $\sigma_1(Z_1) \geq \cdots \geq \sigma_n(Z_1) \geq 0$.
\end{theorem}
Because singular values are nonnegative, it is clear that
\begin{equation}
    \alpha_{\Orth(n)} - \alpha_{\SO(n)} = \frac{1}{n} \E\l[ \sigma_n(Z_1) \r] \geq 0.
\end{equation}
In particular we will see that $\E\l[ \sigma_n(Z_1) \r] > 0$ for all finite $n$, so the rounding algorithm guarantees a strictly smaller approximation ratio for the problem over $\SO(n)$ than over $\Orth(n)$.

The proof of Theorem~\ref{thm:O(n)_approx_ratio} requires two lemmas regarding the expected value of random Gaussian matrices under the rounding operator $\mathcal{P}$. Analogously, our proof of Theorem~\ref{thm:SO(n)_approx_ratio} requires a modification of those lemmas when $\mathcal{P}$ is replaced by $\widetilde{\mathcal{P}}$.
\begin{lemma}[{Adapted from \cite[Lemma 5]{bandeira2016approximating}}]\label{lem:proj_arbitrary}
    Let $M, N \in \R^{n \times mn}$ obey $MM^\T = NN^\T = \I_n$. For $Z \in \R^{mn \times n}$ with i.i.d.~entries drawn from $\mathcal{N}(0, n^{-1})$, we have
    \begin{equation}
        \E\l[ \widetilde{\mathcal{P}}(MZ) (NZ)^\T \r] = \E\l[ (MZ) \widetilde{\mathcal{P}}(NZ)^\T \r] = \alpha_{\SO(n)} MN^\T.
    \end{equation}
\end{lemma}
This lemma is proved with the help of the following lemma.
\begin{lemma}[{Adapted from \cite[Lemma 6]{bandeira2016approximating}}]\label{lem:proj_gaussian}
    Let $Z_1 \in \R^{n \times n}$ with i.i.d.~entries drawn from $\mathcal{N}(0, 1/n)$. Then
    \begin{equation}
        \E\l[ \widetilde{\mathcal{P}}(Z_1) Z_1^\T \r] = \E\l[ Z_1 \widetilde{\mathcal{P}}(Z_1)^\T \r] = \alpha_{\SO(n)} \I_n.
    \end{equation}
\end{lemma}
Before we prove these two lemmas, we will use them to prove Theorem~\ref{thm:SO(n)_approx_ratio}. The proof idea here is entirely analogous to the original argument of Theorem~\ref{thm:O(n)_approx_ratio} from Ref.~\cite{bandeira2016approximating}, but with the appropriate replacements of $\mathcal{P}$ by $\widetilde{\mathcal{P}}$. Nonetheless we sketch the proof below for completeness.

\begin{proof}[Proof (of Theorem~\ref{thm:SO(n)_approx_ratio})]
We wish to lower bound the average rounded value
\begin{equation}
    \E[f(Q_1, \ldots, Q_m)] = \E\l[ \sum_{(u,v) \in E} \langle C_{uv}, \widetilde{\mathcal{P}}(X_u Z) \widetilde{\mathcal{P}}(X_v Z)^\T \rangle \r]
\end{equation}
in terms of the relaxed value $\sum_{(u,v) \in E} \langle C_{uv}, X_u X_v^\T \rangle$. Assuming we have such a lower bound with ratio $0 < \alpha^2 \leq 1$, this leads to a chain of inequalities establishing the desired approximation ratio to the original problem:
\begin{equation}\label{eq:approx_ratio_from_relaxation}
\begin{split}
    \E\l[ \sum_{(u,v) \in E} \langle C_{uv}, \widetilde{\mathcal{P}}(X_u Z) \widetilde{\mathcal{P}}(X_v Z)^\T \rangle \r] &\geq \alpha^2 \sum_{(u,v) \in E} \langle C_{uv}, X_u X_v^\T \rangle\\
    &\geq \alpha^2 \max_{R_1, \ldots, R_m \in \Orth(n)} \sum_{(u,v) \in E} \langle C_{uv}, R_u R_v^\T \rangle\\
    &\geq \alpha^2 \max_{R_1, \ldots, R_m \in \SO(n)} \sum_{(u,v) \in E} \langle C_{uv}, R_u R_v^\T \rangle,
\end{split}
\end{equation}
where the second inequality follows from the fact that the relaxation provides an upper bound to the original problem, and the third inequality is a consequence of $\SO(n) \subset \Orth(n)$. The task is then to determine such an $\alpha$ which satisfies the first inequality of Eq.~\eqref{eq:approx_ratio_from_relaxation}. The core argument is a generalization of the Rietz method~\cite{alon2004approximating}, which proceeds by constructing a positive semidefinite matrix $S \in \R^{mn \times mn}$ whose $(u, v)$th block is defined as
\begin{equation}
    S_{uv} \coloneqq \l( X_u Z - \alpha^{-1} \widetilde{\mathcal{P}}(X_u Z) \r) \l( X_v Z - \alpha^{-1} \widetilde{\mathcal{P}}(X_v Z) \r)^\T.
\end{equation}
The expected value of this matrix is
\begin{equation}
\begin{split}
    \E[S_{uv}] &= \E\l[ X_u Z (X_v Z)^\T - \alpha^{-1} \widetilde{\mathcal{P}}(X_u Z) (X_v Z)^\T - \alpha^{-1} (X_u Z) \widetilde{\mathcal{P}}(X_v Z)^\T + \alpha^{-2} \widetilde{\mathcal{P}}(X_u Z) \widetilde{\mathcal{P}}(X_v Z)^\T \r]\\
    &= X_u \E\l[ ZZ^\T \r] X_v - \alpha^{-1} \E\l[ \widetilde{\mathcal{P}}(X_u Z) (X_v Z)^\T \r] - \alpha^{-1} \E\l[ (X_u Z) \widetilde{\mathcal{P}}(X_v Z)^\T \r] + \alpha^{-2} \E\l[ \widetilde{\mathcal{P}}(X_u Z) \widetilde{\mathcal{P}}(X_v Z)^\T \r].
\end{split}
\end{equation}
Because $ZZ^\T$ is a Wishart matrix with covariance matrix $\I_n / n$, we have $\E[ZZ^\T] = \I_n$. Meanwhile, $\E\l[ \widetilde{\mathcal{P}}(X_u Z) \widetilde{\mathcal{P}}(X_v Z)^\T \r]$ is the quantity we wish to bound. To compute the expected values of the two cross terms, we invoke Lemma~\ref{lem:proj_arbitrary} which holds because $X_u X_u^\T = X_v X_v^\T = \I_n$:
\begin{equation}
    \E\l[ \widetilde{\mathcal{P}}(X_u Z) (X_v Z)^\T \r] = \E\l[ (X_u Z) \widetilde{\mathcal{P}}(X_v Z)^\T \r] = \alpha_{\SO(n)} X_u X_v^\T.
\end{equation}
Thus, setting $\alpha = \alpha_{\SO(n)}$, we obtain
\begin{equation}
\begin{split}
    \E[S_{uv}] &= X_u X_v^\T - X_u X_v^\T - X_u X_v^\T + \alpha_{\SO(n)}^{-2} \E\l[ \widetilde{\mathcal{P}}(X_u Z) \widetilde{\mathcal{P}}(X_v Z)^\T \r]\\
    &= - X_u X_v^\T + \alpha_{\SO(n)}^{-2} \E\l[ \widetilde{\mathcal{P}}(X_u Z) \widetilde{\mathcal{P}}(X_v Z)^\T \r].
\end{split}
\end{equation}
Finally, using the fact that $C, S \succeq 0$, we have that $\langle C, S \rangle \geq 0$ and so $\E\langle C, S \rangle \geq 0$, which implies that
\begin{equation}
    \E\l[ \sum_{(u,v) \in E} \langle C_{uv}, \widetilde{\mathcal{P}}(X_u Z) \widetilde{\mathcal{P}}(X_v Z)^\T \rangle \r] \geq \alpha_{\SO(n)}^2 \sum_{(u,v) \in E} \langle C_{uv}, X_u X_v^\T \rangle.
\end{equation}
Then by Eq.~\eqref{eq:approx_ratio_from_relaxation} the claim follows.
\end{proof}

We now establish the value of
\begin{equation}
    \alpha_{\SO(n)} = \E\l[ \frac{1}{n} \sum_{i \in [n-1]} \sigma_i(Z_1) \r]
\end{equation}
from Lemmas~\ref{lem:proj_arbitrary} and \ref{lem:proj_gaussian}. Because Lemma~\ref{lem:proj_arbitrary} is somewhat technical and the argument is virtually unchanged by replacing $\mathcal{P}$ with $\widetilde{\mathcal{P}}$, we refer the reader to Ref.~\cite{bandeira2016approximating} for proof details. Instead, we simply note that the only part of the proof for Lemma~\ref{lem:proj_arbitrary} which does depend on the change to $\widetilde{\mathcal{P}}$ is the final result, wherein it is established that
\begin{equation}
    \E\l[ (MZ) \widetilde{\mathcal{P}}(NZ)^\T \r] = \E\l[ Z_1 \widetilde{\mathcal{P}}(Z_1)^\T \r] MN^\T,
\end{equation}
where $Z_1 \in \R^{n \times n}$ has entries i.i.d.~from $\mathcal{N}(0, 1 / n)$.
Thus proving Lemma~\ref{lem:proj_gaussian} is the key component in establishing the value of the approximation ratio $\alpha_{\SO(n)}^2$.

\begin{proof}[Proof (of Lemma~\ref{lem:proj_gaussian})]
    Consider the singular value decomposition of $Z_1 = U \Sigma V^\T \in \R^{n \times n}$. Its special singular value decomposition can be written as $Z_1 = U (\Sigma J_{UV^\T}) (V J_{UV^\T})^\T$, where $J_{UV^\T}$ is the $n \times n$ diagonal matrix
\begin{equation}
    J_{UV^\T} \coloneqq \begin{bmatrix}
    I_{n-1} & 0\\
    0 & \det UV^\T
    \end{bmatrix}.
\end{equation}
Note that $J_{UV^\T} = J_U J_V$. Using the fact that the (special) rounding operator returns
\begin{equation}
    \widetilde{\mathcal{P}}(Z_1) = U J_{UV^\T} V^\T,
\end{equation}
we have
\begin{equation}\label{eq:PZZ}
    \widetilde{\mathcal{P}}(Z_1) Z_1^\T = U J_U J_V \Sigma U^\T.
\end{equation}
Because $Z_1$ is a random Gaussian matrix with i.i.d.~entries, its singular values and left- and right-singular vectors are distributed independently~\cite{tulino2004random}. In particular, both $U$ and $V$ are distributed according to the Haar measure on $\Orth(n)$. The expected value of Eq.~\eqref{eq:PZZ} can therefore be split into three independent averages:
\begin{equation}
    \E_{Z_1 \sim \mathcal{N}(0, \I_n / n)} \l[\widetilde{\mathcal{P}}(Z_1) Z_1^\T \r] = \E_{\Sigma \sim D} \E_{U \sim \Orth(n)} \l[ U J_U \E_{V \sim \Orth(n)} \l[ J_V \r] \Sigma U^\T \r].
\end{equation}
(We shall comment on the distribution $D$ of singular values later.) Because $\Orth(n)$ is evenly divided into its unconnected $({+1})$- and $({-1})$-determinant components, the average determinant vanishes:~$\E_{V \sim \Orth(n)}[\det V] = 0$. This leaves us with
\begin{equation}\label{eq:special_haar_avg}
    \E\l[\widetilde{\mathcal{P}}(Z_1) Z_1^\T \r] = \E\l[ U \overline{\Sigma} U^\T \r],
\end{equation}
where
\begin{equation}
    \overline{\Sigma} = \begin{bmatrix}
    \sigma_1(Z_1) & & & \\
    & \ddots & & \\
    & & \sigma_{n-1}(Z_1) & \\
    & & & 0
    \end{bmatrix}.
\end{equation}
The Haar average over $U \sim \Orth(n)$ in Eq.~\eqref{eq:special_haar_avg} is well-known~\cite{collins2006integration} to be proportional to the identity, $\E[U \overline{\Sigma} U^\T] = \lambda \I_n$, and the constant of proportionality can be determined by considering its trace:
\begin{equation}
\begin{split}
    n\lambda = \tr\l( \lambda \I_n \r) &= \tr\l( \E[ U \overline{\Sigma} U^\T ] \r)\\
    &= \E[\tr\overline{\Sigma}]\\
    &= \E\l[ \sum_{i \in [n-1]} \sigma_i(Z_1) \r].
\end{split}
\end{equation}
Hence $\lambda = \alpha_{\SO(n)}$ and so
\begin{equation}
    \E\l[\widetilde{\mathcal{P}}(Z_1) Z_1^\T \r] = \alpha_{\SO(n)} \I_n.
\end{equation}
The corresponding statement for $\E\l[ Z_1 \widetilde{\mathcal{P}}(Z_1)^\T \r]$ follows completely analogously, essentially by interchanging the roles of $U$ and $V$. The entire argument is equivalent because $U$ and $V$ are i.i.d.
\end{proof}

To numerically evaluate $\alpha_{\SO(n)}$ we can use the linearity of expectation,
\begin{equation}
    \alpha_{\SO(n)} = \alpha_{\Orth(n)} - \frac{1}{n} \E\l[ \sigma_n(Z_1) \r].
\end{equation}
The distribution of singular values of random Gaussian matrices can be analyzed from the theory of Wishart matrices. In particular, $Z_1 Z_1^\T = U \Sigma^2 U^\T$ is a Wishart matrix with covariance matrix $\I_n / n$, so the distribution of singular values $\Sigma$ is the square root of the Wishart distribution of eigenvalues. The quantity $\alpha_{\Orth(n)}$ in terms of the marginal distribution $p_n^{\text{(avg)}}(x)$ of Wishart eigenvalues $x \in (0, \infty)$ was studied in Ref.~\cite{bandeira2016approximating}, yielding the expression
\begin{equation}
    \alpha_{\Orth(n)} = \frac{1}{\sqrt{n}} \int_0^\infty p_n^{\text{(avg)}}(x) \sqrt{x} \, dx.
\end{equation}
Note the factor of $n^{-1/2}$, which is introduced because the distribution $p_n^{\text{(avg)}}(x)$ is normalized to have unit variance. An explicit expression of $p_n^{\text{(avg)}}(x)$ can be found in Refs.~\cite[Lemma 21]{bandeira2016approximating} and \cite[Eq.~(16)]{livan2011momets}.

For our newly derived approximation ratio $\alpha_{\SO(n)}$, we need to additionally evaluate the expected smallest singular value of this Wishart distribution. This minimum-eigenvalue distribution was studied in Ref.~\cite{edelman1988eigenvalues}, wherein an analytical expression was derived (again assuming unit variance):
\begin{equation}
    p_n^{\text{(min)}}(x) = \frac{n}{2^{n - 1/2}} \frac{\Gamma(n)}{\Gamma(n/2)} \frac{e^{-xn/2}}{\sqrt{x}} U\l( \frac{n-1}{2}, -\frac{1}{2}, \frac{x}{2} \r).
\end{equation}
Here, $U(a, b, z)$ with $a > 0$ and $b < 1$ is the Tricomi confluent hypergeometric function, the unique solution to the differential equation
\begin{equation}
    z \frac{d^2 U}{dz^2} + (b - z) \frac{dU}{dz} - aU = 0
\end{equation}
with boundary conditions $U(a, b, 0) = \Gamma(1-b)/\Gamma(1+a-b)$ and $\lim_{z \to \infty} U(a, b, z) = 0$. The expression for the average smallest singular value is therefore
\begin{equation}
    \E\l[ \sigma_n(Z_1) \r] = \frac{1}{\sqrt{n}} \int_0^\infty p_n^{\text{(min)}}(x) \sqrt{x} \, dx.
\end{equation}
Altogether, we arrive at the integral expression for
\begin{equation}
    \alpha_{\SO(n)} = \frac{1}{\sqrt{n}} \int_0^\infty \l[ p_n^{\text{(avg)}}(x) - \frac{1}{n} p_n^{\text{(min)}}(x) \r] \sqrt{x} \, dx.
\end{equation}

\section{\label{sec:symmetry_vertex}LNCG Hamiltonian symmetries}
Here we demonstrate the local $\Orth(n)$ symmetry discussed in Section~\ref{sec:vertex_rounding}. Consider an edge term
\begin{equation}
    H_{uv} = \sum_{i,j \in [n]} [C_{uv}]_{ij} \sum_{k \in [n]} P_{ik}^{(u)} \otimes P_{jk}^{(v)}.
\end{equation}
Because $P_{ik} = \i \widetilde{\gamma}_i \gamma_k$ and the sum over $k$ is independent of $C_{uv}$, we can factor out each $\gamma_k^{(u)} \otimes \gamma_k^{(v)}$ and rewrite the Hamiltonian term as as
\begin{equation}
    H_{uv} = -\sum_{i,j \in [n]} [C_{uv}]_{ij} \l( \widetilde{\gamma}_i^{(u)} \otimes \widetilde{\gamma}_j^{(v)} \r) \l( \sum_{k \in [n]} \gamma_k^{(u)} \otimes \gamma_k^{(v)} \r).
\end{equation}
The operator $\sum_{k \in [n]} \gamma_k^{(u)} \otimes \gamma_k^{(v)}$ is invariant to any orthogonal transformation $V \in \Orth(n)$ which acts identically on both vertices:
\begin{equation}
\begin{split}
    \mathcal{U}_{(\I_n, V)}^{\otimes 2} \l( \sum_{k \in [n]} \gamma_k^{(u)} \otimes \gamma_k^{(v)} \r) (\mathcal{U}_{(\I_n, V)}^{\otimes 2})^\dagger &= \sum_{k \in [n]} \l( \sum_{\ell \in [n]} V_{k\ell} \gamma_\ell^{(u)} \r) \otimes \l( \sum_{\ell' \in [n]} V_{k\ell'} \gamma_{\ell'}^{(v)} \r)\\
    &= \sum_{\ell,\ell' \in [n]} \l( \sum_{k \in [n]} [V^\T]_{\ell k} V_{k\ell'} \r) \gamma_\ell^{(u)} \otimes \gamma_{\ell'}^{(v)}\\
    &= \sum_{\ell,\ell' \in [n]} \delta_{\ell\ell'} \gamma_\ell^{(u)} \otimes \gamma_{\ell'}^{(v)}\\
    &= \sum_{\ell \in [n]} \gamma_\ell^{(u)} \otimes \gamma_\ell^{(v)}.
\end{split}
\end{equation}
Because $\mathcal{U}_{(\I_n, V)}$ acts trivially on all $\widetilde{\gamma}_i$, it follows that
\begin{equation}
    \mathcal{U}_{(\I_n, V)}^{\otimes 2} H_{uv} (\mathcal{U}_{(\I_n, V)}^{\otimes 2})^\dagger = H_{uv}
\end{equation}
for each $(u, v)$. Finally, this symmetry can be straightforwardly extended to all $m$ vertices:
\begin{equation}\label{eq:O(n)_sym}
    \mathcal{U}_{(\I_n, V)}^{\otimes m} H (\mathcal{U}_{(\I_n, V)}^{\otimes m})^\dagger = H.
\end{equation}

Now we investigate some consequences of this continuous symmetry. The following lemma is particularly important, as it necessitates the use of the one-body perturbation $\zeta H_1$ to break this symmetry when preparing of eigenstates of $H$.

\begin{lemma}
Let $\ket{\psi}$ be a nondegenerate eigenstate of $H$. Then for each single-vertex marginal $\sigma_v \coloneqq \tr_{\neg v} \op{\psi}{\psi}$, $v \in [m]$, we have
\begin{equation}
    Q(\sigma_v) = 0.
\end{equation}
\end{lemma}

\begin{proof}
Consider the expansion of its density matrix $\op{\psi}{\psi}$ in the Majorana operator basis, up to the relevant one-body expectation values:
\begin{equation}
    \op{\psi}{\psi} = \frac{1}{d^m} \l( \openone^{\otimes m} + \sum_{v \in [m]} \sum_{i,j \in [n]} [Q(\sigma_v)]_{ij} \i \widetilde{\gamma}_i^{(v)} \gamma_j^{(v)} + \cdots \r),
\end{equation}
where we recall that $[Q(\sigma_v)]_{ij} = \ev{\i \widetilde{\gamma}_i^{(v)} \gamma_j^{(v)}}{\psi}$. Due to the symmetry [Eq.~\eqref{eq:O(n)_sym}], for every $V \in \Orth(n)$ the state $\ket{\psi(V)} = {\mathcal{U}}_{(\I_n, V)}^{\otimes m} \ket{\psi}$ is also an eigenvector of $H$ with the same eigenvalue. The one-body expectation values of $\ket{\psi}$ are therefore transformed as
\begin{equation}
\begin{split}
    {\mathcal{U}}_{(\I_n, V)}^{\otimes m} \l( \sum_{v \in [m]} \sum_{i,j \in [n]} [Q(\sigma_v)]_{ij} \i \widetilde{\gamma}_i^{(v)} \gamma_j^{(v)} \r) ({\mathcal{U}}_{(\I_n, V)}^{\otimes m})^\dagger &= \sum_{v \in [m]} \sum_{i,j \in [n]} [Q(\sigma_v)]_{ij} \i \sum_{i' \in [n]} V_{ii'} \widetilde{\gamma}_{i'}^{(v)} \gamma_{j}^{(v)}\\
    &= \sum_{v \in [m]} \sum_{i',j \in [n]} [V^\T Q(\sigma_v)]_{i'j} \i \widetilde{\gamma}_{i'}^{(v)} \gamma_{j}^{(v)}.
\end{split}
\end{equation}
Now suppose that $\ket{\psi}$ is nondegenerate. Then we have that $\op{\psi}{\psi} = \op{\psi(V)}{\psi(V)}$ for all $V \in \Orth(n)$, and in particular we can take the Haar integral over $\Orth(n)$ of this identity:
\begin{equation}
\begin{split}
    \int_{\Orth(n)} d\mu(V) \op{\psi(V)}{\psi(V)} = \int_{\Orth(n)} d\mu(V) \op{\psi}{\psi} = \op{\psi}{\psi},
\end{split}
\end{equation}
where $\mu$ is the normalized Haar measure satisfying $\mu(\Orth(n)) = 1$. Because the Haar integral over linear functions vanishes, i.e., $\int_{\Orth(n)} d\mu(V) \, V_{ij} = 0$~\cite{collins2006integration}, it follows that
\begin{equation}
    \int_{\Orth(n)} d\mu(V) \, V^\T Q(\sigma_v) = 0.
\end{equation}
Furthermore, because $\i \widetilde{\gamma}_{i}^{(v)} \gamma_{j}^{(v)}$ are linearly independent (as elements of an operator basis), the equality $\op{\psi}{\psi} = \int_{\Orth(n)} d\mu(V) \op{\psi(V)}{\psi(V)}$ implies that
\begin{equation}
    Q(\sigma_v) = \int_{\Orth(n)} d\mu(V) \, V^\T Q(\sigma_v) = 0
\end{equation}
for all $v \in [m]$.
\end{proof}

\section{\label{sec:pin_circuits}The Pin group from quantum circuits}
In the main text we showed that each $x \in \Pin(n)$ corresponds to the eigenstates of a family of free-fermion Hamiltonians, which are (pure) fermionic Gaussian states. Here we provide an alternative perspective of this correspondence, using quantum circuits which prepare such states.

Recall that every $x \in \Pin(n)$ can be written as
\begin{equation}
    x = u_1 \cdots u_k
\end{equation}
for some $k \leq n$, where we may expand each $u_j \in S^{n-1}$ in the standard basis as
\begin{equation}
    u_j = \sum_{i \in [n]} v_i^{(j)} e_i
\end{equation}
for some unit vector $v^{(j)} \in \R^n$. The product of these unit vectors can be expressed using the right-multiplication operator $\rho_{u_j}$ acting on the identity element,
\begin{equation}
\begin{split}
    x &= e_\varnothing x\\
    &= e_\varnothing u_1 \cdots u_k\\
    &= (\rho_{u_k} \cdots \rho_{u_1})(e_\varnothing)\\
\end{split}
\end{equation}
On the other hand, consider the so-called Clifford loader~\cite{kerenidis2022quantum}, a circuit primitive defined (in our notation) as
\begin{equation}
    \Gamma(v) = \sum_{i \in [n]} v_i \gamma_i
\end{equation}
for any unit vector $v \in \R^n$. It is straightforward to check that this operator is Hermitian and unitary, and Ref.~\cite{kerenidis2022quantum} provides an explicit circuit constructions based on two-qubit Givens rotation primitives.\footnote{Givens rotations themselves are representations of fermionic Gaussian transformations acting on two modes at a time.} Using the relation $\gamma_i = \rho_i \alpha$ and acting this circuit on the vacuum state $\ket{0^n} \equiv \ket{e_\varnothing}$, we see that the state
\begin{equation}
\begin{split}
    \ket{x} &= \Gamma(v^{(k)}) \cdots \Gamma(v^{(1)}) \ket{0^n}\\
    &= (\rho_{u_k} \alpha \cdots \rho_{u_1} \alpha)\ket{e_\varnothing}\\
    &= (-1)^{\binom{k}{2}} (\rho_{u_k} \cdots \rho_{u_1})\ket{e_\varnothing}
\end{split}
\end{equation}
indeed is equivalent to $x \in \Pin(n)$, up to a global sign (recall that $\alpha^2 = \openone$ and $\alpha\ket{e_\varnothing} = \ket{e_\varnothing}$). In other words, the Clifford loader is precisely the quantum-circuit representation of generators of the $\Pin$ group.

It is worth noting that in Ref.~\cite{kerenidis2022quantum} they construct ``subspace states'' from this composition of Clifford loaders. In the language of fermions, subspace states are Slater determinants:~free-fermion states with fixed particle number. Preparing Slater determinants in this fashion requires that the unit vectors $v^{(1)}, \ldots, v^{(k)}$ be linearly independent (and thus, without loss of generality, they can be made orthonormal while preserving the subspace that they span, hence the alternative name). However, the definition of the Pin group demands all possible unit vectors in such products, not just those which are linearly independent. Indeed, one can see that if the state $\ket{x}$ is a Slater determinant, then its trace is an integer, as
\begin{equation}
    \tr[Q(x)] = \langle x | \sum_{i \in [n]} \i \widetilde{\gamma}_i \gamma_i | x \rangle = \langle x | (n \I_{2^n} - 2 N) | x \rangle = n - 2k \in \{-n, \ldots, n\},
\end{equation}
where $N = \sum_{i \in [n]} a_i^\dagger a_i$ is the total number operator. Clearly not all orthogonal matrices have integer trace, so Slater determinants are insufficient to cover all of $\Pin(n)$. To reach the remaining elements, we note that if the unit vectors are linearly dependent, then one can show that the $\ket{x} = \Gamma(v^{(k)}) \cdots \Gamma(v^{(1)}) \ket{0^n}$ does not have fixed particle number, so $\ev{N}{x}$ is not necessarily an integer.

\section{\label{sec:tomography}Measurement schemes}
In this section we comment on the efficient schemes available for measuring the relevant expectation values. This is important even in the context of a phase-estimation approach, as one needs to obtain the values of the decision variables to perform the rounding procedure.

\subsection{Tomography of edge marginals}

To measure the energy (for variational approaches) or to perform edge rounding, we require the expectation values of the two-body observables
\begin{align}
    \Gamma_{ij}^{(u, v)} &= \sum_{k \in [n]} P_{ik}^{(u)} \otimes P_{jk}^{(v)}, \quad G = \Orth(n),\\
    \widetilde{\Gamma}_{ij}^{(u, v)} &= \sum_{k \in [n]} \widetilde{P}_{ik}^{(u)} \otimes \widetilde{P}_{jk}^{(v)}, \quad G = \SO(n),
\end{align}
for each $(u, v) \in E$ and $i, j \in [n]$. When considering $G = \Orth(n)$, because each $P_{ik}^{(u)} \otimes P_{jk}^{(v)}$ is a fermionic two-body operator, we can straightforwardly apply the partial tomography schemes developed for local fermionic systems, such as Majorana swap networks~\cite{bonet2020nearly} or classical shadows~\cite{zhao2021fermionic}. In either case, the measurement circuits required are fermionic Gaussian unitaries and the sample complexity is $\Ord(N^2/\epsilon^2)$, where $N = n|V|$ is the total number of qubits and $\epsilon > 0$ is the desired estimation precision of each expectation value.

\subsection{Tomography of vertex marginals}

The vertex-rounding procedure requires the expectation values of only single-qudit observables $P_{ij}^{(v)}$ or $\widetilde{P}_{ij}^{(v)}$ on each vertex $v \in V$, $i, j \in [n]$. In this case the observables being measured commute across vertices, so it suffices to talk about the tomography of a single vertex, as the same process can be executed in parallel across all vertices. Again, because these operators are fermionic one-body observables, the same fermionic partial tomography technology~\cite{bonet2020nearly,zhao2021fermionic} can be applied here, incurring a sampling cost of $\Ord(n/\epsilon^2)$. In fact, further constant-factor savings can be achieved in the one-body setting by using the measurement scheme introduced in Ref.~\cite{arute2020hartree}. This scheme requires only particle-conserving fermionic Gaussian unitaries, which can be compiled with only half the depth of the more general Gaussian unitaries required of the previous two methods. Note that each operator $P_{ij}$ is of the form of either $XX$ or $YY$ when $|i - j| = 1$, and so they correspond precisely to the observables measured to reconstruct the real part of the fermionic one-body reduced density matrix~\cite{arute2020hartree}.

\subsection{Estimating observables via gradient method}

Ref.~\cite{huggins2021nearly} introduces a quantum algorithm for estimating a large collection of (generically noncommuting) $M$ observables $\{O_j \mid j \in [M]\}$ to precision $\epsilon$ by encoding their expectation values into the gradient of a function. This function is implemented as a quantum circuit which prepares the state of interest and applies $\Ot(\sqrt{M}/\epsilon)$ gates of the form $c\text{-}e^{-\i\theta O_j}$, controlled on $\Ord(M\log(1/\epsilon))$ ancilla qubits. Finally, using the algorithm of Ref.~\cite{gilyen2019optimizing} for gradient estimation, one calls this circuit $\Ot(\sqrt{M}/\epsilon)$ times to estimate the encoded expectation values (the notation $\Ot(\cdot)$ suppresses polylogarithmic factors). Although this approach demands additional qubits and more complicated circuitry, it has the striking advantage of a quadratically improved scaling in the number of state preparations with respect to estimation error $\epsilon$, compared to the refinement of sampling error in tomographic approaches. In our context, we have either $M = n^3|E|$ or $M = n^2|V|$ observables of interest (satisfying a technical requirement of having their spectral norms bounded by $1$), corresponding to the measurement of edge or vertex terms respectively. The gates required are then simply controlled Pauli rotations.


\end{document}